\documentclass[10pt,english]{article}
\usepackage{lmodern}

\usepackage[T1]{fontenc}
\usepackage[latin9]{inputenc}
\usepackage{geometry}
\usepackage{amsmath}
\usepackage{amsthm}
\usepackage{amssymb}
\allowdisplaybreaks[1]

\makeatletter
\numberwithin{equation}{section}
\numberwithin{figure}{section}
\theoremstyle{plain}
\newtheorem{thm}{\protect\theoremname}
\theoremstyle{definition}
\newtheorem{defn}[thm]{\protect\definitionname}
\theoremstyle{plain}
\newtheorem*{thm*}{\protect\theoremname}
\theoremstyle{plain}
\newtheorem{lem}[thm]{\protect\lemmaname}
\newtheorem*{lem*}{\protect\lemmaname}
\theoremstyle{plain}
\newtheorem{prop}[thm]{\protect\propositionname}
\theoremstyle{plain}

\newtheorem{cor}[thm]{\protect\corollaryname}

\usepackage{hyperref}
\hypersetup{
    colorlinks,
    allcolors=magenta
}
\usepackage[lined,boxed,ruled,norelsize,algo2e]{algorithm2e}
\@addtoreset{section}{part}
\usepackage{graphicx}
\usepackage{mathrsfs}
\usepackage{ragged2e}
\usepackage{wrapfig}
\usepackage{thmtools}
\usepackage{enumitem}
\usepackage{mathtools}
\usepackage{microtype}
\usepackage{booktabs}       
\usepackage{amsfonts}       
\usepackage{nicefrac}       
\usepackage{thm-restate}

\usepackage{tikz}
\usetikzlibrary{calc}
\usetikzlibrary{positioning, arrows.meta, decorations.pathreplacing}

\makeatother

\usepackage{babel}
\providecommand{\corollaryname}{Corollary}
\providecommand{\definitionname}{Definition}
\providecommand{\lemmaname}{Lemma}
\providecommand{\propositionname}{Proposition}
\providecommand{\theoremname}{Theorem}

\begin{document}
\global\long\def\defeq{\stackrel{\mathrm{{\scriptscriptstyle def}}}{=}}%

\global\long\def\norm#1{\left\Vert #1\right\Vert }%
\global\long\def\prob#1{\mathbf{P}\left(#1\right)}%
\global\long\def\Par#1{\left(#1\right)}%
\global\long\def\Brack#1{\left[#1\right]}%
\global\long\def\Abs#1{\left|#1\right|}%
\global\long\def\ginverse{g^{-1}}%
\global\long\def\ginv#1{g(#1)^{-1}}%
\global\long\def\ginvnorm#1{\left\Vert #1\right\Vert _{g^{-1}}}%
\global\long\def\Brace#1{\left\{  #1\right\}  }%
\global\long\def\inner#1{\left\langle #1\right\rangle }%
\global\long\def\normg#1{\left\Vert #1\right\Vert _{g}}%
\global\long\def\normginv#1{\left\Vert #1\right\Vert _{g^{-1}}}%

\global\long\def\R{\mathbb{R}}%
\global\long\def\Rn{\mathbb{R}^{n}}%
\global\long\def\E{\mathbb{E}}%
\global\long\def\P{\mathbf{P}}%
\global\long\def\S{\mathbb{S}}%
\global\long\def\Rd{\mathbb{R}^{d}}%
\global\long\def\P{\mathbf{P}}%

\global\long\def\acal{\mathcal{A}}%
\global\long\def\bcal{\mathcal{B}}%
\global\long\def\ccal{\mathcal{C}}%
\global\long\def\dcal{\mathcal{D}}%
\global\long\def\ecal{\mathcal{E}}%
\global\long\def\tcal{\mathbb{\mathcal{T}}}%
\global\long\def\mcal{\mathbb{\mathcal{M}}}%
\global\long\def\pcal{\mathcal{P}}%
\global\long\def\ncal{\mathcal{N}}%
\global\long\def\kcal{\mathcal{K}}%
\global\long\def\half{\frac{1}{2}}%
 
\global\long\def\tr{\mathrm{Tr}}%
\global\long\def\diag{\mathrm{diag}}%
\global\long\def\cov{\mathrm{Cov}}%
\global\long\def\Var{\mathrm{Var}}%
\global\long\def\rank{\mathrm{rank}}%
\global\long\def\range{\mathrm{Range}}%
\global\long\def\nulls{\mathrm{Null}}%
\global\long\def\vgood{V_{\text{good}}^{x}}%
\global\long\def\vbad{V_{\text{bad}}^{x}}%
\global\long\def\Diag{\mathrm{Diag}}%

\global\long\def\ov{\overline{v}}%
\global\long\def\ox{\overline{x}}%
\global\long\def\xmid{x_{\text{mid}}}%
\global\long\def\vmid{v_{\text{mid}}}%
\global\long\def\tx{\widetilde{x}}%
\global\long\def\tv{\widetilde{v}}%
\global\long\def\mid{\mathrm{mid}}%
\global\long\def\idftn{\mathrm{id}}%

\global\long\def\dtv{d_{\text{TV}}}%
\global\long\def\tp{\tilde{p}}%
\global\long\def\tP{\tilde{P}}%
\global\long\def\dv{\delta_{v}}%
\global\long\def\dx{\delta_{x}}%
\global\long\def\del{\mathcal{\partial}}%
\global\long\def\tpose{^{\top}}%
\global\long\def\grad{\nabla}%
\global\long\def\hess{\nabla^{2}}%
\global\long\def\veps{\varepsilon}%
\global\long\def\lda{\lambda}%
\global\long\def\vphi{\varphi}%

\global\long\def\ot{\frac{1}{3}}%
\global\long\def\tt{\frac{2}{3}}%
\global\long\def\ham{\mathrm{Ham}}%
\global\long\def\bg{\bar{g}}%
\global\long\def\md{\mathrm{mid}}%
\global\long\def\bx{\bar{x}}%
\global\long\def\bv{\bar{v}}%
\global\long\def\bp{\overline{p}}%
\global\long\def\bP{\overline{P}}%
\global\long\def\oprop{\overline{\pcal}}%

\global\long\def\Tx{T_{x}}%
\global\long\def\bTx{\overline{T}_{x}}%
\global\long\def\onecirc{\textcircled{\small{1}}}%
\global\long\def\twocirc{\textcircled{\small{2}}}%
\global\long\def\threecirc{\textcircled{\small{3}}}%
\global\long\def\fourcirc{\textcircled{\small{4}}}%

\title{Condition-number-independent convergence rate of Riemannian Hamiltonian
Monte Carlo with numerical integrators}
\author{Yunbum Kook\thanks{Georgia Tech, yb.kook@gatech.edu}, Yin Tat Lee\thanks{University of Washington and Microsoft Research, yintat@uw.edu},
Ruoqi Shen\thanks{University of Washington, shenr3@cs.washington.edu},
Santosh S. Vempala\thanks{Georgia Tech, vempala@gatech.edu}}
\maketitle
\begin{abstract}
We study the convergence rate of discretized Riemannian Hamiltonian Monte Carlo on sampling from distributions in the form of $e^{-f(x)}$ on a convex body $\mathcal{M}\subset\R^{n}$. 
We show that for distributions in the form of $e^{-\alpha^{\top}x}$ on a polytope with $m$ constraints, the convergence rate of a family of commonly-used integrators is independent of $\left\Vert \alpha\right\Vert _{2}$ and the geometry of the polytope.
In particular, the implicit midpoint method (IMM) and the generalized Leapfrog method (LM) have a mixing time of $\widetilde{O}\left(mn^{3}\right)$ to achieve $\epsilon$ total variation distance to the target distribution.
These guarantees are based on a general bound on the convergence rate for densities of the form $e^{-f(x)}$ in terms of parameters of the manifold and the integrator. 
Our theoretical guarantee complements the empirical results of \cite{kook2022sampling}, which shows that RHMC with IMM can sample ill-conditioned, non-smooth and constrained distributions in very high dimension efficiently in practice. 
\end{abstract}

\tableofcontents{}

\newpage{}


\section{Introduction}

Efficient sampling from high dimensional distributions is a fundamental question that arises in many fields such as statistics, machine learning, and theoretical computer science. One class of distributions that arises in many applications is constrained distributions, where the distribution is defined on a constrained set. Sampling from such distribution can be an efficient way to study the geometric properties of the constrained set when direct computation is not feasible. For instance, in systems biology, a metabolic network is defined by a set of equalities and inequalities that represents feasible steady state reaction rates \cite{lewis2012constraining, thiele2013community}. For large metabolic networks, sampling from the constraint set can be an efficient way to simulate the biochemical network and evaluate its capacity. In mathematics, computing the volume of the Birkhoff polytope plays a key role in several areas, including algebraic geometry, and probability. However, computing the volume exactly using algebraic representations can take years even for a small dimension $n = 11$. On the other hand, a sampling based-algorithm can compute the volume efficiently up to dimension half-million \cite{kook2022sampling}. 

\paragraph{Traditional samplers}The current primary approach for sampling is Markov Chain Monte Carlo (MCMC) method, which for many problems is the only known method with provable efficiency guarantees. For general non-smooth distributions, the traditional class of samplers is the zeroth-order samplers, which query the density of the distributions to determine the algorithm's trajectory. This class of samplers includes Ball walk \cite{lovasz1993random, kannan1997random}, its affine-invariant version Dikin walk \cite{kannan2012random, laddha2020strong} and Hit-and-Run \cite{smith1984efficient,lovasz1999hit}, which avoids an explicit step size. However, this class of sampler is inefficient in practice because it intrinsically needs a step size smaller than $O(1/\sqrt n )$, where $n$ is the dimension, to avoid stepping outside the constraint set, which leads to a bottleneck of quadratic mixing time in dimension. Moreover, without putting the convex body into an isotropic position, which requires expensive computation in practice, the mixing time of Ball walk and Hit-and-Run, $\widetilde{O} (n^2 R^2)$, depends on the condition number $R$ of the convex body. The condition number of the distributions appearing in practical applications can be large, \textit{e.g.,} the condition number of RECON1 \cite{king2016bigg}, a human
metabolic network, can be as large as $10^6$. Using Hit-and-Run to sample from metabolic networks can be over 100 times slower than the better algorithms on this problem \cite{cousins2016practical}. Sampling from the Birkhoff polytope can be prohibitively expensive for any dimension higher than $n = 20$ \cite{cousins2016practical}.

Another class of samplers commonly used is the first-order samplers, which update the Markov chain based on the gradient information. The mixing time of the continuous processes as well as the various discretization methods of this class of samplers has been studied in a long line of recent works. The most well-studied first-order samplers include Langevin algorithm  \cite{dalalyan2017further,dwivedi2018log,durmus2019analysis,vempala2019rapid,chewi2021optimal,chewi2021analysis}, its variant Underdamped Langevin algorithm \cite{cheng2018underdamped,shen2019randomized}, and Hamiltonian Monte Carlo (HMC) \cite{chen2019optimal,chen2020fast,lee2020logsmooth}. The mixing time of this class of samplers also suffers from dependence on the condition number of the distributions. Moreover, this class of samplers cannot be applied to constrained distributions directly because their Markov chain can easily step outside the constraint set. Currently, popular sampling packages such as Stan \cite{stan} and Pyro \cite{bingham2019pyro} that are based on this class of samplers are not able to handle constrained distributions, despite their effectiveness in other settings.

\paragraph{Non-Euclidean Samplers} 
Given the limitations of the traditional samplers, researchers have sought to extend these methods to non-Euclidean samplers, which leverage the local geometry of distributions to speed up the samplers. For instance, Riemannian Hamiltonian Monte Carlo (RHMC) extends the traditional HMC by considering the dynamics on a Riemannian manifold that uses a non-Euclidean metric corresponding to the distribution's local geometry. 
When combined with a local metric induced by the Hessian of a self-concordant barrier function, RHMC can sample from ill-conditioned and non-smooth distributions efficiently. A recent work \cite{kook2022sampling} showed that RHMC can achieve a 1000-fold acceleration on the benchmark dataset RECON3D \cite{king2016bigg}, the largest published human metabolic network, compared to previous methods. While RHMC has demonstrated superior practical performance, the convergence rate of discretized RHMC remains open. \cite{lee2018convergence} bounded the convergence rate of continuous RHMC in terms of the isoperimetry and natural smoothness parameters of the associated Riemannian manifold. However, to implement RHMC, sophisticated integrators such as implicit midpoint integrator (IMM) or the generalized Leapfrog integrator (LM) are necessary to maintain measure-preservation and time reversibility. Simple integrators, such as the naive Leapfrog method, are not suitable for RHMC as they are no longer symplectic on general Riemannian manifolds \cite{cobb2019introducing}. These sophisticated integrators provide accurate discretization and efficient convergence in practice, but their theoretical analysis is challenging. In particular, there is no theoretical guarantee that the convergence rate of RHMC remains independent of the condition number after discretization, which is the main motivation for using non-Euclidean samplers in our case. 

In fact, analyzing discretized non-Euclidean samplers has been a persistent challenge in many recent works. Another commonly studied class of non-Euclidean samplers is the Riemannian Langevin algorithm (RLA) \cite{girolami2011riemann}, which extends the Langevin algorithm to non-Euclidean space. A closely related process is the Mirror Langevin diffusion (MLD) \cite{zhang2020wasserstein}, which is a special case of RLA when the metric is given by the Hessian of a Legendre-type convex potential $\phi$. Many recent works have focused on obtaining the convergence rate of discretized MLD or RLA, but many of them require strong assumptions or oracles for accurate discretization. The analysis of \cite{zhang2020wasserstein, jiang2021mirror, li2022mirror} and the empirical results in \cite{jiang2021mirror} suggest that unless a strong regularity assumption between the target distribution and $\phi$ is satisfied, the naive integrators can lead to a bias term that exists even when the step size tends to zero. This bias arises from the third-order error terms resulting from non-Euclidean geometries and is hard to control. \cite{ahn2021efficient} circumvented this issue by proposing an alternative discretization method that uses the exact solution to the Brownian motion term, but it remains unclear whether the discretization is feasible for general $\phi$. Similarly, \cite{gatmiry2022convergence} analyzed the convergence rate of RLA using an oracle to sample from the natural Brownian motion on the manifold. Given the current limitations in our understanding of the integrators for non-Euclidean samplers, we believe it is crucial to investigate the integrators more thoroughly and explore alternative integrators.

\paragraph{Contribution} 
We provide (to our knowledge) the first convergence rate of discretized RHMC on a class of numerical integrators. We consider a general class of constrained distributions that can be written as
\begin{equation}
	e^{-f(x)}\text{ subject to }x\in\mathcal{M},\label{eq:constrained_dist}
\end{equation}
where we assume $f$ is a convex function and $\mathcal{M}\subset\R^{n}$
is a convex body with a (highly) self-concordant barrier. We give theoretical guarantees showing
that a large class of integrators can maintain smoothness and condition number independence when sampling from distributions in the form of $e^{-\alpha^{\top}x}$ on a polytope with $m$ constraints. In fact, many applications can be written in this form because 
any log-concave density in the form of $\eqref{eq:constrained_dist}$ can be reduced to 
\begin{equation}
	e^{-t} \text{ subject to } (x,t)\in\mathcal{M}', 
\end{equation}
where $\mathcal{M}'=\left\{ (x,t):f(x)\leq t,x\in\mathcal{M}\right\}$ is convex in $(x, t)$.
We show for distributions in the form of $e^{-\alpha^{\top}x}$, the implicit midpoint method (IMM) and the generalized Leapfrog method (LM) have a mixing time of $\widetilde{O}\left(mn^{3}\right)$
to achieve $\epsilon$ total variation distance to the target distribution.
In addition, we give a general convergence result on sampling from distributions in the form of $e^{-f(x)}$ on a convex body in terms of parameters of the manifold and the integrator, which can be useful for future works that analyze the convergence rate on other integrators or distributions. 

While numerical integration is a rich and active field \cite{hairer2006geometric}, and the study of the local convergence of numerical estimators is quite sophisticated, we are not aware of global polynomial-time mixing time guarantees based on commonly-used numerical integrators such as IMM and LM. Our convergence result is the theoretical foundation of \cite{kook2022sampling}
and extends \cite{lee2018convergence} to settings of practical importance. Our results apply to not only IMM and LM, but also a more general class of symplectic and time-reversible integrators that satisfies a sensitivity condition, which advances our understanding of integrators for RHMC and the more general non-Euclidean samplers.

Moreover, in our algorithm, we use a Metropolis filter to correct the distribution, which is a crucial step for high-accuracy sampling. To address the discretization issues of RLA and MLD, applying a Metropolis filter to correct the bias is one potential solution. Nevertheless, to the best of our knowledge, there is no general-purpose metropolized non-Euclidean Langevin algorithm in the literature. We believe that our analysis of metropolized RHMC can provide valuable insights into the design and analysis of future metropolized non-Euclidean Langevin algorithms.

It is important for readers to be aware that although the convergence rate we obtain is independent of the condition number, the convergence rate is likely to be far from optimal due to the complicated analysis of the integrators used. To couple the discretized and ideal RHMC in our analysis, we need a step size much smaller than what is typically required in practice. \cite{kook2022sampling} demonstrated that RHMC with IMM can achieve sublinear mixing times in dimension on metabolic networks and structured polytopes including hypercubes, simplices, and Birkhoff polytopes. We believe a tighter convergence bound is possible with more advanced analysis. 

\subsection{Prior work}

The convergence rate of MCMC methods in sampling from a convex body has been a topic of active research for decades (see \cite{lee2022manifold} for a more detailed discussion). 
The mixing time of ball walk on isotropic log-concave density is bounded by $\widetilde{O}(n^2)$ from a warm start \cite{kannan1997random}, where a convex body can be put into a near isotropic position in $\widetilde{O}(n^3)$ membership queries \cite{jia2021reducing}. Dikin walk uses the local geometry to improve the mixing rate to $O(mn)$ on polytopes, where $m$ is the number of constraints. Moreover, due to its affine invariance, there is no need to put the polytope into an isotropic position. With an LS barrier \cite{lee2014path}, Dikin walk can achieve a mixing rate of $\widetilde{O}(n^2)$ for any polytope \cite{laddha2020strong}. Geodesic walk utilizes non-Euclidean geometry by taking a random walk on a manifold. Geodesic walk with an exact exponential map and a Metropolis filter can converge to the uniform density in $O(mn^{3/4})$ steps \cite{lee2017geodesic}.  Continuous RHMC avoids the use of a Metropolis filter due to its measure preservation and time reversibility, which further improves the mixing time to $O(mn^{2/3})$ \cite{lee2018convergence} on uniform density. Our paper extends the mixing time result to discretized RHMC with feasible integrators on more general distributions. Note that even an extension to distribution $e^{-\alpha^\top x}$ needs nontrivial work to avoid dependence on quantities such
as the domain diameter. 

\section{RHMC with numerical integrators}

\subsection{Basics of RHMC}
Hamiltonian Monte Carlo (HMC) is one of the most widely used MCMC methods and is the default sampler implementation in many sampling packages (\cite{stan,salvatier2016probabilistic,bingham2019pyro,kook2022sampling}).
HMC introduces an auxiliary velocity variable $v$ in addition to the position $x$, defines a joint density on $(x,v)$, and determines its trajectory according to the \emph{Hamiltonian dynamics}.
The Hamiltonian dynamics is characterized by the \emph{Hamiltonian equations}, the first-order differential equations of the \emph{Hamiltonian} $H$ with respect to $x$ and $v$.
The Hamiltonian has a natural interpretation as the total energy of a particle consisting of the kinetic and potential energy at position $x$ with velocity $v$. 

The dynamic can be naturally generalized to the setting of Riemannian manifold with local metric $\{g(x)\}_{x\in \mcal}$.
A natural extension of the Hamiltonian is given by 
\[
H(x,v)=f(x)+\half v^{\top}\ginv xv+\half\log\det g(x),
\]
with $g(x)$ viewed as a positive-definite matrix. 
For later use, we split $H$ into two parts $H_{1}(x,v)=f(x)+\half\log\det g(x)$
and $H_{2}(x,v)=\half v^{\top}\ginv xv$. A curve $(x(t),v(t))\in\mcal\times T_{x}\mcal\subset\Rn\times\Rn$ is called the \emph{Hamiltonian curve} if it is the solution to the Hamiltonian equations:
\begin{align}
\frac{dx}{dt} & =\frac{\del H}{\del v}(x,v)=\ginv xv,\nonumber \\
\frac{dv}{dt} & =-\frac{\del H}{\del x}(x,v)=-\Par{\underbrace{\grad f(x)+\half\tr\Brack{g(x)^{-1}Dg(x)}}_{\frac{\del H_{1}}{\del x}}+\underbrace{\Par{-\half Dg(x)\Brack{\frac{dx}{dt},\frac{dx}{dt}}}}_{\frac{\del H_{2}}{\del x}}}.\label{eq:hmc_intro}
\end{align}
When clear from context, the Hamiltonian curve refers to $x(t)\in\mcal$
only. The Hamiltonian curves $(x(t), v(t))$ have several geometric properties. For a map $F_t: (x,v) \mapsto (x(t),v(t))$,
\begin{enumerate}
\item Hamiltonian preservation: $\frac{d}{dt}H(x(t),v(t))=0.$
\item Symplectic: $DF_{t}(x,v)^{\top}\cdot J\cdot DF_{t}(x, v) = J$ for any $t\geq0$ and $J=\left[\begin{array}{cc}
0 & I_{n}\\
-I_{n} & 0
\end{array}\right]$.
\item Measure-preservation: $\det(DF_{t}(x, v)=1$ for any $t\geq0$. Note that measure-preservation immediately follows from symplecticity.
\item Time-reversible: $F_{t}(x(t),-v(t))=(x,-v)$.
\end{enumerate}

Just as the Hamiltonian dynamics can be extended to the Riemannian setting, \cite{girolami2011riemann} extended HMC to a Riemannian version called Riemannian Hamiltonian Monte Carlo (RHMC); see Algorithm~\ref{alg:RHMC} for its one-step description.
In fact, HMC can be recovered from RHMC using the Euclidean metric (i.e., $g(x)=I$).

Our goal is to sample from a probability density proportional to $e^{-f(x)}$ supported on a convex body. 
To this end, we use RHMC with the Hamiltonian $H:\mcal\times\R^{n}\subset\R^{n}\times\Rn\to\R$, viewing the convex body as a Riemannian manifold $\mcal$ with a local metric $g$.

\begin{algorithm2e}[h!]

\caption{$\texttt{Riemannian Hamiltonian Monte Carlo}$}\label{alg:RHMC}

\SetAlgoLined

\textbf{Input:} Initial point $x$, step size $h$

\tcp{Step 1: Sample an initial velocity $v$}

Sample $v\sim\mathcal{N}(0,g(x))$.

\ 

\tcp{Step 2: Solve the Hamiltonian equations}

Solve the Hamiltonian equations (\ref{eq:hmc_intro}) to obtain $(x(t),v(t))$.

\

\tcp{Step 3: Metropolis-filter (skipped for ideal RHMC)}
Accept $x(h)$ with probability $\min \Par{1 , \frac{e^{-H(x(h),v(h))}}{e^{-H(x,v)}}}$. 
Otherwise, stay at $x$.
\end{algorithm2e}

\subsection{Notation and setting}

We use $\Par{\mcal,g}$ to denote a connected and compact Riemannian manifold
with a boundary and a metric $g$ on which a target distribution is
supported. For a function $f:\mcal\subset\Rn\to\R$, we denote a target
distribution by $\pi(x)$ whose density is proportional to $e^{-f(x)}$
(i.e., $\frac{d\pi}{dx}\sim e^{-f(x)}$). We use $T_{x}\mcal$ to
denote the tangent space of $\mcal$ at $x\in\mcal$. We denote by
$\pi_{x}$ the projection map onto $x$-space (i.e., $\pi_{x}(x,v)\defeq x$)
and by $i_{x}$ the inclusion map (i.e., $i_{x}(v)\defeq(x,v)$).
We reserve $h$ for the step size of RHMC. 

With both manifold $\mcal$ and tangent space $T_{x}\mcal$ endowed with the Euclidean metric, we define a map $F_{t}:\mcal\times T_{x}\mcal\to\mcal\times \bigcup_{z\in \mcal} T_{z}\mcal$ by $F_{t}(x,v)\defeq(x(t),v(t))$, where $(x(t),v(t))$ is the solution to the Hamiltonian equations at time $t$ with an initial condition $(x,v)$. 
In particular, we define $T_{x,h}:T_{x}\mcal\to\mcal$ by $T_{x,h}(v)\defeq(\pi_{x}\circ F_{h} \circ i_x)(v)=x(h)$. 
When both $\mcal$ and $T_{x}\mcal$ are endowed with the local metric $g$, we instead use $\ham_{x,t}:T_{x}\mcal\to\mcal$ defined by $\ham_{x,t}(v)\defeq x(t)$.

When a numerical integrator with step size $h$ outputs $(\bx_{h},\bv_{h})$ by solving the Hamiltonian equations with an initial condition $(x,v)$, we denote $\overline{F}_{h}(x,v)\defeq(\bx_{h},\bv_{h})$ for a function $\overline{F}_{h}:\mcal\times T_{x}\mcal\to\mcal\times \bigcup_{z\in \mcal} T_{z}\mcal$, where the domain and range are endowed with the Euclidean metric.
We define $\overline{T}_{x,h}:T_{x}\mcal\to\mcal$ (endowed with the
Euclidean metric) by $\overline{T}_{x,h}(v)=(\pi_{x}\circ\overline{F}_{h}\circ i_x)(v)=\bx_{h}$.
We drop $h$ from $T_{x,h},\overline{F}_{h}$ and $\overline{T}_{x,h}$
if the step size is clear from context.

We assume that the domain $\mcal\subset\Rn$ with a boundary is convex
and has a (highly) self-concordant barrier $\phi:\mcal\subset\Rn\to\R$
(Definition \ref{def:sc-barrier}), and that the metric $g$ is induced
by the Hessian of the barrier (i.e., $g(x)=\hess\phi(x)$). We denote
the local norm of a vector $v$ by $\norm v_{x}$ or $\norm v_{g(x)}$,
and the Riemannian distance by $d_{\phi}$ (Definition \ref{def:RiemannDistance}).
We use $a\lesssim b$ to indicate that $a\leq cb$ for some universal
constant $c>0$.

\subsection{Discretized RHMC}

We use \emph{ideal} RHMC to denote the algorithm when the Hamiltonian equations in Step 2 is accurately solved without any error. 
However, we cannot expect such an accurate ODE solver to always exist in reality, so numerical integrators with solutions that approximate the accurate ODE solutions are necessary. We use \emph{discretized} RHMC to denote the algorithm when Step 2 of RHMC is solved by a numerical integrator and a Metropolis-filter is used to correct the distribution.

We now define a condition of numerical integrators that plays an important role in our convergence-rate analysis.
\begin{defn} \label{def:stable} 
For a numerical integrator $\overline{F}$ and $(x,v)\in\mcal\times T_{x}\mcal\subset\Rn\times\Rn$, we call $\overline{F}$ \emph{sensitive} at $(x,v)$ if there exists step size $h_{0}(x,v)$ such that the numerical integrator with step size $h$ less than $h_{0}$ satisfies
\[
\frac{\Abs{D\overline{T}_{x,h}(v)}}{\Abs{DT_{x,h}(v')}}\geq0.998,
\]
where $v'$ satisfies $T_{x,h}(v')=\overline{T}_{x,h}(v)$ and the Jacobian $DT$ is taken with respect to the velocity variable. In other words, the solution of the numerical integrator changes almost as fast as the ideal solution does. Unless specified otherwise, a sensitive integrator is additionally assumed to be measure-preserving (i.e., $\det(D\overline{F}_h(x,v))=1$) and time-reversible (i.e., $\overline{F}_h(\overline{x}_h, -\overline{v}_h) = (x, -v)$).
\end{defn}

As a time-reversible numerical integrator is even-order, second-orderness automatically follows. 
That is, for sufficiently small step size $h>0$, $d_{g}(\bx_{h},x_{h})\leq C_{x}(x,v)h^{2}$
and $\norm{\bv_{h}-v_{h}}_{g(x)^{-1}}\leq C_{v}(x,v)h^{2}$ for some functions of $x$ and $v$, $C_{x}$
and $C_{v}$. In other words, the errors
of the numerical integrator $\overline{F}_{h}$ with respect to the
exact ODE solver $F_{h}$ grow at most quadratically in the step size
$h$.

This family of numerical integrators turns out to cover many commonly used integrators in practice. 
For example, the implicit midpoint integrator (IMM) (Algorithm \ref{alg:Leap-IMM}) and the generalized Leapfrog integrator (LM) (Algorithm \ref{alg:LM}) satisfy symplecticity, time-reversibility, and sensitivity (as shown in Section~\ref{sec:numerical}). Measure-preservation gives the simple formula of the acceptance probability in Step 3 of Algorithm~\ref{alg:RHMC}. Measure-preservation together with time-reversibility plays an important role  in showing that the discretized RHMC converges to its stationary distribution with density proportional to $e^{-f(x)}$ (see Theorem 8 in \cite{kook2022sampling}).

\section{Our results} \label{sec:result}

We analyze the mixing time of RHMC discretized by
numerical integrators commonly used in practice, with the Hamiltonian
set to be $H(x,v)=f(x)+\half v^{\top}\ginv xv+\half\log\det g(x).$
Previous analysis of RHMC was based on high accuracy numerical integrators,
which are not always achievable in practice \cite{lee2018convergence}, and
the complexity bounds were derived for uniform
density on a polytope. We extend the setting to sampling exponential
densities with practically feasible integrators. In the next theorem,
we denote $\mcal_{\rho}:=\Brace{x\in\mcal:\norm{\alpha}_{g(x)^{-1}}^{2}\leq10n^{2}\log^{2}\frac{1}{\rho}}$
for $\rho>0$.

\begin{restatable}{thmre}{thmdiscPoly} \label{thm:discPoly} 
Let $\pi$ be a target distribution on a polytope with $m$ constraints
in $\Rn$ such that $\frac{d\pi}{dx}\sim e^{-\alpha^{\top}x}$ for
$\alpha\in\Rn$. Let $\mcal$ be the Hessian manifold of the polytope
induced by the logarithmic barrier of the polytope. Let $\Lambda=\sup_{S\subset\mcal}\frac{\pi_{0}(S)}{\pi(S)}$
be the warmness of the initial distribution $\pi_{0}$. Let $\pi_{T}$
be the distribution obtained after $T$ steps of RHMC discretized
by a sensitive integrator on $\mcal$. For any $\veps>0$,
if for $x\in\mcal_{\frac{\veps}{2\Lambda}}$ and $v\in\Rn$ randomly
drawn from $\ncal(0,g(x))$, we have that with probability at least
$0.99,$ step size $h\leq h_0(x,v)$,
\[
h\leq\frac{10^{-20}}{n^{7/12}\log^{1/2}\frac{\Lambda}{\veps}},\ hC_{x}(x,v)\leq\frac{10^{-20}}{\sqrt{n}},\ h^{2}C_{x}(x,v)\leq\frac{10^{-10}}{n\log\frac{\Lambda}{\veps}}\text{ and }h^{2}C_{v}(x,v)\leq\frac{10^{-10}}{\sqrt{n\log\frac{\Lambda}{\veps}}},
\]
then $\dtv(\pi_{T},\pi)\leq\veps$ for $T=O\Par{mh^{-2}\log\frac{\Lambda}{\veps}}$.
\end{restatable}

By setting $C_{x} = C_{v}=0$, we can obtain the following corollary for the mixing time of the ideal RHMC in this setting.

\begin{restatable}{corre}{coridealPoly} \label{cor:idealPoly} 
Let $\pi$ be a target distribution on a polytope with $m$ constraints
in $\Rn$ such that $\frac{d\pi}{dx}\sim e^{-\alpha^{\top}x}$ for
$\alpha\in\Rn$. Let $\mcal$ be the Hessian manifold of the polytope
induced by the logarithmic barrier of the polytope. Let $\Lambda=\sup_{S\subset\mcal}\frac{\pi_{0}(S)}{\pi(S)}$
be the warmness of the initial distribution $\pi_{0}$. Let $\pi_{T}$
be the distribution obtained after $T$ iterations of the ideal RHMC
on $\mcal$. For any $\veps>0$ and step size $h=O\Par{\frac{1}{n^{7/12}\log^{1/2}\frac{\Lambda}{\veps}}}$,
there exists $T=O\Par{mn^{7/6}\log^{2}\frac{\Lambda}{\veps}}$ such
that $\dtv(\pi_{T},\pi)\leq\veps$.
\end{restatable}

After we compute the parameters $C_{x}$ and $C_{v}$ of IMM and LM
(see Section \ref{sec:numerical} and \ref{sec:polyotopes}), and identify the
sufficient conditions on the step size for their sensitivity, the following mixing
times of RHMC discretized by IMM or LM immediately follow.

\begin{restatable}{corre}{corimmPoly} \label{cor:immPoly} 
Let $\pi$ be a target distribution on a polytope with $m$ constraints
in $\Rn$ such that $\frac{d\pi}{dx}\sim e^{-\alpha^{\top}x}$ for
$\alpha\in\Rn$. Let $\mcal$ be the Hessian manifold of the polytope
induced by the logarithmic barrier of the polytope. Let $\Lambda=\sup_{S\subset\mcal}\frac{\pi_{0}(S)}{\pi(S)}$
be the warmness of the initial distribution $\pi_{0}$. Let $\pi_{T}$
be the distribution obtained after $T$ iterations of RHMC discretized
by IMM on $\mcal$. For any $\veps>0$ and step size $h=O\Par{\frac{1}{n^{3/2}\log\frac{\Lambda}{\veps}}}$,
there exists $T=O\Par{mn^{3}\log^{3}\frac{\Lambda}{\veps}}$ such
that $\dtv(\pi_{T},\pi)\leq\veps$.
\end{restatable}

\begin{restatable}{corre}{corleapPoly} \label{cor:leapPoly} 
Let $\pi$ be a target distribution on a polytope with $m$ constraints
in $\Rn$ such that $\frac{d\pi}{dx}\sim e^{-\alpha^{\top}x}$ for
$\alpha\in\Rn$. Let $\mcal$ be the Hessian manifold of the polytope
induced by the logarithmic barrier of the polytope. Let $\Lambda=\sup_{S\subset\mcal}\frac{\pi_{0}(S)}{\pi(S)}$
be the warmness of the initial distribution $\pi_{0}$. Let $\pi_{T}$
be the distribution obtained after $T$ iterations of RHMC discretized
by LM on $\mcal$. For any $\veps>0$ and step size $h=O\Par{\frac{1}{n^{3/2}\log\frac{\Lambda}{\veps}}}$,
there exists $T=O\Par{mn^{3}\log^{3}\frac{\Lambda}{\veps}}$ such
that $\dtv(\pi_{T},\pi)\leq\veps$.
\end{restatable}

In fact, Theorem \ref{thm:discPoly} comes from a general result on the mixing time of RHMC for density $e^{-f}$ on a convex body $\mcal\subset\Rn$.
We provide its informal version here and defer its full statement (Theorem \ref{thm:discGen}) to Section \ref{sec:discretizedRHMC}.
\begin{thm*}(Informal) Let $\pi$ be a target distribution on a convex set $\mcal\subset\Rn$
and $\Lambda=\sup_{S\subset\mcal}\frac{\pi_{0}(S)}{\pi(S)}$ be the
warmness of the initial distribution $\pi_{0}$. Let $\mcal$ be the
Hessian manifold with its metric induced by the Hessian of a highly
self-concordant barrier and $\pi_{T}$ the distribution obtained after
$T$ steps of RHMC discretized by a sensitive integrator on $\mcal$. 
For any $\veps>0$, let $\mcal_{\frac{\veps}{2\Lambda}}\subset\mcal$ be a convex subset of measure at least $1-\frac{\veps}{2\Lambda}$.
There is a step size bound $h_{0}$, defined in terms of smoothness parameters of the manifold and the integrator, so that for any step size $h\le h_{0}$, there exists $T=O\Par{\Par{h\psi_{\mcal_{\frac{\veps}{2\Lambda}}}}^{-2}\log\frac{\Lambda}{\veps}}$ where $\psi_{\mcal_{\frac{\veps}{2\Lambda}}}$ is the isoperimetry
of $\mcal_{\frac{\veps}{2\Lambda}}$, such that $\dtv(\pi_{T},\pi)\leq\veps$.
\end{thm*}

 
\section{Technical overview \label{sec:prelim}}

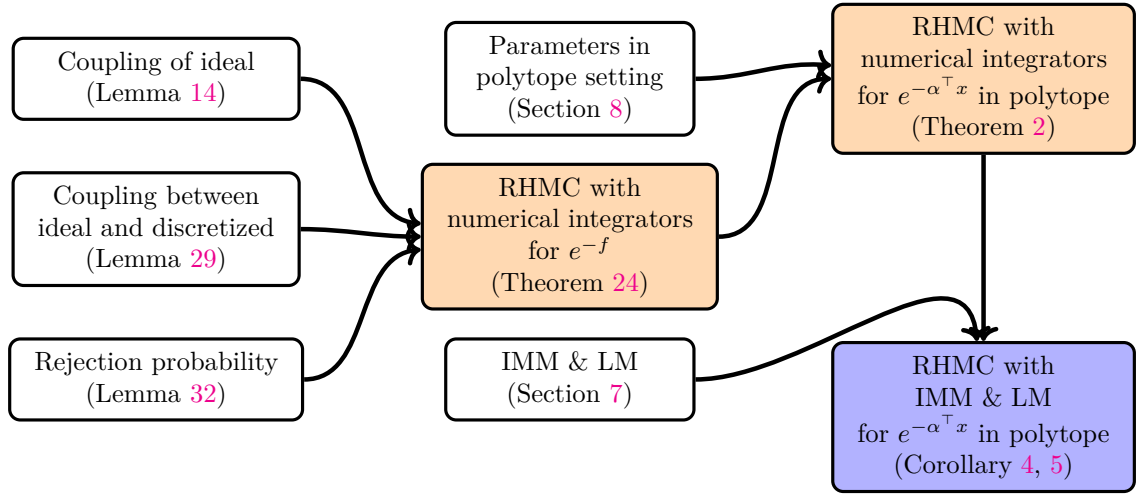
\begin{figure}[t]
	\centering \begin{tikzpicture}
		\node[rectangle,draw,rounded corners, very thick, fill=orange!30, minimum width=2.8cm] at (5.5, -0.1) (l2) {
			\begin{tabular}{c}
				RHMC with\\
				numerical integrators\\ 
				for $e^{-f}$\\
				(Theorem~\ref{thm:discGen})
		\end{tabular}};
		
		\node[rectangle,draw,rounded corners, very thick, minimum width=3.8cm] at (0, 2) (l11) {
			\begin{tabular}{c}
				Coupling of ideal\\
				(Lemma~\ref{lem:onestepIdeal})
		\end{tabular}};
		
		\node[rectangle,draw,rounded corners, very thick, minimum width=3.3cm] at (0, 0) (l12) {
			\begin{tabular}{c}
				Coupling between \\
				ideal and discretized\\
				(Lemma~\ref{lem:tvidealdis})
		\end{tabular}};
		
		\node[rectangle,draw,rounded corners, very thick, minimum width=3.3cm] at (0, -2) (l13) {
			\begin{tabular}{c}
				Rejection probability\\
				(Lemma~\ref{lem:rejFinal})
		\end{tabular}};

		\node[rectangle,draw,rounded corners, very thick, minimum width=3.3cm] at (5.5, -2) (t2) {
			\begin{tabular}{c}
				IMM \& LM\\
				(Section~\ref{sec:numerical})
		\end{tabular}};
		
		\node[rectangle,draw,rounded corners, very thick, minimum width=3.3cm] at (5.5, 2) (t1) {
			\begin{tabular}{c}
				Parameters in\\
				polytope setting\\
				(Section~\ref{sec:polyotopes})
		\end{tabular}};
		
		\node[rectangle,draw,rounded corners, very thick, fill=orange!30, minimum width=2.8cm] at (11, 2) (l3) {
			\begin{tabular}{c}
				RHMC with\\ 
				numerical integrators\\
				for $e^{-\alpha^{\top}x}$
				in polytope\\
				(Theorem~\ref{thm:discPoly})
		\end{tabular}};
		
		\node[rectangle,draw,rounded corners, very thick, fill=blue!30, minimum width=2.8cm] at (11, -2.5) (l4) {
			\begin{tabular}{c}
				RHMC with\\ 
				IMM \& LM\\
				for $e^{-\alpha^{\top}x}$
				in polytope\\
				(Corollary~\ref{cor:immPoly}, \ref{cor:leapPoly})
		\end{tabular}};

		\draw[->,  line width=.6mm] (l11) to[out=0,in=175] (l2);
		\draw[->,  line width=.6mm] (l12) to[out=0,in=180] (l2);
		\draw[->,  line width=.6mm] (l13) to[out=0,in=185] (l2);
		
		\draw[->,  line width=.6mm] (l2) to[out=0,in=180] (l3);
		\draw[->,  line width=.6mm] (t1) to[out=0,in=175] (l3);
		
		\draw[->,  line width=.6mm] (t2) to[out=0,in=95] (l4);
		\draw[->,  line width=.6mm] (l3) to[out=270,in=90] (l4);

	\end{tikzpicture} \caption{Proof outline \label{fig:outline}}
\vspace{-3mm}
\end{figure}

In this section, we provide a summary of the key proof ingredients that gives the convergence rate of RHMC with numerical integrators to samples from density $e^{-f(x)}$ on a convex body; see Figure~\ref{fig:outline} for the roadmap.
In Section~\ref{subsec:s-conductance}, we review a general technique using $s$-conductance for bounding the mixing time of a Markov chain.
In Section~\ref{subsec:pfOutline}, we summarize a refined analysis of the ideal RHMC (Section~\ref{sec:idealRHMC}) and the technique to couple the ideal and discretized RHMC (Section \ref{sec:discretizedRHMC}).
Finally in Section~\ref{sec:CandD}, we describe the high-level ideas of our analysis of the numerical integrators (Section~\ref{sec:numerical}), IMM and LM, and how to get the results for sampling from $e^{-\alpha^{\top}x}$ on the Hessian manifolds of polytopes (Section~\ref{sec:polyotopes}).

\subsection{Mixing time via $s$-conductance: isoperimetry and one-step coupling
\label{subsec:s-conductance}}

Consider a Markov chain with a state space $\mcal$, a transition
distribution $\tcal_{x}$ and stationary distribution $\pi$. 
We consider a \emph{lazy} Markov chain to avoid a uniqueness issue
of the stationary distribution. At each step, the lazy version
of the Markov chain does nothing with probability $\frac{1}{2}$
(\textit{i.e.}, stays at where it is and does not move). Note that
this change for the purpose of proof worsens the mixing time only
by a factor of $2$.

We use a standard conductance-based argument in \cite{vempala2005geometric}
to bound the mixing time, which consists of two main ingredients --
\emph{the isoperimetry} and \emph{the total variation (TV) distance coupling of one-step distributions }(Definition \ref{def:tvDistance})
staring from two close points.
\begin{defn}[$s$-conductance] Consider a Markov chain with a state space $\mcal$,
a transition distribution $\tcal_{x}$ and stationary distribution
$\pi$. For any $s\in[0,1/2)$, the $s$-\emph{conductance} of the
Markov chain is 
\[
\Phi_{s}\defeq\inf_{\pi(S)\in(s,1-s)}\frac{\int_{S}\tcal_{x}(S^{c})\pi(x)dx}{\min(\pi(S)-s,\pi(S^{c})-s)}.
\]
\end{defn}

As shown by \cite{lovasz1993random}, a lower bound on the $s$-conductance
of a Markov chain leads to an upper bound on the mixing time of the
Markov chain. 
\begin{lem}[\cite{lovasz1993random}] \label{lem:tvDecrease} Let $\pi_{t}$
be the distribution of the points obtained after $t$ steps of a lazy
reversible Markov chain with the stationary distribution $\pi$. 
Let $\Lambda=\sup_{S\subset\mcal}\frac{\pi_{0}(S)}{\pi(S)}$ be the warmness of an initial distribution $\pi_{0}$. For $H_{s}=\sup\Brace{\Abs{\pi_{0}(A)-\pi(A)}:A\subset\mcal,\,\pi(A)\leq s}$ with $0<s\leq\half$,
it follows that 
\[
\dtv(\pi_{t},\pi)\leq H_{s}+\frac{H_{s}}{s}\Par{1-\frac{\Phi_{s}^{2}}{2}}^{t}.
\]
\end{lem}

We now define the isoperimetry of a subset of $\mcal$.
\begin{defn}[Isoperimetry]
Let $(\mcal,g)$ be a Riemannian manifold and $\mcal'$ a measurable subset of $\mcal$ with $\pi(\mcal')>\half$.
The \emph{isoperimetry} $\psi$ of the subset with stationary distribution
$\pi$ is defined by 
\[
\psi_{\mcal'}=\inf_{S\subset\mcal'}\frac{\lim_{\delta\to0^{+}}\frac{1}{\delta}\int_{\{x\in\mcal'\,:\,0<d_{g}(S,x)\leq\delta\}}\pi(x)dx}{\min(\pi(S),\pi(\mcal'\backslash S))}.
\]
\end{defn}

The following illustrates how one-step coupling with the
isoperimetry leads to a lower bound on the $s$-conductance. It can
be proved similarly as Lemma 13 in \cite{lee2018convergence}.
\begin{prop}
\label{prop:conductance} For a Riemannian manifold $(\mcal,g)$,
let $\pi$ be the stationary distribution of a reversible Markov chain
on $\mcal$ with a transition distribution $\tcal_{x}$. Let $\mcal'\subset\mcal$
be a subset with $\pi(\mcal')\geq1-\rho$ for some $\rho<\frac{1}{2}$.
We assume the following one-step coupling: if $d_{g}(x,x')\leq\Delta\leq1$
for $x,x'\in\mcal'$, then $\dtv(\tcal_{x},\tcal_{x'})\leq0.9$. Then
for any $\rho\leq s<\frac{1}{2}$, the $s$-conductance is bounded
below by 
\[
\Phi_{s}\geq\Omega\Par{\psi_{\mcal'}\Delta}.
\]
\end{prop}

\begin{figure}[t]
\centering \begin{tikzpicture}[thick]

\node[label=left:$\delta_x$] at (0, 2) {};
\node[label=left:$\delta_y$] at (0, -2) {};
\node[label=right:{$\dtv(\pcal_x, \pcal_y)$: Between ideal RHMCs (Lemma~\ref{lem:onestepIdeal})}] at (5.1, 0) {};
\node[label=right:{$\dtv(\pcal_x, \oprop_x)$: Between ideal and discretized (Lemma~\ref{lem:tvidealdis})}] at (5.1, 0.9) {};
\node[label=right:{$\dtv(\pcal_y, \oprop_y)$}] at (5.1, -0.9) {};

\draw[thick, fill=black] (0, 2) circle (1mm); 
\draw[thick, fill=black] (0, -2) circle (1mm); 

\draw[thick, ->] (0, 2) ..   controls (2, 1) and (4, 1) .. (5, 0.5);
\draw[thick, ->] (0, -2) ..   controls (2, -1) and (4, -1) .. (5, -0.5);
\draw[decorate, 
	decoration = {brace}] (5.1, 0.5) --  (5.1, -0.5);
\draw[decorate, 
	decoration = {brace}] (5.1, 1.3) --  (5.1, 0.6);
\draw[decorate, 
	decoration = {brace}] (5.1, -0.6) --  (5.1, -1.3);

\draw[dashed, ->] (0, 2) ..   controls (1, 1.5) and (4, 1.5) .. (5, 1.3);
\draw[dashed, ->] (0, -2) ..   controls (1, -1.5) and (4, -1.5) .. (5, -1.3);\end{tikzpicture} \vspace{-5mm} \caption{An illustration of our approach to one-step coupling. The thick line
indicates the ideal RHMC, and the dashed line indicates the discretized
RHMC. \label{fig:illustration1}}
\vspace{-5mm}
\end{figure}
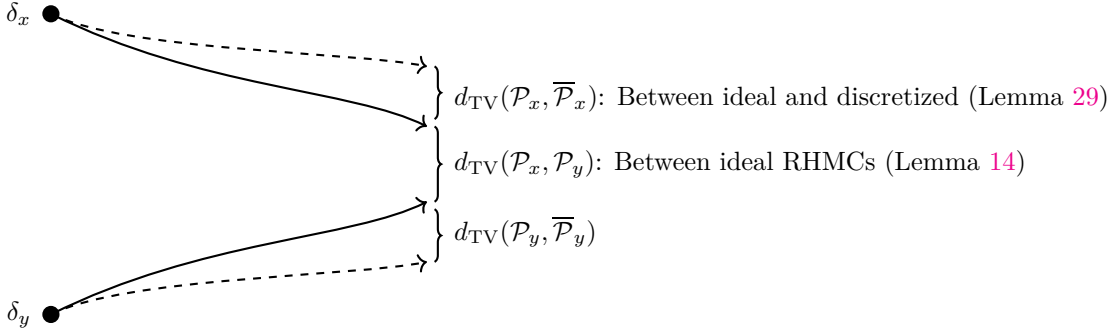

\subsection{One-step coupling of discretized RHMC\label{subsec:pfOutline}}

In light of Proposition~\ref{prop:conductance}, we can focus on coupling the one-step distributions of the discretized RHMC starting from two close-by points. 
Let $\pcal_{x}$ and $\oprop_{x}$ be the one-step distributions on $\mcal$ of the ideal and the discretized RHMC starting from $x$, respectively. 
We use $\oprop_{x}'$ to denote the discretized RHMC without the Metropolis filter. 
As illustrated in Figure~\ref{fig:illustration1}, for two close points $x$ and $y$, the triangle inequality leads to
\begin{align*}
\dtv(\oprop_{x},\oprop_{y}) & \leq\dtv(\oprop_{x},\pcal_{x})+\dtv(\pcal_{x},\pcal_{y})+\dtv(\pcal_{y},\oprop_{y})\\
 & \leq\Par{\dtv(\oprop_{x}',\pcal_{x})+\dtv(\pcal_{x},\pcal_{y})+\dtv(\pcal_{y},\oprop_{y}')}+\Par{\dtv(\oprop_{x}',\oprop_{x})+\dtv(\oprop_{y}',\oprop_{y})}.
\end{align*}
Hence, it suffices to bound $\dtv(\pcal_{x},\pcal_{y}), \dtv(\oprop_{x}',\pcal_{x})$
and $\dtv(\oprop_{x}',\oprop_{x})$, respectively. 
We bound in Section~\ref{sec:idealRHMC} the first term $\dtv(\pcal_{x},\pcal_{y})$, the TV distance of one-step distributions of the ideal RHMC. 
For the remaining terms, when numerical integrators do not preserve the Hamiltonian, a Metropolis filter is necessary to ensure that the discretized RHMC converges to a target distribution. 
Due to the filter, we need to handle a point-mass distribution at $x$.
We address this by first bounding the second term $\dtv(\oprop_{x}',\pcal_{x})$, the TV distance between the ideal and discretized RHMC \emph{without} the Metropolis filter in Section \ref{subsec:couple-ideal-dis}. 
We then separately bound the rejection probability $\dtv(\oprop_{x}',\oprop_{x})$ in Section \ref{subsec:rej-prob}.

\subsubsection{Coupling of ideal RHMC}
We summarize how to bound $\dtv(\pcal_{x}{},\pcal_{y})$ here (see Section~\ref{sec:idealRHMC} for the full version).

\begin{lem*}(Informal, Lemma~\ref{lem:onestepIdeal}) For most of $x$ and $y$, and step size $h$ small enough, if $d_{\phi}(x, y) \leq \frac{1}{100}$, then $\dtv(\pcal_{x}{},\pcal_{y})\leq O\Par{\frac{1}h}d_{\phi}(x,y) + \frac{1}{25}$.	
\end{lem*}

\paragraph{Previous approach}
\cite{lee2018convergence} provided a general framework for computing the mixing rate of RHMC on a manifold embedded in $\Rn$, in terms of the isoperimetry and smoothness
parameters depending on the manifold and step size. 
One of the major proof ingredients is one-step coupling: for two close points $x$ and $y$, the one-step distributions at $x$ and $y$ have large overlap.

They use the notion of a `regular' Hamiltonian curve, which enables them to handle this task in low level, where the regularity can be understood as \emph{average} behavior of Hamiltonian curves with high probability and is quantified by some auxiliary functions.
As the starting point of a regular Hamiltonian curve changes from $x$ to $y$ along a length-minimizing geodesic $c(s)$ joining $x$ and $y$, they find a one-to-one correspondence between regular Hamiltonian curves started at $x$ and $y$, and bound $\Abs{\frac{d}{ds}\dtv(\pcal_x, \pcal_{c(s)})}$ over $s$. 
They achieve this by quantifying the rate of changes of the probability density (see \eqref{eq:densityOneStep}).

It is daunting to directly work with the exact density function, so they make use of the following techniques: (1) Show that the determinant of Jacobian is close to $h^{n}$ up to small step size by applying a matrix-ODE theory to the second-order ODE of the Hamiltonian equation (see Lemma~\ref{lem:HamODE}). 
It allows them to work with an approximate but simpler density with the Jacobian replaced by $h^n$ (see \eqref{eq:approxDensity}).
(2) Establish the one-to-one correspondence along variations of Hamiltonian curves by the implicit function theorem; for a given endpoint $z$, as the starting point of a Hamiltonian curve moves along $c(s)$, there exists a unique initial velocity $v_{c(s)}$ at each point on $c(s)$ that brings $c(s)$ to the endpoint $z$ in step size $h$ (i.e., $\ham_{c(s),h}(v_{c(s)})=z$).
At the same time, by using the matrix-ODE theory again they show that the regularity of Hamiltonian curves does not blow up along $c(s)$ and quantify how much the proper initial velocity changes.


\paragraph{Refined analysis}
\cite{lee2018convergence}  bounded $\Abs{\frac{d}{ds}\dtv(\pcal_x, \pcal_{c(s)})}$ in terms of smoothness parameters, sumpremum bounds on some quantities defined over the regular Hamiltonian curves starting from \emph{any} point in $\mcal$. 
However, considering all starting points leads to a weaker coupling in the end.
In fact, this makes sampling from an exponential density have dependence on then condition number, since one of the smoothness parameters requires the supremum bound on $\norm{\alpha}_{g(x)^{-1}}$ over $x\in \mcal$, which can be as large as $\norm{\alpha}_2$ times the diameter of the convex body.

To achieve a condition-number independent mixing time, we work in a convex subset $\mcal_{\rho}$ (call a good region) instead of $\mcal$, which requires refinement of the framework by generalizing the smoothness parameters (Section~\ref{sec:smoothparam}) and theorems in their paper accordingly. 
It allows us to obtain a stronger coupling by only considering Hamiltonian curves starting from $\mcal_\rho$. This region is the region $\mcal'$ in Proposition~\ref{prop:conductance}.

This simple change, however, yields technical difficulties in following how \cite{lee2018convergence} proceeds with the original parameters.
Recall in the one-step coupling, they consider a Hamiltonian variation along a geodesic joining two points, but the geodesic might step out of the good region. 
To address this issue, we use the straight line between the points instead of the geodesic, as the straight line is contained in $\mcal_\rho$ due to the convexity.
We elaborate on how the technical details of the previous approach can be modified accordingly under the redefined parameters and new variation curve in order to get valid one-step coupling on this smaller region in Section~\ref{subsec:onestepofIeal}.

\subsubsection{Coupling between ideal and discretized RHMC \& Rejection probability}

We provide a summary of Section~\ref{sec:discretizedRHMC}, where we prove the following lemma and Theorem~\ref{thm:discGen}.
\begin{lem*}(Informal, Lemma~\ref{lem:tvidealdis} and Lemma~\ref{lem:rejFinal})
	For most of $(x,v)$, if step size $h$ is small enough and falls under a sensitivity regime at $(x,v)$ of a numerical integrator, then $\dtv(\oprop_{x}',\pcal_{x}) < \frac{1}{10}$ and $\dtv(\oprop_{x}',\oprop_{x}) <\frac{1}{10^3}$.
\end{lem*}

For the former (bound on $\dtv(\oprop_{x}',\pcal_{x})$), we show that the densities of the ideal and discretized RHMC are similar by relating two velocities $v$ and $v^*$, where $\overline{T}_x(v) = T_x(v^*)$. It can be reduced to establishing a constant lower bound on 
$\frac{p_{x}^{\ast}(v^{*})}{p_{x}^{\ast}(v)} \frac{\Abs{D\bTx(v)}}{\Abs{DT_{x}(v^{*})}}$ 
for the probability density $p_x^{\ast}$ of Gaussian $\ncal(0,g(x))$.

We first define numerical integrators' analogues of the smoothness parameters.
Then, we elaborate the idea above in Section~\ref{subsec:couple-ideal-dis}, where we study the dynamics of the ideal and discretized RHMC. 
In particular, we show the existence of $v^*$ for a given $v$ by the Banach fixed-point theorem and a one-to-one correspondence between them, together with an upper bound on $\norm{v-v^*}_{g^{-1}}$. This upper bound allows us to bound the ratio of $p_{x}^{\ast}(v^{*})/p_{x}^{\ast}(v)$.
We note that these results heavily rely on the stability of local norm (see Section~\ref{subsec:stability_SC}), which follows from that the local metric is given by the Hessian of self-concordant barriers. 
The ratio of the Jacobian follows from the sensitivity of the integrators. 

For the latter (bound on $\dtv(\oprop_{x}',\oprop_{x})$), we observe that the acceptance probability comes down to bounding the difference of the Hamiltonian at ideal and numerical solutions.
To bound this, we heavily use the stability of local norm as well as the quantitative relationships between the ideal and discretized RHMC established above. 
Putting these pieces together, we can obtain the mixing-time bound of the discretized RHMC in Theorem~\ref{thm:discGen}.

\subsection{Analysis of numerical integrators \& Parameter estimation in polytopes} \label{sec:CandD}

To apply the framework established so far, we analyze in Section~\ref{sec:numerical} two practical numerical integrators, IMM (Section~\ref{subsec:imm}) and LM (Section~\ref{subsec:leapfrog}), by estimating second-order parameters $C_x$ and $C_v$ and quantifying their sensitivity regimes, and then apply these estimations to sampling from distribution $e^{-\alpha^\top x}$ on a polytope in Section~\ref{sec:polyotopes}.

\begin{lem*}(Informal, adapted to polytope setting)
	For most of $(x,v)$, both IMM and LM have $C_x(x,v)= \widetilde{O}(n)$ and $C_v(x,v) = \widetilde{O}(n^{3/2})$ for step size $h=\widetilde{O}(1/\sqrt n)$.
	The sensitivity region is $h = \widetilde{O}(1/n)$ for IMM and $h = \widetilde{O}(1/n^{3/2})$ for LM.
\end{lem*}

To analyze one-step process of each numerical integrator, we need to keep track all the quantities explicitly to obtain the condition-number independence. 
However, the implicit nature of both integrators lead to \emph{coupled} equations for $x$ and $v$, making the  analysis complicated.
To address this, we handle these coupled equations parallelly by moving back and forth between local norms at different points, where we use the stability of local norm due to self-concordance.
We remark that our approach to analyze each integrator depends on the specific implementation of the integrator, so each integrator requires slightly different techniques in this task.

For the sensitivity, we apply implicit differentiation to the one-step equation of each integrator, obtaining a matrix equation in the form of $(I-h E)D\overline{F}_h = h C$ for matrices $E, C \in \R^{2n\times 2n}$.
We use matrix-perturbation theory to quantify a sufficient condition on the step size $h$ that ensures the invertibility of $(I-h E)$, obtaining an equation of the form $D\overline{F}_h = \sum_{i=0}^{\infty} (hC')^i$ for a matrix $C'\in \R^{2n\times 2n}$.
By extracting the upper-right $n\times n$ block matrix, it follows that $D\overline{T}_h = I + E'$ for $E' = \sum_{i=1}^{\infty} (hC^{*})^i $ with a matrix $C^{*}\in \R^{n\times n}$.
Using the self-concordance of the local metric, we get upper bounds on matrix quantities of $C^*$ including the trace, two-norm and Frobenius norm.
With $E'$ viewed as perturbation, we apply matrix-perturbation theory again to estimate a lower bound on $\Abs{D\overline{T}_h}$.

Lastly in Section~\ref{sec:polyotopes}, we show that $\mcal_{\rho} = \Brace{x\in \mcal : \norm{\alpha}^2_{g(x)^{-1}} \leq 10n^2 \log^2\frac{1}{\rho}}$ is convex by checking the second-order condition and that $\mcal_{\rho} $ has large measure by using a functional inequality. 
Then we compute all parameters discussed so far -- isoperimetry, smoothness parameters of the manifold and numerical integrator -- for the polytope setting, putting them together to obtain the results in Section~\ref{sec:result}.

\section{Convergence rate of ideal RHMC\label{sec:idealRHMC}}

\cite{lee2018convergence} provided a general framework
for computing the mixing rate of RHMC on a manifold embedded in $\Rn$.
They represent the mixing rate in terms of the isoperimetry and smoothness
parameters depending on the manifold and step size. In particular,
they explicitly compute those parameters and isoperimetry for the
uniform distribution on a polytope with $m$ constraints, concluding
that the mixing rate of RHMC on the Hessian manifold induced by the
logarithmic barrier of the polytope is $O\Par{mn^{2/3}}$. Notably,
this mixing rate is independent of the condition number of the polytope.
Independence of the condition number is desirable in practice, since
real-world instances are highly skewed and thus make it challenging
for sampling algorithms to sample efficiently.

Going beyond uniform sampling, we would like to obtain the condition-number-independence
of RHMC for more densities. However, even an extension to an exponential
density needs care to avoid dependence on a condition number (such
as the diameter of the domain).

In this section, we refine this framework by working on a subset $\mcal_{\rho}$
instead of $\mcal$ and extending the smoothness parameters and theorems
developed in their paper accordingly. It enables us to couple the
one-step distributions of the ideal RHMC starting at two close points
by bounding the TV distance in terms of the smoothness parameters. 

\subsection{Auxiliary function and smoothness parameters \label{sec:smoothparam}}

We redefine those smoothness parameters in \cite{lee2018convergence}
that depend on a subset $\mcal_{\rho}$ of manifold (internally parameterized
by $\rho>0$) and step size $h$, pointing out how ours differ from
the original ones. We then develop the theory for one-step coupling
based on the new parameters.

\subsubsection{Working in high probability region\label{subsec:highprobregion}}

When defining smoothness parameters, \cite{lee2018convergence} pays attention to ``well-behaved'' Hamiltonian curves $\gamma$ starting at \emph{any} point in $\mcal$, where the well-behavedness may be viewed as the
average behavior of Hamiltonian curves with \emph{high probability}
and is quantified by some auxiliary function. Then the smoothness
parameters are estimated by bounding some quantities along the curves.
To do so, they should give supremum bounds on those parameters over
all points in $\mcal$, which lead to a weaker mixing rate in the
end.

For a refined analysis, we apply a high-probability idea once again
to starting points of curves this time. In other words, we consider
well-behaved Hamiltonian curves starting only from a good region that
has high probability. Then we couple the one-step distributions at
two close-by points only in this region. This region will serve as
$\mcal'$ in Proposition \ref{prop:conductance}.

This simple change, however, turns out to yield technical difficulties
in following how \cite{lee2018convergence} proceeds with the original
parameters. In bounding the overlap of the one-step distributions,
they deal with Hamiltonian curves and Hamiltonian variations, starting
points of which are on a geodesic between two points, but the geodesic
might step out of the good region. Hence, it leads to us considering
a different path joining two points instead of the geodesic. We choose
the straight line between two points instead and carefully check if
the original approach to one-step coupling still goes through. In
addition to this, we have to redefine each of the smoothness parameters
and modify most of the statements proven in \cite{lee2018convergence}
accordingly, as we work in the region smaller than the entire domain.
We now formalize this approach.
\begin{defn}
Let $\pi$ be a target distribution on $\mcal$ such that $\frac{d\pi}{dx}\sim e^{-f(x)}$.
Given $\rho>0$, we call a measurable subset $\mcal_{\rho}$ of $\mcal$
a \emph{good region }if it is convex and has measure $\pi(\mcal_{\rho})\geq1-\rho$.
\end{defn}

\paragraph{Good region for exponential density.}

As mentioned earlier, the necessity of a refined analysis naturally
arises in attempts to obtain a condition-number-independent mixing
rate of RHMC for $f(x)=e^{-\alpha^{\top}x}$. One of parameters in
\cite{lee2018convergence} depends on $\sup_{x\in\mcal}\norm{\grad f(x)}_{g(x)^{-1}}^{2}=\sup_{x\in\mcal}\norm{\alpha}_{g(x)^{-1}}^{2}$,
but this supremum bound can be worsened by scaling up $\alpha$, and
even for fixed $\alpha$ it can be as large as the diameter of $\mcal$.
To address this issue, for given $\rho>0$ we work on a smaller convex
region that has probability at least $1-\rho$, in which the quantity
only depends on the dimension $n$ and $\rho$, and set it to be a
good region. To be precise, we will take $\mcal_{\rho}=\Brace{x\in\mcal:\norm{\grad f(x)}_{g(x)^{-1}}^{2}\leq10n^{2}\log^{2}\frac{1}{\rho}}$
for the exponential densities.

\subsubsection{Auxiliary function $\ell$ with parameters $\ell_{0}$ and $\ell_{1}$}

Initial velocities of Hamiltonian trajectories drawn from $\ncal(0,g(x)^{-1})$
can be large even though it rarely happens, as seen in the standard
concentration inequality for Gaussian distributions. Since those worst-case
trajectories lead to a weaker coupling, \cite{lee2018convergence}
focuses on ``well-behaved'' Hamiltonian trajectories rather than
all trajectories. They formalize this idea by defining an auxiliary
function $\ell$, which measures how regular a Hamiltonian trajectory
is, along with two parameters $\ell_{0}$ and $\ell_{1}$.
\begin{defn}
\label{def:auxiliary} An auxiliary function $\ell$ with parameters
$\ell_{0}$ and $\ell_{1}$ is a function that assigns a non-negative
real value to any Hamiltonian curve with step size $h$, such that 
\begin{itemize}
\item For any $x\in\mcal_{\rho}$, we have 
\[
\P_{\gamma}\Par{\ell(\gamma)>\half\ell_{0}}<\frac{1}{100}\min\Par{1,\frac{\ell_{0}}{\ell_{1}h}},
\]
where $\gamma$ is a Hamiltonian trajectory starting at $x$ with
an initial random velocity drawn from $\ncal(0,g(x)^{-1})$.
\item For any variations $\gamma_{s}$ starting from $\mcal_{\rho}$ with
$\ell(\gamma_{s})\leq\ell_{0}$, we have 
\[
\Abs{\frac{d}{ds}\ell(\gamma_{s})}\leq\ell_{1}\cdot\Par{\norm{\frac{d}{ds}\gamma_{s}(0)}_{\gamma_{s}(0)}+\delta\norm{D_{s}\gamma_{s}'(0)}_{\gamma_{s}(0)}},
\]
where the variations $\gamma_{s}$ satisfy the Hamiltonian equations,
and $D_{s}$ denotes the covariant derivative of the velocity field
$\gamma_{s}'(0)$ along a curve of starting points of the variations.
\end{itemize}
\end{defn}

In the original definitions, the parameter $\ell_{0}$ is defined
over the Hamiltonian curves starting from any $x\in\mcal$, and the
parameter $\ell_{1}$ is defined over the variations starting from
any $x\in\mcal$ with $\ell(\gamma_{s})\leq\ell_{0}$. 

Intuitions behind these parameters can be understood in the following
way. The auxiliary function $\ell$ measures how regular Hamiltonian
trajectories are, and $\ell_{0}$ serves as a threshold that allows
us to consider only Hamiltonian curves with regularity below the threshold,
while it is large enough to capture most trajectories.

To see the role of $\ell_{1}$, we run through a high-level idea for
one-step coupling. For a given endpoint $z$, we consider the set
of regular Hamiltonian curves $\gamma_{x}$ stating at $x$ with $\ell(\gamma_{x})\leq\half\ell_{0}$,
which takes into account most trajectories due to the definition of
$\ell_{0}$. Along the straight line joining $x$ and $y$, we smoothly
vary the starting point of the Hamiltonian curve to obtain a Hamiltonian
curve $\gamma_{y}$ starting at $y$ with the same endpoint $z$ and
then find a correspondence between $\gamma_{x}$ and $\gamma_{y}$.
In doing so, it is desirable to maintain the regularity of Hamiltonian
curves. In other words, the auxiliary function should not change rapidly
so that $\ell(\gamma_{y})$ is still bounded by $\ell_{0}$. We enforce
this situation via the parameter $\ell_{1}$ that bounds the rate
of change of the auxiliary function, $\frac{d}{ds}\ell(\gamma_{s})$,
along the straight line.

\subsubsection{Smoothness parameters $R_{1},R_{2},R_{3}$}

In relating the regular Hamiltonian curves $\gamma_{x}$ and $\gamma_{y}$,
some quantities naturally arise from the proof. We begin with the
definition of Riemannian curvature tensor and then define three important
parameters that govern those quantities.
\begin{defn}
The Riemannian curvature tensor is a map $R:V(\mcal)\times V(\mcal)\times V(\mcal)\to V(\mcal)$
for $V(\mcal)$, the collection of vector fields on $\mcal$, defined
by
\[
R(u,v)w=\grad_{u}\grad_{v}w-\grad_{v}\grad_{u}w-\grad_{[u,v]}w\quad\text{for }u,v,w\in V(\mcal),
\]
where $\grad$ is the Levi-Civita connection on $\mcal$, and $[u,v]\defeq\grad_{u}v-\grad_{v}u$
is the Lie bracket of the vector fields $u$ and $v$.
\end{defn}

\begin{defn}
\label{def:smoothness} Given an auxiliary function $\ell$ with parameters
$\ell_{0}$ and $\ell_{1}$ and the operator $\Phi(\gamma,t):V(\mcal)\to V(\mcal)$
defined by $\Phi(\gamma,t)u\defeq D_{u}\mu(\gamma(t))-R(u,\gamma'(t))\gamma'(t)$, 
\begin{itemize}
\item $R_{1}$ is a parameter such that for any $t\in[0,h]$ and any Hamiltonian
curves $\gamma$ starting from $\mcal_{\rho}$ with step size $h$
and $\ell(\gamma)\leq\ell_{0}$ 
\[
\norm{\Phi(\gamma,t)}_{F,\gamma(t)}\defeq\sqrt{\E_{v,w\sim\ncal(0,g(\gamma(t))^{-1})}\inner{v,\Phi(\gamma,t)w}_{g(\gamma(t))}^{2}}\leq R_{1}.
\]
\item $R_{2}$ is a parameter such that for any $t\in[0,h]$, any Hamiltonian
curves $\gamma$ starting from $\mcal_{\rho}$ with step size $h$
and $\ell(\gamma)\leq\ell_{0}$, any curve $c(s)$ starting from $\gamma(t)$
and any vector field $v(s)$ along the curve $c(s)$ with $v(0)=\gamma'(t)$,
\[
\Abs{\frac{d}{ds}\tr\Phi(v(s))\bigg\vert_{s=0}}\leq R_{2}\cdot\Par{\norm{\frac{dc}{ds}\bigg\vert_{s=0}}_{\gamma(t)}+h\norm{D_{s}v(s)\vert_{s=0}}_{\gamma(t)}},
\]
 where $v(s)$ in $\Phi$ indicates a Hamiltonian curve at time $t$
starting with an initial condition $(c(s),v(s))$.
\item $R_{3}$ is a parameter such that for any Hamiltonian curves $\gamma$
starting from $\mcal_{\rho}$ with step size $h$ and $\ell(\gamma)\leq\ell_{0}$,
if $\zeta(t)\in T_{\gamma(t)}\mcal$ is the parallel transport of
the vector $\gamma'(0)$ along $\gamma$, then 
\[
\sup_{t\in[0,h]}\norm{\Phi(\gamma,t)\zeta(t)}_{\gamma(t)}\leq R_{3}.
\]
\end{itemize}
\end{defn}

\subsection{One-step coupling and convergence rate}\label{subsec:onestepofIeal}

In this section, we bound the TV distance of two one-step distributions
of the ideal RHMC starting at two close points in terms of the redefined
parameters and step size. The following result is a slight tweak of
Theorem 29 in \cite{lee2018convergence}.
\begin{lem}
\label{lem:onestepIdeal} For $x,y\in\mcal_{\rho}$ and step size
$h\leq\min\Par{\frac{1}{10^{5}R_{1}^{1/2}},\Par{\frac{\ell_{0}}{10^{3}R_{1}^{2}\ell_{1}}}^{1/5}}$,
if $d_{\phi}(x,y)\leq\frac{1}{100}\min\Par{1,\frac{\ell_{0}}{\ell_{1}}}$,
then
\[
\dtv(\pcal_{x},\pcal_{y})\leq O\Par{\frac{1}{h}+h^{2}R_{2}+hR_{3}}d_{\phi}(x,y)+\frac{1}{25}.
\]
\end{lem}

This provides the convergence rate of the ideal RHMC for a general
density $e^{-f}$, which is a slight generalization of Theorem 30
in \cite{lee2018convergence}.

\begin{restatable}{propre}{propidealConv} \label{prop:idealConv}
Let $\pi_{T}$ be the distribution obtained after $T$ steps of a
lazy ideal RHMC with the stationary distribution $\pi$ satisfying
$\frac{d\pi}{dx}\sim e^{-f(x)}$. Let $\Lambda=\sup_{S\subset\mcal}\frac{\pi_{0}(S)}{\pi(S)}$
be the warmness of an initial distribution $\pi_{0}$. For any $\veps>0$,
let $\rho=\frac{\veps}{2\Lambda}$ and $\mcal_{\rho}$ a good region.
If step size $h$ satisfies 
\[
h^{2}\leq\frac{1}{10^{10}R_{1}},\ h^{5}\leq\frac{\ell_{0}}{10^{3}R_{1}^{2}\ell_{1}},\ h^{3}R_{2}+h^{2}R_{3}\leq\frac{1}{10^{10}}\text{ and }h\leq\frac{1}{10^{10}}\min\Par{1,\frac{\ell_{0}}{\ell_{1}}},
\]
where the parameters are defined in Definition \ref{def:auxiliary}
and \ref{def:smoothness}, then for the isoperimetry $\psi_{\mcal_{\rho}}$
of $\mcal_{\rho}$ there exists $T=O\Par{\Par{h\psi_{\mcal_{\rho}}}^{-2}\log\frac{1}{\rho}}$
such that $\dtv(\pi_{T},\pi)\leq\veps$. 
\end{restatable}

Toward this result, we walk through how each lemma and theorem should
change so that they can be put together well, along with the modified
smoothness parameters and auxiliary function. We start with the formula
of the probability density of the one-step distribution at $x$.
\begin{lem}[\cite{lee2018convergence}, Lemma 10] The probability density
of one-step distribution of RHMC at $x\in\mcal\subset\Rn$ is 
\begin{equation}
p_{x}(z)=\sum_{v_{x}:\ham_{x,h}(v_{x})=z}\underbrace{\Abs{D\ham_{x,h}(v_{x})}^{-1}\sqrt{\frac{\Abs{g(z)}}{(2\pi)^{n}}}\exp\Par{-\half\norm{v_{x}}_{x}^{2}}}_{\defeq p_{x}^{0}(v_{x})}.\label{eq:densityOneStep}
\end{equation}
\end{lem}

Note that the velocity $v_{x}$ is normalized by $g(x)^{-1}$, since
the domain of $\ham_{x,h}$ is endowed with the local metric $g(x)$.
In the Euclidean coordinate, the density can be rewritten as 
\begin{align}
p_{x}(z) & =\sum_{v_{x}':T_{x,h}(v_{x}')=z}\Abs{DT_{x,h}(v_{x}')}^{-1}\frac{1}{\sqrt{(2\pi)^{n}\Abs{g(x)}}}\exp\Par{-\half\norm{v_{x}'}_{g(x)^{-1}}^{2}}\nonumber \\
 & =\sum_{v_{x}':T_{x,h}(v_{x}')=z}\Abs{DT_{x,h}(v_{x}')}^{-1}p_{x}^{*}(v_{x}'),\label{eq:onestepEuclidean}
\end{align}
where $p_{x}^{*}$ is the probability density of the Gaussian distribution
$\ncal(0,g(x))$. This relation follows from $v_{x}'=g(x)v_{x}$ and
$\Abs{DT_{x,h}(v_{x}')}=\frac{\Abs{D\ham_{x,h}(v_{x})}}{\sqrt{\Abs{g(x)}\Abs{g(x')}}}$
(see the proof of Proposition \ref{prop:stability}) for $x'=\ham_{x,h}(v_{x})=T_{x,h}(v_{x}')$.

We can derive (\ref{eq:onestepEuclidean}) in the following way as
well. Intuitively, the probability of moving from $x$ to $z$ through
one step of RHMC is the summation of the probability of choosing a
proper initial velocity that brings $x$ to $z$, which is the probability
density function of $\ncal(0,g(x))$ divided by $\Abs{DT_{x,h}(v_{x}')}$.
This Jacobian term comes from the change of variables used when moving
back from the position space $z$ to the velocity space $v_{x}'$.

\paragraph{High-level idea.}

We are ready to run through a high-level idea of the one-step coupling
proof in \cite{lee2018convergence}. For two close-by points, it is
plausible that the probability densities at $x$ and $y$ are similar,
and one should relate those two densities to quantify how close they
are, which in turn results in a bound on the overlap of two one-step
distributions. It is a Hamiltonian curve that enables them to handle
this task in low level. Then they find a one-to-one correspondence
between the set of regular Hamiltonian curves from $x$ and $y$.

As one varies a starting point of a regular Hamiltonian curve from
$x$ to $y$ along a curve $c(s)$ joining $x$ and $y$, one should
quantify how fast each term in (\ref{eq:densityOneStep}) changes.
To this end, they first prove that the determinant of Jacobian is
close to $h^{n}$ and that $\ham_{x,h}$ is locally injective, which
makes it possible to work with an approximate but simpler density
\begin{equation} \label{eq:approxDensity}
    \tilde{p}_{x}(z)\defeq\sum_{v_{x}:\ham_{x,h}(v_{x})=z}\sqrt{\frac{\Abs{g(z)}}{(2\pi h^{2})^{n}}}\exp\Par{-\half\norm{v_{x}}_{x}^{2}}.
\end{equation}

Next, they prove the following on variations of Hamiltonian curves;
for a given endpoint $z$, as the starting point of a Hamiltonian
curve moves along $c(s)$, there exists a unique initial velocity
$v_{c(s)}$ at each point on $c(s)$ that brings $c(s)$ to the fixed
endpoint $z$ in step size $h$ (i.e., $\ham_{c(s),h}(v_{c(s)})=z$).
At the same time, they prove that the regularity of Hamiltonian curves,
$\ell(\gamma_{c(s)})$, does not change too rapidly along $c(s)$,
quantifying how much the proper initial velocity changes as well.
These are enough to achieve one-step coupling in terms of only $R_{1}$.
For further improvement, a more accurate estimate (\ref{eq:oneStepApproximate})
of the determinant of Jacobian is used instead, leading to an improved
bound via $R_{2}$ and $R_{3}$.

Following this approach, we elaborate on how each of the proof ingredients
can be formalized under the redefined parameters. The first ingredient
about the local injectivity of $\ham_{x,h}$ and an approximation
of its Jacobian follows from Lemma 22 in \cite{lee2018convergence}
by restricting starting points of Hamiltonian curves to $\mcal_{\rho}$
in the statement.
\begin{lem}
\label{lem:invertHam} Let $\gamma(t)=\ham_{x,t}(v_{x})$ be a Hamiltonian
curve starting at $x\in\mcal_{\rho}$ with $\ell(\gamma)\leq\ell_{0}$
and step size $h$ satisfying $h^{2}\leq1/R_{1}$. Then $D\ham_{x,h}$
is invertible and 
\[
\Abs{\log\Abs{\frac{1}{h}D\ham_{x,h}(v_{x})}-\int_{0}^{h}\frac{t(h-t)}{h}\tr\Phi(t)dt}\leq\frac{(h^{2}R_{1})^{2}}{10}.
\]
\end{lem}

As a corollary, we obtain the following estimate on the Jacobian of
the Hamiltonian map.
\begin{cor}
\label{lem:jacoIdeal} Let $(x(t),v(t))$ be the Hamiltonian curve
starting with $(x,v)\in\mcal_{\rho}\times T_{x}\mcal$, where $T_{x}\mcal\subset\Rn$
is endowed with the local metric $g(x)$. For step size $h$ with
$h^{2}\leq\frac{1}{10^{5}\sqrt{n}R_{1}}$, and $v\in T_{x}\mcal$
with $\ell(\ham_{x,t}(v))\leq\ell_{0}$, we have 
\[
h^{n}e^{-\frac{1}{600}}\leq\Abs{D\ham_{x,h}(v)}\leq h^{n}e^{\frac{1}{600}}.
\]
Namely, $\Abs{\frac{1}{h^{n}}\Abs{D\ham_{x,h}(v)}-1}\leq0.002$
\end{cor}

The next one is the local uniqueness and existence of Hamiltonian
variations, obtained by adjusting Lemma 23 in \cite{lee2018convergence}.
\begin{lem}
\label{lem:localUnique} Let $\gamma(t)=\ham_{x,t}(v_{x})$ be a Hamiltonian
curve starting at $x\in\mcal_{\rho}$ with $\ell(\gamma)\leq\ell_{0}$
and step size $h$ satisfying $h^{2}\leq1/R_{1}$. Let $x=\gamma(0)$
and $z=\gamma(h)$ be its endpoints. 
\begin{itemize}
\item For a neighborhood $U$ of $x\in\mcal_{\rho}$ and neighborhood $V$
of $v_{x}$, there exists a unique smooth invertible vector field
$v:U\to V$ such that $v(x)=v_{x}$ and $z=\ham_{x,h}(v(y))$ for
any $y\in U$.
\item For $\eta\in T_{x}\mcal$, we have that $\norm{\grad_{\eta}v(x)}_{x}\leq\frac{5}{2h}\norm{\eta}_{x}$
and $\norm{\frac{1}{h}\eta+\grad_{\eta}v(x)}_{x}\leq\frac{3}{2}R_{1}h\norm{\eta}_{x}$.
\item Let $\gamma_{s}(t)=\ham_{c(s),h}(v(c(s))$ be a variation of $\gamma$
along a path $c(s)$ in $U$ with $c(0)=x$ and $c'(0)=\eta$. For
$t\in[0,h]$, we have $\norm{\frac{\del\gamma_{s}(t)}{\del s}\big|_{s=0}}_{\gamma(t)}\leq5\norm{\eta}_{x}$
and $\norm{D_{s}\gamma_{s}'(t)\big|_{s=0}}_{\gamma(t)}\leq\frac{10}{h}\norm{\eta}_{x}$.
\end{itemize}
\end{lem}

The first item reveals the local uniqueness and existence of proper
initial velocities at any starting point around $x$. The second item
bounds how fast the initial velocity at $x$ change in a given direction
$\eta$. The last item extends the second result to each point on
$\gamma(t)$.

The corresponding result in \cite{lee2018convergence} is made for
all regular Hamiltonian curves from $\mcal$. Since its proof relies
on Lemma \ref{lem:invertHam} to apply the implicit function theorem
to $f(y,w)=\ham_{y,h}(w)$ and also on the definition of $R_{1}$
for the second result, the statement should be restricted to Hamiltonian
curves starting from $\mcal_{\rho}$.

We can now prove that regular Hamiltonian curves starting at $x$
with an endpoint $z$ can be smoothly varied along the straight line
between $x$ and $y$, with the regularity of variations almost preserved.
\begin{lem}
\label{lem:extendGeo} Let $\gamma(t)=\ham_{x,t}(v_{x})$ be a Hamiltonian
curve starting at $x\in\mcal_{\rho}$ with $\ell(\gamma)\leq\half\ell_{0}$
and step size $h$ satisfying $h^{2}\leq1/R_{1}$. Let $x=\gamma(0)$
and $z=\gamma(h)$ be its endpoints. For $y\in\mcal_{\rho}$ and $\beta=\frac{y-x}{\norm{y-x}_{x}}$,
let $c(s)=s\beta+x$ be a straight line joining $x$ and $y$ with
$c(0)=x$ and $c(s')=y$. Let $s'\leq\frac{1}{100}\min\Par{1,\frac{\ell_{0}}{\ell_{1}}}$.
\begin{itemize}
\item There exists a unique velocity field $v$ along $c$ such that $z=\ham_{c(s),h}(v(c(s)))$.
Furthermore, this vector field is also uniquely determined by $c(s)$
and $v(c(s))$ on $c(s)$.
\item $\ell(\ham_{c(s),h}(v(c(s)))\leq\ell_{0}$ for all $s$.
\end{itemize}
\end{lem}

Compared to the original result, we use a straight line instead of
a unit-speed geodesic between $x$ and $y$, as the geodesic might
escape the good region $\mcal_{\rho}$. The first item implies that
there exists an initial velocity $v(c(s))$ at each point $c(s)$
such that we can reach the fixed endpoint $z$ via the Hamiltonian
trajectory with the initial condition $(c(s),v(c(s))$. The second
item indicates that the regularity of such Hamiltonian trajectories
is preserved up to constant along the straight line.
\begin{proof}
The first result can be proven similarly as in \cite{lee2018convergence}.
For the second, we denote by $\gamma_{s}$ the Hamiltonian trajectory
starting at $c(s)$ with the proper initial velocity $v(c(s))$. We
note that $s'=\norm{y-x}_{x}$ and $\frac{dc(s)}{ds}=\beta$. By self-concordance
of $g$, we have that $\norm{\beta}_{c(s)}\leq\Par{1+\norm{x-c(s)}_{x}}\norm{\beta}_{x}\leq1+s'$.
Thus by Lemma \ref{lem:localUnique}, 
\[
\norm{D_{s}v(s)}_{c(s)}\leq\frac{5}{2h}\norm{\beta}_{c(s)}\leq\frac{5}{2h}\Par{1+\norm{x-c(s)}_{x}},
\]
and 
\begin{align*}
\ell(\gamma_{y}) & \leq\ell(\gamma_{x})+\int_{0}^{s'}\Abs{\frac{d}{ds}\ell(\gamma_{s})}ds\leq\half\ell_{0}+\ell_{1}\int_{0}^{s'}\Par{\norm{\beta}_{c(s)}+\frac{5}{2}\norm{\beta}_{c(s)}}ds\\
 & \leq\half\ell_{0}+\ell_{1}s'\cdot\frac{7}{2}(1+s')\leq\half\ell_{0}+\ell_{1}s'\cdot4\\
 & \leq\ell_{0},
\end{align*}
where we used that $1+s'\leq1.01$ and $s'\leq\frac{1}{100}\frac{\ell_{0}}{\ell_{1}}$.
\end{proof}
The next two lemmas provide bounds on some quantities via $R_{2}$
and $R_{3}$, which are modifications of Lemma 34 and 32 in \cite{lee2018convergence}.
\begin{lem}
\label{lem:R2bound} Let $\gamma_{s}$ be a family of Hamiltonian
curves joining $c(s)$ and $z$ defined in Lemma \ref{lem:extendGeo}
with $\ell(\gamma_{s})\leq\ell_{0}$ and step size $h$ satisfying
$h^{2}\leq1/R_{1}$. Then, 
\[
\Abs{\int_{0}^{h}\frac{t(h-t)}{h}\frac{d}{ds}\tr\Phi\Par{\gamma_{s}'(t)}dt}\leq O\Par{h^{2}R_{2}}.
\]
\end{lem}

Recall that $\gamma_{s}$ given in Lemma \ref{lem:extendGeo} has
a starting point in $\mcal_{\rho}$ with $\ell(\gamma_{s})\leq\ell_{0}$.
Since its original proof uses the definition of $R_{2}$ and Lemma
\ref{lem:localUnique}, and they are applicable to regular Hamiltonian
curves starting from $\mcal_{\rho}$ with $\ell(\gamma_{s})\leq\ell_{0}$,
the original proof of this lemma still works with our new definitions
of the parameters.
\begin{lem}
\label{lem:R3bound} Let $\gamma(t)=\ham_{x,t}(v_{x})$ be a Hamiltonian
curve starting at $x\in\mcal_{\rho}$ with $\ell(\gamma)\leq\ell_{0}$
and step size $h$ satisfying $h^{2}\leq1/R_{1}$. Then, 
\[
\frac{h}{2}\Abs{\grad_{\eta}\norm{v(x)}_{x}^{2}}\leq\Abs{\inner{v_{x},\eta}_{x}}+3h^{2}R_{3}\norm{\eta}_{x}.
\]
\end{lem}

We can follow its original proof by using Lemma \ref{lem:localUnique}
and the definition of $R_{3}$, as the regular Hamiltonian curve considered
starts at $x\in\mcal_{\rho}$. We are now ready to prove Lemma \ref{lem:onestepIdeal}.

\begin{proof}[Proof of Lemma~\ref{lem:onestepIdeal}]
Let $c(s)$ be the straight line joining $x$ and $y$, contained
in $\mcal_{\rho}$ due to the convexity of $\mcal_{\rho}$. We denote
$\tilde{\ell}\defeq\min\Par{1,\frac{\ell_{0}}{\ell_{1}h}}$. For $x\in\mcal_{\rho}$,
let $V_{x}$ be the set of velocities $v_{x}$ such that $\ell(\ham_{x,h}(v_{x}))\leq\half\ell_{0}$.
Note that $\pcal_{x}^{*}(V_{x}^{c})\leq\frac{1}{100}\tilde{\ell}$
by the definition of $\ell_{0}$, where $\pcal_{x}^{*}$ is the one-step
distribution over velocities (not position) at $x$. Since $c(s)$
is contained in $\mcal_{\rho}$ and $\gamma(t)=\ham_{x,t}(v_{x})$
has regularity at most $\half\ell_{0}$, Lemma \ref{lem:extendGeo}
guarantees the existence of a family of Hamiltonian variations $\gamma_{s}(t)$
joining $c(s)$ and $\gamma(h)$ with $\ell(\gamma_{s})\leq\ell_{0}$
for all $s\in[0,\norm{y-x}_{x}]$.

We define an approximate probability density $\tilde{p}_{c(s)}$ of
$p_{c(s)}$, where $p_{c(s)}$ is the probability density of $\pcal_{c(s)}$.
Driven by Lemma \ref{lem:invertHam}, for $z\in\mcal$ we define 
\begin{equation}
\tilde{p}_{c(s)}(z)\defeq\sum_{v:\,\ham_{c(s),h}(v)=z}\underbrace{\sqrt{\frac{\Abs{g\Par{\ham_{c(s),h}(v)}}}{(2\pi h^{2})^{n}}}\exp\Par{-\int_{0}^{h}\frac{t(h-t)}{h}\tr\Phi(\gamma_{s},t)dt}\cdot\exp\Par{-\half\norm v_{c(s)}^{2}}}_{\defeq\tilde{p}_{c(s)}^{0}(v)},\label{eq:oneStepApproximate}
\end{equation}
which is obtained by using $\exp\Par{-\int_{0}^{h}\frac{t(h-t)}{h}\tr\Phi(t)dt}$
in place of $\Abs{D\ham_{c(s),h}(v)}^{-1}$ in $p_{x}(z)$ (see (\ref{eq:densityOneStep})). 

We now relate $\tilde{p}_{c(s)}$ to $p_{c(s)}$. Note that the ratio
of the summand of $\tilde{p}_{c(s)}(z)$ and $p_{c(s)}(z)$ is equal
to $\tilde{p}_{c(s)}^{0}(v)/p_{c(s)}^{0}(v)=\frac{\frac{1}{h^{n}}\exp\Par{-\int_{0}^{h}\frac{t(h-t)}{h}\tr\Phi(\gamma_{s},t)dt}}{\Abs{D\ham_{c(s),h}(v)}^{-1}}$
(see (\ref{eq:densityOneStep})). Due to $\ell(\gamma_{s})\leq\ell_{0}$,
we can apply Lemma \ref{lem:invertHam} to $\gamma_{s}(t)$, obtaining
\[
\exp\Par{-\frac{(h^{2}R_{1})^{2}}{10}}\leq\frac{\frac{1}{h^{n}}\exp\Par{-\int_{0}^{h}\frac{t(h-t)}{h}\tr\Phi(\gamma_{s},t)dt}}{\Abs{D\ham_{c(s),h}(v)}^{-1}}\leq\exp\Par{\frac{(h^{2}R_{1})^{2}}{10}}.
\]
Using the conditions on the step size $h$, we can show that for $C\defeq1+\frac{1}{10^{3}}\tilde{\ell}$
\begin{align*}
\exp\Par{\frac{(h^{2}R_{1})^{2}}{10}} & \leq1+2\frac{(h^{2}R_{1})^{2}}{10}\leq1+2\min\Par{\frac{1}{10^{10}},\frac{\ell_{0}}{10^{3}\ell_{1}}}\leq C.
\end{align*}
Thus, the ratio is bounded below by $C^{-1}$ and above by $C$, and
it implies that 
\begin{equation}
C^{-1}\cdot p_{c(s)}^{0}(v)\leq\tilde{p}_{c(s)}^{0}(v)\leq C\cdot p_{c(s)}^{0}(v).\label{eq:relateApproximate}
\end{equation}
By Lemma \ref{lem:extendGeo}, for each $v_{x}\in V_{x}$ with $\ham_{x,h}(v_{x})=z$
there is a one-to-one correspondence between $v_{x}$ and $v_{y}$,
where $v_{y}$ satisfies $\ham_{y,h}(v_{y})=z$. For this $v_{y}$,
(\ref{eq:relateApproximate}) leads to 
\begin{align}
p_{x}^{0}(v_{x})-p_{y}^{0}(v_{y}) & \leq C\cdot\tilde{p}_{x}^{0}(v_{x})-C^{-1}\cdot\tilde{p}_{y}^{0}(v_{y})\nonumber \\
 & =\Par{C^{2}-1}C^{-1}\tilde{p}_{x}^{0}(v_{x})+C^{-1}\Par{\tilde{p}_{x}^{0}(v_{x})-\tilde{p}_{y}^{0}(v_{y})}\nonumber \\
 & \leq\Par{C^{2}-1}p_{x}^{0}(v_{x})+C^{-1}\Par{\tilde{p}_{x}^{0}(v_{x})-\tilde{p}_{y}^{0}(v_{y})}.\label{eq:ineq2}
\end{align}
In a similar way, we can show that 
\begin{equation}
\Par{C^{-2}-1}p_{x}^{0}(v_{x})+C\Par{\tilde{p}_{x}^{0}(v_{x})-\tilde{p}_{y}^{0}(v_{y})}\leq p_{x}^{0}(v_{x})-p_{y}^{0}(v_{y}).\label{eq:ineq3}
\end{equation}
Using this,
\begin{align}
p_{x}(z)-p_{y}(z) & =\sum_{v_{x}:\ham_{x,h}(v_{x})=z}p_{x}^{0}(v_{x})-\sum_{v_{y}:\ham_{y,h}(v_{y})=z}p_{y}^{0}(v_{y})\nonumber \\
 & \leq\sum_{v_{x}\notin V_{x}:\ham_{x,h}(v_{x})=z}p_{x}^{0}(v_{x})+\sum_{v_{x}\in V_{x}:\ham_{x,h}(v_{x})=z}\Par{p_{x}^{0}(v_{x})-p_{y}^{0}(v_{y})},\label{eq:ineq1}
\end{align}
where in the inequality we only left $v_{y}$ such that $\ham_{y,h}(v_{y})=z$
and that $v_{y}$ is the counterpart of $v_{x}\in V_{x}$ given by
the one-to-one correspondence.

We now bound the TV distance between $\pcal_{x}$ and $\pcal_{y}$
as follow:
\begin{align}
\dtv(\pcal_{x},\pcal_{y}) & =\half\int\Abs{p_{x}(z)-p_{y}(z)}dz\nonumber \\
 & \underset{\eqref{eq:ineq1}}{\leq}\int_{z}\sum_{v_{x}\notin V_{x}:\ham_{x,h}(v_{x})=z}p_{x}^{0}(v_{x})dz+\int_{z}\sum_{v_{x}\in V_{x}:\ham_{x,h}(v_{x})=z}\Abs{p_{x}^{0}(v_{x})-p_{y}^{0}(v_{y})}dz\nonumber \\
 & \underset{\eqref{eq:ineq2},\,\eqref{eq:ineq3}}{\leq}\pcal_{x}^{*}(V_{x}^{c})+\Par{C^{2}-1}\int_{z}\sum_{v_{x}\in V_{x}:\ham_{x,h}(v_{x})=z}p_{x}^{0}(v_{x})dz\nonumber \\
 & \qquad\qquad+2\int_{z}\sum_{v_{x}\in V_{x}:\ham_{x,h}(v_{x})=z}\Abs{\tilde{p}_{x}^{0}(v_{x})-\tilde{p}_{y}^{0}(v_{y})}dz\nonumber \\
 & \leq\frac{\tilde{\ell}}{100}+\frac{\tilde{\ell}}{100}\int_{V_{x}}p_{x}^{*}(v)dv+2\int_{z}\sum_{v_{x}\in V_{x}:\ham_{x,h}(v_{x})=z}\int_{s}\Abs{\frac{d}{ds}\tilde{p}_{c(s)}^{0}(v_{c(s)})}dsdz\nonumber \\
 & \leq\frac{\tilde{\ell}}{50}+2\int_{s}\underbrace{\int_{z}\sum_{v_{x}\in V_{x}:\ham_{x,h}(v_{x})=z}\Abs{\frac{d}{ds}\tilde{p}_{c(s)}^{0}(v_{c(s)})}dz}_{\defeq F_{s}}ds,\label{eq:tvineq1}
\end{align}
where we used that $\int_{V_{x}}p_{x}^{*}(v)dv\leq1$ in the last
inequality, and $v_{c(s)}$ is the initial velocity at $c(s)$ corresponding
to $v_{x}\in V_{x}$ (via the one-to-one correspondence).

Let us bound $F_{s}$ in terms of the parameters. From direct computation
\[
\frac{d}{ds}\tilde{p}_{c(s)}^{0}(v_{c(s)})=\Par{-\int_{0}^{h}\frac{t(h-t)}{h}\frac{d}{ds}\tr\Phi(\gamma_{s}'(t))dt-\half\frac{d}{ds}\norm{v_{c(s)}}_{c(s)}^{2}}\tilde{p}_{c(s)}^{0}(v_{c(s)}).
\]
Due to $\tilde{p}_{c(s)}^{0}(v_{c(s)})\leq2p_{c(s)}^{0}(v_{c(s)})$,
we have
\[
\Abs{\frac{d}{ds}\tilde{p}_{c(s)}^{0}(v_{c(s)})}\leq2\Par{\Abs{\int_{0}^{h}\frac{t(h-t)}{h}\frac{d}{ds}\tr\Phi(\gamma_{s}'(t))dt}+\half\Abs{\frac{d}{ds}\norm{v_{c(s)}}_{c(s)}^{2}}}p_{c(s)}^{0}(v_{c(s)}).
\]
As $\ell(\gamma_{s})\leq\ell_{0}$ due to Lemma \ref{lem:extendGeo},
it follow from Lemma \ref{lem:R2bound} that 
\begin{align*}
F_{s} & \leq4\int_{z}\sum_{v_{x}\in V_{x}:\ham_{x,h}(v_{x})=z}\Par{\Abs{\int_{0}^{h}\frac{t(h-t)}{h}\frac{d}{ds}\tr\Phi(\gamma_{s}'(t))dt}+\half\Abs{\frac{d}{ds}\norm{v_{c(s)}}_{c(s)}^{2}}}p_{c(s)}^{0}(v_{c(s)})dz\\
 & \leq O\Par{h^{2}R_{2}}\int_{V_{x}}p_{v}^{*}(v)dv+2\int_{z}\sum_{v_{x}\in V_{x}:\ham_{x,h}(v_{x})=z}\Abs{\frac{d}{ds}\norm{v_{c(s)}}_{c(s)}^{2}}p_{c(s)}^{0}(v_{c(s)})dz\\
 & \leq O\Par{h^{2}R_{2}} + 2\underbrace{\int_{ \Brace{v\,:\,\ell\Par{\ham_{c(s),h}(v)}\leq\ell_{0}}  }\Abs{\frac{d}{ds}\norm v_{c(s)}^{2}}p_{c(s)}^{*}(v)dv}_{\defeq S},
\end{align*}
where we used that $\int_{V_{x}}p_{v}^{*}(v)dv\leq1$ again for the
first term and that $\ell(\gamma_s)\leq \ell_0$ as well as the change of variable with $z = \ham_{c(s), h}(v(c(s)))$ for the second term.

We now bound $S$ in terms of $R_{3}$. As $\ell(\gamma_{s})=\ell(\ham_{c(s),h}(v(c(s))))\leq\ell_{0}$,
we use Lemma \ref{lem:R3bound} to show that 
\begin{align*}
S
 & =\E_{\ell(\gamma_{s})\leq\ell_{0}}\Abs{\frac{d}{ds}\norm v_{c(s)}^{2}}\\
 & \leq\frac{2}{h}\E_{\ell(\gamma_{s})\leq\ell_{0}}\Abs{\inner{v,\frac{d}{ds}c(s)}_{c(s)}}+6hR_{3}\E_{\ell(\gamma_{s})\leq\ell_{0}}\norm{\frac{d}{ds}c(s)}_{c(s)}.
\end{align*}
We recall from the proof of Lemma \ref{lem:extendGeo} that $\norm{\frac{d}{ds}c(s)}_{c(s)}\leq1.01$.
In addition to this, as $v$ is a Gaussian vector with respect to
the local metric, $\Abs{\inner{v,\frac{d}{ds}c(s)}_{c(s)}}=O(1)$
with high probability, which easily follows from the standard concentration
inequality for the Gaussian distributions. Therefore,
\[
S\leq O\Par{\frac{1}{h}}+7hR_{3}.
\]
Substituting this back to the inequality for $F_{s}$, we have 
\[
F_{s}\leq O\Par{h^{2}R_{2}+\frac{1}{h}+hR_{3}}.
\]
Putting this to (\ref{eq:tvineq1}), it follows that 
\[
\dtv(\pcal_{x},\pcal_{y})\leq O\Par{h^{2}R_{2}+\frac{1}{h}+hR_{3}}\norm{x-y}_{x}+\frac{\tilde{\ell}}{50}.
\]
Due to $\norm{x-y}_{g(x)}\leq2d_{\phi}(x,y)$ by Lemma \ref{lem:dist_const}
and $\tilde{\ell}\leq\frac{1}{100}$, it follows that 
\[
\dtv(\pcal_{x},\pcal_{y})\leq O\Par{h^{2}R_{2}+\frac{1}{h}+hR_{3}}d_{\phi}(x,y)+\frac{1}{5000}.
\]
\end{proof}

Using this one-step coupling, we can prove Proposition \ref{prop:idealConv}
on the mixing time of the ideal RHMC for a general density $e^{-f}$.
\begin{proof}
Due to the assumptions on the step size $h$, Lemma \ref{lem:onestepIdeal}
implies that if $d_{\phi}(x,y)\lesssim h$, then $\dtv(\pcal_{x},\pcal_{y})\leq\frac{1}{1000}$.
By Proposition \ref{prop:conductance} with $\rho=s=\frac{\epsilon}{2\Lambda}$,
we can obtain the following lower bound on the $s$-conductance:
\[
\Phi_{s}=\Omega\Par{h\psi_{\mcal_{\rho}}}.
\]
By Lemma \ref{lem:tvDecrease}, we have 
\[
\dtv(\pi_{t},\pi)\leq s\Lambda+\Lambda\Par{1-\frac{\Phi_{s}^{2}}{2}}^{t}.
\]
Therefore, it suffices to choose $T=O\Par{\Par{h\psi_{\mcal_{\rho}}}^{-2}\log\frac{1}{\rho}}$
to ensure $\dtv(\pi_{T},\pi)\leq\veps$.
\end{proof}

\section{Convergence rate of discretized RHMC \label{sec:discretizedRHMC}}

We bound the remaining two terms, $\dtv(\oprop_{x}',\pcal_{x})$ in
Section \ref{subsec:couple-ideal-dis} and $\dtv(\oprop_{x},\oprop_{x}')$
in Section \ref{subsec:rej-prob}, obtaining a result on the one-step
coupling of RHMC discretized by a numerical integrator with parameters
$C_{x}$ and $C_{v}$. To analyze the convergence rate of the discretized
RHMC, we define additional parameters.
\begin{defn}
\label{def:Mparameters} Given an auxiliary function $\ell$, a good
region $\mcal_{\rho}$ and step size $h$, we define new parameters
$M_{1},M_{1}^{*},M_{2},M_{2}^{*}$ and $\bar{\ell}_{0},\bar{\ell}_{1},\bar{R}_{1}$.
\begin{itemize}
\item $M_{1}$ is a parameter such that for any $t\in[0,h]$ and any Hamiltonian
curve $\gamma$ starting at $x\in\mcal_{\rho}$ with step size $h$
and $\ell(\gamma)\leq\ell_{0}$
\[
n\leq M_{1}\quad\text{and}\quad\norm{\grad f(\gamma(t))}_{g(x)^{-1}}^{2}\leq M_{1}.
\]
\item $M_{2}$ is a parameter such that for any $t\in[0,h]$ and any two
Hamiltonian curves $\gamma_{1},\gamma_{2}$ starting at $x\in\mcal_{\rho}$
with step size $h$ and $\ell(\gamma_{i})\leq\ell_{0}$ for $i=1,2$
\[
\frac{\norm{\grad f(\gamma_{1}(t))-\grad f(\gamma_{2}(t))}_{g(x)^{-1}}}{\norm{\gamma_{1}(t)-\gamma_{2}(t)}_{x}}\leq M_{2}.
\]
\item Let $\gamma$ be any Hamiltonian curve $\gamma$ starting from $(x,v)\in\mcal_{\rho}\times T_{x}\mcal$
with step size $h$ and $\ell(\gamma)\leq\ell_{0}$. Let $\bx_{j}$'s
be intermediate points produced by a numerical integrator with step
size $h$ and an initial condition $(x,v)$. We define $M_{1}^{*}$
to be the smallest number such that for any $t\in[0,h]$ 
\[
\frac{\Abs{f(\gamma(t))-f(\bx_{j})}}{\norm{\gamma(t)-\bx_{j}}_{x}}\leq\sqrt{M_{1}^{*}}\ \text{for all }j.
\]
We define $M_{2}^{*}$ to be the smallest number such that for any
$t\in[0,h]$ 
\[
\frac{\norm{\grad f(\gamma(t))-\grad f(\bx_{j})}_{g(x)^{-1}}}{\norm{\gamma(t)-\bx_{j}}_{x}}\leq M_{2}^{*}\ \text{for all }j.
\]
\item Let $\overline{\mcal_{\rho}}$ be a convex subset of $\mcal$ that
contains $\bx_{h}$ and $\gamma(h)$. We call an auxiliary function
$\bar{\ell}$ symmetric if $\bar{\ell}(\ham_{x,h}(v))=\bar{\ell}(\ham_{x',h}(-v'))$
for $F_{h}(x,v)=(x',v')$. For a symmetric auxiliary function $\bar{\ell}$,
the parameters $\bar{\ell}_{0},\bar{\ell}_{1}$ and $\bar{R}_{1}$
are defined as in Definition \ref{def:auxiliary} and \ref{def:smoothness}
with $\overline{\mcal_{\rho}}$ in place of $\mcal_{\rho}$.
\end{itemize}
\end{defn}
Note that such $\overline{\mcal_{\rho}}$ always exists, as $\mcal$
is convex. We are now ready to formalize the informal statement on
the convergence rate of RHMC with a sensitive integrator
for a density $e^{-f}$ on the Hessian manifold induced by the highly
self-concordant barrier of $\mcal$.

\begin{restatable}{thmre}{thmdiscGen} \label{thm:discGen} 
Let $\pi$ be a target distribution on a convex set $\mcal\subset\Rn$
and $\Lambda=\sup_{S\subset\mcal}\frac{\pi_{0}(S)}{\pi(S)}$ be the
warmness of the initial distribution $\pi_{0}$. Let $\mcal$ be the
Hessian manifold with its metric induced by the Hessian of a strongly
self-concordant barrier and $\pi_{T}$ the distribution obtained after
$T$ steps of RHMC discretized by a numerical integrator
on $\mcal$. For any $\veps>0$, let $\rho=\frac{\veps}{2\Lambda}$
and $\mcal_{\rho}$ any good region. If step size $h$ guarantees
the sensitivity of the integrator and 
\begin{align*}
h^{2} & \leq\frac{10^{-10}}{\max(R_{1},\bar{R}_{1})},\ h^{5}\leq\frac{\ell_{0}}{10^{3}R_{1}^{2}\ell_{1}},\ h^{3}R_{2}+h^{2}R_{3}\leq1,\ h\leq\frac{1}{10^{10}}\min\Par{1,\frac{\ell_{0}}{\ell_{1}}},\ h^{2}\leq\frac{10^{-10}}{n+\sqrt{M_{1}}+M_{2}},\\
h & C_{x}(x,v)\leq\frac{10^{-10}}{\sqrt{n}},\ h^{2}C_{x}(x,v)\leq10^{-10}\min\Par{1,\frac{\bar{\ell}_{0}}{\bar{\ell}_{1}},\frac{1}{n+\sqrt{M_{1}}+\sqrt{M_{1}^{*}}}},\ h^{2}C_{v}(x,v)\leq\frac{10^{-10}}{\sqrt{n+\sqrt{M_{1}}}}
\end{align*}
for $x\in\mcal_{\rho}$ and $v\in V_{\text{good}}^{x}=\Brace{v\in\Rn:\norm v_{g^{-1}}\leq128\sqrt{n},\ \bar{\ell}(\ham_{x,t}(g(x)^{-1}v))\leq\half\bar{\ell}_{0}}$
(see (\ref{eq:vgood})), where the parameters are defined in Definition
\ref{def:auxiliary}, \ref{def:smoothness} and \ref{def:Mparameters},
then for the isoperimetry $\psi_{\mcal_{\rho}}$ of $\mcal_{\rho}$
there exists $T=O\Par{\Par{h\psi_{\mcal_{\rho}}}^{-2}\log\frac{1}{\rho}}$
such that $\dtv(\pi_{T},\pi)\leq\veps$. 
\end{restatable}

\subsection{Stability via self-concordance} \label{subsec:stability_SC}

We summarize computational lemmas used in coupling one-step distributions
and bounding rejection probability. Going forward, the self-concordance
of $g$ is repetitively used to relate local metrics $g$ at two close
points (see Lemma \ref{lem:sc_facts}). We recall that $(1-\norm{x-y}_{g(x)})^{2}g(x)\preceq g(y)\preceq\frac{1}{(1-\norm{x-y}_{g(x)})^{2}}g(x)$
for the local metric $g$ induced by the Hessian of a self-concordant
barrier when $\norm{x-y}_{g(x)}<c<1$. It implies that the local norm
of a vector with respect to $g(x)$ is within a constant factor of
the local norm with respect to $g(y)$ (and vice versa). Namely, for
a vector $v$ we have $\norm v_{g(x)}\leq O(1)\cdot\norm v_{g(y)}$
and $\norm v_{g(y)}\leq O(1)\cdot\norm v_{g(x)}$. It enables us to
move back and forth between the local metric $g(x)$ and $g(y)$ whenever
$x$ and $y$ are sufficiently close in the local metric $g(x)$ or
$g(y)$.
\begin{lem}
\label{lem:sc_facts} Let $g(x)=\nabla^{2}\phi(x)$ for some highly
self-concordant barrier $\phi$.
\begin{itemize}
\item $(1-\|y-x\|_{g(x)})^{2}g(x)\preceq g(y)\preceq\frac{1}{(1-\|y-x\|_{g(x)})^{2}}g(x).$ 
\item $\|Dg(x)[v,v]\|_{g(x)^{-1}}\leq2\|v\|_{g(x)}^{2}.$ 
\item $\|Dg(x)[v,v]-Dg(y)[v,v]\|_{g(x)^{-1}}\leq\frac{6}{(1-\|y-x\|_{g(x)})^{3}}\|v\|_{g(x)}^{2}\|y-x\|_{g(x)}.$ 
\item $\norm{Dg(x)[v,v]-Dg(x)[w,w]}_{g(x)^{-1}}\leq2\norm{v-w}_{g(x)}\norm{v+w}_{g(x)}.$ 
\end{itemize}
\end{lem}

\begin{proof}
The first fact follows from Theorem 4.1.6 in \cite{nesterov2003introductory}.
The second fact follows from Lemma 4.1.2 in \cite{nesterov2003introductory}.
To be precise, 
\[
\norm{Dg(x)[v,v]}_{g(x)^{-1}}=\max_{\norm v_{g(x)}=1}Dg(x)[v,v,u]\leq2\norm v_{g(x)}^{2}.
\]

The third fact is from the following calculation: 
\begin{align*}
 & \|Dg(y)[v,v]-Dg(x)[v,v]\|_{g(x)^{-1}}\\
\leq & \int_{0}^{1}\|D^{2}g(x+t(y-x))[v,v,y-x]\|_{g(x)^{-1}}dt\\
\leq & \int_{0}^{1}\frac{1}{1-t\|y-x\|_{g(x)}}\|D^{2}g(x+t(y-x))[v,v,y-x]\|_{g(x+t(y-x))^{-1}}dt\\
\leq & \int_{0}^{1}\frac{6}{1-t\|y-x\|_{g(x)}}\|v\|_{g(x+t(y-x))}^{2}\|y-x\|_{g(x+t(y-x))}dt\\
\leq & \int_{0}^{1}\frac{6}{(1-t\|y-x\|_{g(x)})^{4}}dt\cdot\|v\|_{g(x)}^{2}\|y-x\|_{g(x)}\\
\leq & \frac{6}{(1-\|y-x\|_{g(x)})^{3}}\|v\|_{g(x)}^{2}\|y-x\|_{g(x)}.
\end{align*}
where the third and fifth line above follow from the first fact, and
the fourth line follows from Proposition 9.1.1 in \cite{nesterov1994interior}.

The fourth fact is from the following calculation: 
\begin{align*}
 & \norm{Dg(x)[v,v]-Dg(x)[w,w]}_{g(x)^{-1}}\\
= & \max_{\norm u_{g(x)}=1}Dg(x)[v,v,u]-Dg(x)[w,w,u]\\
= & \max_{\norm u_{g(x)}=1}Dg(x)[v-w,v,u]+Dg(x)[w,v-w,u]\\
= & \max_{\norm u_{g(x)}=1}Dg(x)[v-w,v,u]+Dg(x)[v-w,w,u]\\
= & \max_{\norm u_{g(x)}=1}Dg(x)[v-w,v+w,u]\\
\leq & \max_{\norm u_{g(x)}=1}2\norm{v-w}_{g(x)}\norm{v+w}_{g(x)}\norm u_{g(x)}\\
\leq & 2\norm{v-w}_{g(x)}\norm{v+w}_{g(x)}.
\end{align*}
\end{proof}
\begin{lem}
\label{lem:computation_sc} For $x,x'\in\mcal$, let $g=g(x)$ and
$g'=g(x')$. Let $\dx:=\norm{x-x'}_{g}<0.99$ and $\dv:=\norm{v-v'}_{g^{-1}}$.
\begin{enumerate}
\item $(1-O(\dx))g\preceq g'\preceq(1+O(\dx))g.$ 
\item $(1-O(\dx))g'\preceq g\preceq(1+O(\dx))g'.$ 
\item $(1-O(\dx))g^{-1}\preceq g'^{-1}\preceq(1+O(\dx))g^{-1}.$ 
\item $(1-O(\dx))g'^{-1}\preceq g^{-1}\preceq(1+O(\dx))g'^{-1}.$ 
\item $-O(\dx)I\preceq I-g^{\half}g'^{-1}g^{\half}\preceq O(\dx)I.$ 
\item $-O(\dx)I\preceq I-g'^{\half}g^{-1}g'^{\half}\preceq O(\dx)I.$ 
\item $\norm{g'^{-\half}g^{\half}}_{2}\leq1+O(\dx)\quad\&\quad\norm{g'^{\half}g^{-\half}}_{2}\leq1+O(\dx).$ 
\item $\norm{g^{\half}g'^{-\half}}_{2}\leq1+O(\dx)\quad\&\quad\norm{g^{-\half}g'^{\half}}_{2}\leq1+O(\dx).$ 
\item $\norm{(\ginverse-g'^{-1})p}_{g}\lesssim\dx\norm p_{\ginverse}.$ 
\item $\norm{\ginverse p-g'^{-1}q}_{g}\leq\norm{p-q}_{\ginverse}+O(\dx)\norm q_{\ginverse}.$ 
\item $\norm{\frac{\del H}{\del v}(x,v)-\frac{\del H}{\del v}(x',v')}_{g}\leq\dv+O(\dx)\norm{v'}_{\ginverse}.$ 
\item $\norm{\frac{\del H}{\del x}(x,v)-\frac{\del H}{\del x}(x',v')}_{\ginverse}\lesssim(\dv+\dx\norm v_{\ginverse})(\norm v_{\ginverse}+\norm{v'}_{\ginverse})+n\delta_{x}+\norm{\grad f(x)-\nabla f(x')}_{g^{-1}}.$
\end{enumerate}
\end{lem}

\begin{proof}
The first four lemmas follow from Lemma \ref{lem:sc_facts}-1. For
5 (and 6, 7, 8 similarly), using 3 
\[
(1-O(\dx))I\preceq g^{\half}g'^{-1}g^{\half}\preceq(1+O(\dx))I.
\]
Thus $-O(\dx)I\preceq I-g^{\half}g'^{-1}g^{\half}\preceq O(\dx)I$.
Also by the definition of two-norm, it follows that 
\[
\norm{g'^{-\half}g^{\half}}_{2}\leq1+O(\dx).
\]
Fact 9 follows from the following computation:
\begin{align*}
\norm{(\ginverse-g'^{-1})p}_{g} & =\norm{(I-g^{\half}g'^{-1}g^{\half})g^{-\half}p}_{2}\leq O(\dx)\norm p_{\ginverse}.\quad(\text{Fact 5})
\end{align*}
Fact 10 follows from the following computation: 
\begin{align*}
\norm{\ginverse p-g'^{-1}q}_{g} & \leq\norm{\ginverse(p-q)+(\ginverse-g'^{-1})q}_{g}\leq\norm{p-q}_{\ginverse}+\underbrace{O(\dx)\norm q_{\ginverse}}_{\text{Fact 9}}.
\end{align*}
Fact 11 follows from the following computation and Fact 10: 
\begin{align*}
\norm{\frac{\del H}{\del v}(x,v)-\frac{\del H}{\del v}(x',v')}_{g} & =\norm{g^{-1}v-g'^{-1}v'}_{g}\leq\dv+O(\dx)\norm{v'}_{\ginverse}.
\end{align*}
For Fact 12, we note that 
\begin{align*}
\frac{\del H}{\del x}(x,v)-\frac{\del H}{\del x}(x',v') & =\Par{\grad f(x)-\grad f(x')}\\
 & \ \ -\half\Par{Dg\Brack{\ginverse v,\ginverse v}-Dg'\Brack{g'^{-1}v,g'^{-1}v}}+\half\Par{\tr(g^{-1}Dg)-\tr(g'^{-1}Dg')}\\
 & =-\half\Par{\underbrace{Dg\Brack{\ginverse v,\ginverse v}-Dg'\Brack{\ginverse v,\ginverse v}}_{F}+\underbrace{Dg'\Brack{\ginverse v,\ginverse v}-Dg'\Brack{g'^{-1}v,g'^{-1}v}}_{S}}\\
 & ~~+\half\Par{\underbrace{\tr(g^{-1}Dg-g'^{-1}Dg)}_{T}+\underbrace{\tr(g'^{-1}Dg-g'^{-1}Dg')}_{R}}+\Par{\grad f(x)-\grad f(x')}.
\end{align*}
For $F$, by the third fact in Lemma \ref{lem:sc_facts} 
\begin{align*}
\norm F_{g^{-1}} & \lesssim\frac{1}{(1-\dx)^{3}}\norm{g^{-1}v}_{g}^{2}\norm{x-x'}_{g}=\frac{1}{(1-\dx)^{3}}\norm v_{\ginverse}^{2}\dx\lesssim\dx\norm v_{\ginverse}^{2}.
\end{align*}
For $S$, by the fourth fact in Lemma \ref{lem:sc_facts} 
\begin{align*}
\norm S_{g^{-1}} & \lesssim\norm S_{g'^{-1}}\lesssim\norm{g^{-1}v-g'^{-1}v'}_{g'}\norm{g^{-1}v+g'^{-1}v'}_{g'}\\
 & \lesssim\Par{\norm{v-v'}_{g'^{-1}}+O(\dx)\norm v_{g'^{-1}}}\Par{\norm v_{g'^{-1}}+\norm{v'}_{g'^{-1}}}\\
 & \lesssim\Par{\dv+\dx\norm v_{g^{-1}}}\Par{\norm v_{g'^{-1}}+\norm{v'}_{g'^{-1}}}.
\end{align*}
For $T$, using the stochastic estimator of trace 
\begin{align*}
\norm{\tr\Par{g^{-1}Dg-g'^{-1}Dg}}_{\ginverse} & =\max_{\norm u_{g}=1}\tr\Par{(\ginverse-g'^{-1})Dg[u]}\\
 & =\max_{\norm u_{g}=1}\tr\Par{g^{\half}(\ginverse-g'^{-1})Dg[u]g^{-\half}}\\
 & =\max_{\norm u_{g}=1}\E_{z\sim\ncal(0,I)}\Brack{z^{\top}g^{\half}(g^{-1}-g'^{-1})Dg[u]g^{-\half}z}\\
 & =\max_{\norm u_{g}=1}\E Dg\Brack{u,g^{-\half}z,(g^{-1}-g'^{-1})g^{\half}z}\\
 & \leq2\E\max_{\norm u_{g}=1}\norm u_{g}\norm{g^{-\half}z}_{g}\underbrace{\norm{\Par{g^{-1}-g'^{-1}}g^{\half}z}_{g}}_{\text{Fact 9}}\\
 & \leq O(\dx)\E\norm z_{2}\norm{g^{\half}z}_{g^{-1}}=O(\dx)\E\norm z_{2}^{2}\\
 & =O(n\dx).
\end{align*}
For $R$, in a similar way that we bounded $\norm T_{g^{-1}}$ 
\begin{align*}
\norm{\tr\Par{g'^{-1}Dg-g'^{-1}Dg'}}_{\ginverse} & =\norm{\E_{z\sim\ncal(0,I)}\Par{Dg[g'^{-\half}z,g'^{-\half}z]-Dg'[g'^{-\half}z,g'^{-\half}z]}}_{g^{-1}}\\
 & \leq\E\underbrace{\ginvnorm{Dg[g'^{-\half}z,g'^{-\half}z]-Dg'[g'^{-\half}z,g'^{-\half}z]}}_{\text{Use Lemma~\ref{lem:sc_facts}}}\\
 & \leq O(\dx)\E\norm{g'^{-\half}z}_{g'}^{2}\leq O(\dx)\E\norm z_{2}^{2}\\
 & =O(n\dx).
\end{align*}
By adding up these bounds, we obtain
\begin{equation*}
\norm{\frac{\del H}{\del x}(x,v)-\frac{\del H}{\del x}(x',v')}_{\ginverse}\lesssim\Par{\dv+\dx\norm v_{\ginverse}}\Par{\norm v_{\ginverse}+\norm{v'}_{\ginverse}}+n\dx+\norm{\grad f(x)-\nabla f(x')}_{g^{-1}}.
\end{equation*}

\end{proof}
We now bound the partial derivatives of $H$ with respect to $x$
and $v$. For $H_{1}$ and $H_{2}$ given by 
\[
H_{1}(x,v)=f(x)+\half\log\det g(x)\ \text{ and }\ H_{2}(x,v)=\frac{1}{2}v^{\top}g(x)^{-1}v,
\]
we recall from \eqref{eq:hmc_intro} that 
\begin{align*}
\frac{\del H_{1}}{\partial x}(x,v) & =\nabla f(x)+\frac{1}{2}\tr(g^{-1}Dg),\\
\frac{\partial H_{2}}{\partial x}(x,v) & =-\frac{1}{2}Dg[g^{-1}v,g^{-1}v]\ \text{ and }\ \frac{\partial H_{2}}{\partial v}(x,v)=\ginverse v.
\end{align*}

\begin{lem}
\label{lem:partH_dist} For $x\in\mcal$ and $g:=g(x)$, the following
inequalities hold. 
\begin{align*}
\norm{\frac{\partial H_{1}(x,v)}{\partial x}}_{g^{-1}} & \leq\norm{\grad f(x)}_{g^{-1}}+n,\\
\norm{\frac{\partial H_{2}(x,v)}{\partial v}}_{g}\leq\norm v_{g^{-1}}\quad & \&\quad\norm{\frac{\partial H_{2}(x,v)}{\partial x}}_{g^{-1}}\leq\norm v_{g^{-1}}^{2}
\end{align*}
\end{lem}

\begin{proof}
For $\frac{\del H_{1}(x,v)}{\partial x}$, 
\begin{align*}
\norm{\frac{\del H_{1}(x,v)}{\partial x}}_{\ginverse} & \leq\norm{\nabla f(x)+\frac{1}{2}\tr(\ginverse Dg)}_{g^{-1}}\leq\norm{\grad f(x)}_{g^{-1}}+\norm{\frac{1}{2}\tr(g^{-1}Dg)}_{g^{-1}}.
\end{align*}
Note that 
\begin{align*}
\norm{\frac{1}{2}\tr(g^{-1}Dg)}_{g^{-1}} & =\frac{1}{2}\max_{\norm u_{g}=1}\tr(g^{-1}Dg[u])
\end{align*}
By self-concordance, for any $h\in\Rn$ we have $h^{\top}Dg[u]h\leq2\norm h_{g}^{2}$
and thus $Dg[u]\preceq2g$, resulting in $g^{-\half}Dg[u]g^{-\half}\preceq2I$.
Then 
\begin{align*}
\tr(g^{-1}Dg[u]) & \leq2\tr(I)\leq2n.
\end{align*}
For $\frac{\del H_{2}(x,v)}{\partial v}$, 
\begin{align*}
\norm{\frac{\del H_{2}(x,v)}{\partial v}}_{g} & =\norm{g^{-1}v}_{g}=\norm v_{g^{-1}}.
\end{align*}
For $\frac{\partial{H}_{2}(x,v)}{\partial x}$, 
\begin{align*}
\norm{\frac{\del H_{2}(x,v)}{\partial x}}_{g^{-1}} & \leq\norm{\frac{1}{2}Dg\Brack{g^{-1}v,g^{-1}v}}_{g^{-1}}\leq\norm{g^{-1}v}_{g}^{2}=\norm v_{g^{-1}}^{2},
\end{align*}
where the second step follows from Lemma \ref{lem:sc_facts}.
\end{proof}

\subsection{Coupling between ideal and discretized RHMC\label{subsec:couple-ideal-dis}}

We bound $\dtv(\oprop_{x}',\,\pcal_{x})$, the TV distance between
the one-step distributions of the ideal RHMC and the discretized RHMC
without the rejection step. We use $\oprop_{x}$ to indicate $\oprop_{x}'$
for simplicity in this section only. We denote by $p_{x}$ and $\overline{p}_{x}$
the probability density functions of $\pcal_{x}$ and $\oprop_{x}$
respectively. We let $g=g(x)$ and $g_{t}=g(x_{t})$.

Let us elaborate on our approach. We work with the Euclidean metric
this time, as we find it easier to handle numerical integrators with
the Euclidean representation. As mentioned in (\ref{eq:onestepEuclidean}),
the one-step distributions $\pcal_{x}$ and $\oprop_{x}$ of the ideal
and discretized RHMC on $\mcal$ are the pushforwards by $T_{x}$
and $\bTx$ of the Gaussian distribution of initial velocities on
the tangent space $T_{x}\mcal$. Thus, it follows by the change of
variables that for $z=\Tx(v^{*})=\bTx(v)$ these two probability densities
on the different spaces (one on $\mcal$ and another on $T_{x}\mcal$)
are related as follows. For $p_{x}^{\ast}$ the probability density
function of $\ncal(0,g(x))$,
\[
p_{x}(z)=\sum_{v^{*}:T_{x}(v^{*})=z}\frac{p_{x}^{\ast}(v^{*})}{\Abs{DT_{x}(v^{*})}}\quad\text{and}\quad\bp_{x}(z)=\sum_{v:\bTx(v)=z}\frac{p_{x}^{\ast}(v)}{\Abs{D\bTx(v)}}.
\]

We aim to couple these $v^{*}$ and $v$ on $T_{x}\mcal$. In this
coupling, we can exclude `bad' velocities, as long as such velocities
have small measure. To see this, let $\vbad$ be a set of bad initial
velocities of measure $\varepsilon<1$ and $\vgood$ be the rest.
Assuming a one-to-one correspondence between $v$ and $v^{*}$ for
$v\in\vgood$, we have that for $x\in\mcal_{\rho}$
\begin{align*}
\dtv(\oprop_{x},\pcal_{x}) & =\sup_{A\subset\mcal}\int_{A}(\overline{p}_{x}(z)-p_{x}(z))dz\\
 & \leq\sup_{A\subset\mcal}\int_{A}\Par{\sum_{v:\bTx(v)=z}\frac{p_{x}^{\ast}(v)}{\Abs{D\bTx(v)}}-\sum_{v^{*}:T_{x}(v^{*})=z}\frac{p_{x}^{\ast}(v^{*})}{\Abs{DT_{x}(v^{*})}}}dz\\
 & \leq\int_{\vbad}p_{x}^{*}(v)dv+\sup_{A\subset\mcal}\int_{A}\sum_{v\in\vgood:\bTx(v)=z}\Par{\frac{p_{x}^{\ast}(v)}{\Abs{D\bTx(v)}}-\frac{p_{x}^{\ast}(v^{*})}{\Abs{DT_{x}(v^{*})}}}dz\\
 & =P_{x}^{\ast}(\vbad)+\sup_{A\subset\mcal}\int_{A}\sum_{v\in\vgood:\bTx(v)=z}\frac{p_{x}^{\ast}(v)}{\Abs{D\bTx(v)}}\Par{1-\frac{p_{x}^{\ast}(v^{*})}{p_{x}^{\ast}(v)}\frac{\Abs{D\bTx(v)}}{\Abs{DT_{x}(v^{*})}}}dz,
\end{align*}
where $P_{x}^{\ast}$ is $\ncal(0,g(x))$, and in the third line we
used the one-to-one correspondence between $v$ and $v^{*}$ to pair
them in the summation. If we show that on $v\in\vgood$ the term of
$1-\frac{p_{x}^{\ast}(v^{*})}{p_{x}^{\ast}(v)}\frac{\Abs{D\bTx(v)}}{\Abs{DT_{x}(v^{*})}}$
is bounded by a small constant (say, $\eta$), then 
\begin{align*}
\dtv(\oprop_{x},\pcal_{x}) & \leq P_{x}^{\ast}(\vbad)+\eta\sup_{A\subset\mcal}\int_{A}\sum_{v\in\vgood:\bTx(v)=z}\frac{p_{x}^{\ast}(v)}{\Abs{D\bTx(v)}}dz\\
 & \leq P_{x}^{\ast}(\vbad)+\eta\int_{\vgood}p_{x}^{\ast}(v)dv\\
 & =P_{x}^{\ast}(\vbad)+\eta P_{x}^{\ast}(\vgood)\\
 & \leq\eta+(1-\eta)\varepsilon.
\end{align*}
By taking $\varepsilon$ sufficiently small, we can bound
$\dtv(\oprop_{x},\pcal_{x})$ smaller then $1/10$.

For each $x\in\mcal$, our bad set $\vbad$ of velocities is the union
of the following sets: 
\begin{align*}
V_{1} & =\Brace{v\in\Rn:\norm v_{g^{-1}}>128\sqrt{n}},\\
V_{2} & =\Brace{v\in\Rn:\bar{\ell}(\ham_{x,t}(g(x)^{-1}v))>\half\bar{\ell}_{0}},
\end{align*}
and thus 
\begin{equation}
\vgood=\Brace{v\in\Rn:\norm v_{g^{-1}}\leq128\sqrt{n},\ \bar{\ell}(\ham_{x,t}(g(x)^{-1}v))\leq\half\bar{\ell}_{0}}.\label{eq:vgood}
\end{equation}
We remark that a velocity $v\in\Rn$ should be normalized by $g(x)^{-1}$
before feeding into $\ham_{x,t}$, since the domain $T_{x}\mcal$
of $\ham_{x,t}$ is endowed with the local metric. Since the standard
concentration inequality for the Gaussian distributions implies that
$P_{x}^{*}(V_{1})<\frac{1}{100}$, and the definition $\bar{\ell}_{0}$
implies that $P_{x}^{\ast}(V_{2})<\frac{1}{100}$, it follows that
$P_{x}^{\ast}(\vbad)<0.02$.

\subsubsection{Dynamics of ideal and discretized RHMC}

We study the dynamics of the ideal and discretized RHMC.
\begin{prop}
\label{prop:dynamics} For $x\in\mcal_{\rho}$ and $v,v'\in\vgood$
, let $g:=g(x)$ and $h$ step size satisfying 
\[
\underbrace{h^{2}\leq\frac{10^{-10}}{n+\sqrt{M_{1}}+M_{2}}}_{\onecirc},\underbrace{h^{2}\leq\frac{10^{-10}}{\bar{R}_{1}}}_{\twocirc},\underbrace{C_{x}(x,v)h^{2}\leq\frac{1}{10^{10}}\min\Par{1,\frac{\bar{\ell}_{0}}{\bar{\ell}_{1}}}}_{\threecirc},\underbrace{C_{x}(x,v)h\leq\frac{\sqrt{n}}{10^{10}}}_{\fourcirc}.
\]
For $t\in[0,h]$, we let $(x_{t},v_{t})$ and $(x_{t}',v_{t}')$ be
the Hamiltonian curves of the ideal RHMC at time $t$ with initial
conditions $(x,v)$ and $(x,v')$, respectively. Let $(\overline{x},\overline{v})$
be the point obtained from RHMC with a sensitive second-order numerical
integrator with the step size $h$ and initial condition $(x,v)$.
Let $\phi\defeq\sup_{t\in[0,h]}\norm{x_{t}-x_{t}'}_{g(x_{t})}$, $\psi\defeq\sup_{t\in[0,h]}\norm{v_{t}-v_{t}'}_{g(x_{t})^{-1}}$
and $\Gamma_{t}(v)\defeq g(x_{t})^{-1}(v_{t}-v)$. 
\end{prop}

\begin{enumerate}
\item $\norm{x-x_{t}}_{g}=O\Par{t\sqrt{n}+t^{2}(n+\sqrt{M_{1}})}<\frac{1}{4}$
and $\norm{v-v_{t}}_{g^{-1}}=O\Par{t(n+\sqrt{M_{1}})}$. 
\item $(1-o(1))\norm{v_{h}-v_{h}'}_{g^{-1}}\leq\psi\leq(1+o(1))\norm{v-v'}_{g^{-1}}$.
\item $(1-o(1))\norm{\Tx(v)-\Tx(v')}_{g}\leq\phi\leq(1+o(1))h\psi\leq(1+o(1))h\norm{v-v'}_{g}$.
\item $\norm{\Gamma_{t}(v)-\Gamma_{t}(v')}_{g}\leq L\norm{v-v'}_{g^{-1}}$
for some $L<1/10$.
\item For $z=\bTx(v)$, there exists $v^{*}\in\Rn$ with $z=T_{x}(v^{*})$
such that $\bar{\ell}(\ham_{x,t}(g(x)^{-1}v^{*}))\leq\bar{\ell}_{0}$ and $\norm{v-v^{*}}_{\ginverse}=O\Par{\frac{1}{h}}\norm{T_{x}(v)-\bTx(v)}_{g}$. Moreover,
there is a one-to-one correspondence between $v$ and $v^{*}$.
\end{enumerate}
\textbf{Proof of 1}. For $0\leq t\leq h$, let us define $\phi(t):=\norm{x-x_{t}}_{g}$
and $\psi(t):=\norm{v-v_{t}}_{\ginverse}$. Note that 
\begin{align*}
2\norm{x-x_{t}}_{g}\frac{d\phi(t)}{dt} & =\frac{d\phi^{2}(t)}{dt}=\frac{d}{dt}(x_{t}-x)^{\top}g(x_{t}-x)\\
 & =2\Par{\frac{\partial H}{\partial v}(x_{t},v_{t})}^{\top}g(x_{t}-x).
\end{align*}
Hence, 
\[
\Abs{\norm{x-x_{t}}_{g}\phi'(t)}=\Abs{\Par{\frac{\partial H}{\partial v}(x_{t},v_{t})}^{\top}g(x_{t}-x)}\leq\norm{\frac{\partial H}{\partial v}(x_{t},v_{t})}_{g}\norm{x_{t}-x}_{g},
\]
and $|\phi'(t)|\leq\norm{\frac{\partial H}{\partial v}(x_{t},v_{t})}_{g}$.
When $\norm{x-x_{t}}_{g}<\frac{1}{4}$ for $0\leq t\leq h$, since
the local norms at $x$ and $x_{t}$ are within a small constant factor
as follows, we have 
\begin{align*}
\norm{\frac{\partial H}{\partial v}(x_{t},v_{t})}_{g} & \leq2\norm{\frac{\partial H}{\partial v}(x_{t},v_{t})}_{g_{t}}=2\norm{v_{t}}_{g_{t}^{-1}}\quad(\text{Lemma \ref{lem:partH_dist}})\\
 & \leq4\norm{v_{t}}_{\ginverse}\leq4(\norm{v_{t}-v}_{\ginverse}+\norm v_{\ginverse}),
\end{align*}
and thus 
\begin{equation}
\phi'(t)\leq10^{3}\sqrt{n}+4\psi(t)\quad\text{if }\norm{x-x_{t}}_{g}<\frac{1}{4}.\label{ineq1}
\end{equation}
Similarly, we can obtain 
\begin{align*}
2\norm{v-v_{t}}_{\ginverse}\frac{d\psi}{dt} & =\frac{d\psi^{2}(t)}{dt}=\frac{d}{dt}(v_{t}-v)^{\top}\ginverse(v_{t}-v)\\
 & =-2\Par{\frac{\partial H}{\partial x}(x_{t},v_{t})}^{\top}\ginverse(v_{t}-v),
\end{align*}
and thus $|\psi'(t)|\leq\norm{\frac{\partial H}{\partial x}(x_{t},v_{t})}_{\ginverse}$.
If $\norm{x-x_{t}}_{g}<\frac{1}{4}$ for $0\leq t\leq h$, then by
Lemma \ref{lem:partH_dist} 
\begin{align*}
\norm{\frac{\partial H}{\partial x}(x_{t},v_{t})}_{\ginverse} & \leq2(\norm{v_{t}}_{g_{t}^{-1}}^{2}+\sqrt{M_{1}}+n)\quad(\because x\in\mcal_{\rho})\\
 & \leq2\Par{4\norm{v_{t}}_{\ginverse}^{2}+\sqrt{M_{1}}+n}\\
 & \leq2\Par{8(\norm v_{\ginverse}^{2}+\norm{v_{t}-v}_{\ginverse}^{2})+\sqrt{M_{1}}+n}\quad(\because(a+b)^{2}\leq2(a^{2}+b^{2}))\\
 & \leq10^{6}n+16\psi^{2}(t)+2\sqrt{M_{1}},
\end{align*}
and thus 
\begin{equation}
\psi'(t)\leq10^{6}n+16\psi^{2}(t)+2\sqrt{M_{1}}\quad\text{if }\norm{x-x_{t}}_{g}<\frac{1}{4}.\label{ineq2}
\end{equation}
Now let us solve the coupled inequalities (\ref{ineq1}) and (\ref{ineq2}).
When $\psi(t)\leq1000t(n+\sqrt{M_{1}})$, (\ref{ineq2}) becomes 
\[
\psi'(t)\leq10^{6}n+2\sqrt{M_{1}}+16\cdot10^{6}t^{2}(n+\sqrt{M_{1}})^{2},
\]
and this inequality holds up until $h$ satisfying $\int_{0}^{h}\Par{10^{6}n+2\sqrt{M_{1}}+16\cdot10^{6}t^{2}(n+\sqrt{M_{1}})^{2}}dt\leq1000h(n+\sqrt{M_{1}})$.
We can check that for any $h\leq\frac{10^{-10}}{\sqrt{n+\sqrt{M_{1}}}}$
(i.e., condition $\onecirc$) this integral inequality is satisfied.
Recall that we have to ensure that $\phi(t)<\frac{1}{4}$ for $t\leq h$.
By substituting $\psi(t)\leq1000t(n+\sqrt{M_{1}})$ into (\ref{ineq1}),
we have 
\[
\phi'(t)\leq10^{3}\sqrt{n}+4000t(n+\sqrt{M_{1}}),
\]
as long as $\phi(t)=\norm{x_{t}-x}_{g}<\frac{1}{4}$. It is straightforward
to see that for $t\leq h$ one has 
\begin{align*}
\phi(t) & \leq10^{3}\sqrt{n}t+2000t^{2}(n+\sqrt{M_{1}})\leq10^{-7}\sqrt{\frac{n}{n+\sqrt{M_{1}}}}+\frac{2000}{10^{10}}\\
 & <\frac{1}{10^{6}}.
\end{align*}

\textbf{Proof of 2.} From the first item, $\norm{x_{t}-x_{t}'}_{g}\leq\norm{x-x_{t}}_{g}+\norm{x-x_{t}'}_{g}<10^{-5}$.
Due to Lemma \ref{lem:computation_sc}, we can switch the local norms
among $g,g_{t}=g(x_{t})$ and $g_{t}'=g(x_{t}')$ by losing a multiplicative
constant like $1+10^{-4}$.

For $\dx=\norm{x_{t}-x_{t}'}_{g_{t}}$ and $\dv=\norm{v_{t}-v_{t}'}_{g_{t}^{-1}}$,
\begin{align*}
\norm{v_{t}-v_{t}'}_{g_{t}} & =\norm{v-v'+\int_{0}^{t}\Par{\frac{\del H}{\del x}(x_{s},v_{s})-\frac{\del H}{\del x}(x_{s}',v_{s}')}ds}_{g_{t}^{-1}}\\
 & \leq\norm{v-v'}_{g_{t}^{-1}}+O(h)\sup_{t\in[0,h]}\norm{\frac{\del H}{\del x}(x_{t},v_{t})-\frac{\del H}{\del x}(x_{t}',v_{t}')}_{g_{t}^{-1}}\\
 & \leq\norm{v-v'}_{g_{t}^{-1}}+O(h)\sup_{t\in[0,h]}\Par{(\dv+\dx\norm{v_{t}}_{g^{-1}})(\norm{v_{t}}_{g^{-1}}+\norm{v_{t}'}_{g^{-1}})+(n+M_{2})\dx},
\end{align*}
where the last step follows from Lemma \ref{lem:computation_sc}-12.
By the first item and $\onecirc$, we have $\norm{v_{t}}_{g^{-1}},\norm{v_{t}'}_{g^{-1}}\leq7\sqrt{n+\sqrt{M_{1}}}$
and thus 
\begin{align*}
\norm{v_{t}-v_{t}'}_{g_{t}^{-1}} & \leq\norm{v-v'}_{g_{t}^{-1}}+O(h)\Par{\psi\sqrt{n+\sqrt{M_{1}}}+\phi\Par{n+\sqrt{M_{1}}+M_{2}}}\\
 & \leq(1+o(1))\norm{v-v'}_{g^{-1}}+O(h)\psi\Par{\sqrt{n+\sqrt{M_{1}}}+h\Par{n+\sqrt{M_{1}}+M_{2}}},
\end{align*}
where we used $\phi\leq(1+o(1))O(h)\psi$ that we prove in the next
item. Taking the supremum over $t\in[0,h]$, we obtain 
\[
\Par{1-O(h)\Par{\sqrt{n+\sqrt{M_{1}}}+h\Par{n+\sqrt{M_{1}}+M_{2}}}}\psi\leq(1+o(1))\norm{v-v'}_{g^{-1}}.
\]
Taking a sufficiently small constant in $h$ and using $\onecirc$,
it follows that $\psi\leq(1+o(1))\norm{v-v'}_{g^{-1}}$.

\textbf{Proof of 3. }By Lemma \ref{lem:computation_sc}-11, 
\begin{align*}
\norm{x_{t}-x_{t}'}_{g_{t}} & =\norm{x+\int_{0}^{t}\frac{\del H}{\del v}(x_{s},v_{s})ds-\Par{x+\int_{0}^{t}\frac{\del H}{\del v}(x_{s}',v_{s}')ds}}_{g_{t}}\\
 & \leq O(h)\sup_{t\in[0,h]}\norm{\frac{\del H}{\del v}(x_{t},v_{t})-\frac{\del H}{\del v}(x_{t}',v_{t}')}_{g_{t}}\\
 & \leq O(h)\Par{\psi+\phi\sqrt{n+\sqrt{M_{1}}}}.
\end{align*}
By using $\onecirc$ and taking the supremum over $t\in[0,h]$ and
a sufficiently small constant in $h$, we obtain the inequality of
$\phi\leq(1+o(1))O(h)\psi$ as we promised, and the second item implies
that 
\[
\phi\leq(1+o(1))O(h)\norm{v-v'}_{g^{-1}}.
\]

\textbf{Proof of 4. }By Lemma \ref{lem:computation_sc}-12, 
\begin{align*}
\norm{\Gamma_{t}(v)-\Gamma_{t}(v')}_{g_{t}}= & \norm{g_{t}^{-1}(v_{t}-v)-g_{t}'^{-1}(v_{t}'-v')}_{g_{t}}\\
\leq & \norm{v_{t}-v-v_{t}'+v'}_{g_{t}^{-1}}+O(1)\underbrace{\norm{x_{t}-x_{t}'}_{g_{t}}}_{\leq\phi}\underbrace{\norm{v_{t}'-v'}_{g_{t}^{-1}}}_{\leq O(h(n+\sqrt{M_{1}}))}\\
\leq & \norm{\int_{0}^{t}\Par{\frac{\del H}{\del x}(x_{s},v_{s})-\frac{\del H}{\del x}(x_{s}',v_{s}')}ds}_{g_{t}^{-1}}+O\Par{h\Par{n+\sqrt{M_{1}}}}\phi\\
\leq & O(h)\sup_{t\in[0,h]}\norm{\frac{\del H}{\del x}(x_{t},v_{t})-\frac{\del H}{\del x}(x_{t}',v_{t}')}_{g_{t}^{-1}}+O\Par{h\Par{n+\sqrt{M_{1}}}}\phi.
\end{align*}
We can bound the first term by $O(h^{2})\Par{\sqrt{n+\sqrt{M_{1}}}+h\Par{n+\sqrt{M_{1}}+M_{2}}}\psi$ by following the proof of the second item.
Using the second and third items with the condition $\onecirc$, and
taking a sufficiently small constant in $h$, for some $L<1/10$ 
\begin{align*}
\norm{\Gamma_{t}(v)-\Gamma_{t}(v')}_{g} & \lesssim h^{2}\sqrt{n+\sqrt{M_{1}}}+h^{3}\Par{n+\sqrt{M_{1}}+M_{2}}+h^{2}\Par{n+\sqrt{M_{1}}}\norm{v-v'}_{g^{-1}}\\
 & \leq L\norm{v-v'}_{g^{-1}}.
\end{align*}

\textbf{Proof of 5.} Let $z=\bTx(v)$ for $v\in\vgood$. We show that
the map defined on $u\in V_{\text{dom}}=\{v'\in\Rn:\norm{v-v'}_{g^{-1}}\leq4\sqrt{n}\}$
by 
\[
\Upsilon(u)=u-\frac{1}{h}g\Tx(u)+\frac{1}{h}gz,
\]
is Lipschitz in $u$ with respect to the local norm $g^{-1}$, and
then apply the Banach fixed-point theorem to obtain the unique fixed-point
$v^{*}$. Note that it satisfies $g(\Tx(v^{*})-\bTx(v))=0$ and thus
$\Tx(v^{*})=\bTx(v)$.

For Lipschitzness, let $(x_{t},u_{t})$ and $(x_{t}',u_{t}')$ be
the Hamiltonian curves of the ideal RHMC starting from $(x,u)$ and
$(x,u')$ for $u,u'\in V_{\text{dom}}$, respectively. Observe that
\begin{align}
 & \norm{\Upsilon(u)-\Upsilon(u')}_{g^{-1}}=\norm{u-u'-\frac{1}{h}g\Par{T_{x}(u)-T_{x}(u')}}_{g^{-1}}\label{eq:repeat}\\
 & =\norm{u-u'-\frac{1}{h}\int_{0}^{h}g\Par{g_{t}^{-1}u_{t}-g_{t}^{-1}u_{t}'}dt}_{g^{-1}}\nonumber \\
 & =\norm{\Par{I-\frac{1}{h}g\int_{0}^{h}g_{t}^{-1}dt}u-\Par{I-\frac{1}{h}g\int_{0}^{h}g_{t}'^{-1}dt}u'-\frac{1}{h}\int_{0}^{h}g\Par{\Gamma_{t}(u)-\Gamma_{t}(u')}dt}_{g^{-1}}\nonumber \\
 & =\norm{\underbrace{\frac{1}{h}\Par{\int_{0}^{h}(I-gg_{t}^{-1})dt}}_{I_{u}}u-\underbrace{\frac{1}{h}\Par{\int_{0}^{h}(I-gg_{t}'^{-1})dt}}_{I_{u'}}u'-\frac{1}{h}\int_{0}^{h}g\Par{\Gamma_{t}(u)-\Gamma_{t}(u')}dt}_{g^{-1}}\nonumber \\
 & \leq\norm{I_{u}(u-u')+(I_{u}-I_{u'})u'}_{g^{-1}}+\frac{1}{h}\norm{\int_{0}^{h}g\Par{\Gamma_{t}(u)-\Gamma_{t}(u')}dt}_{g^{-1}}\nonumber \\
 & \leq\norm{I_{u}(u-u')}_{g^{-1}}+\norm{(I_{u}-I_{u'})u'}_{g^{-1}}+\sup_{t\in[0,h]}\norm{\Gamma_{t}(u)-\Gamma_{t}(u')}_{g}\nonumber \\
 & \leq\underbrace{\norm{I_{u}(u-u')}_{g^{-1}}}_{F}+\underbrace{\norm{(I_{u}-I_{u'})u'}_{g^{-1}}}_{S}+L\norm{u-u'}_{g^{-1}},\label{eq:lastbound}
\end{align}
where the last inequality follows from the fourth item.

For $F$, let $p=u-u'$ and observe that 
\begin{align*}
\norm{I_{u}p}_{g^{-1}} & \leq\frac{1}{h}\int_{0}^{h}\norm{(I-gg_{t}^{-1})p}_{g^{-1}}dt\leq\sup_{t\in[0,h]}\norm{(I-gg_{t}^{-1})p}_{g^{-1}}\\
 & \leq O(1)\sup_{t\in[0,h]}\norm{I-g^{\half}g_{t}^{-1}g^{\half}}_{2}\norm p_{g^{-1}}\leq O(1)\sup_{t\in[0,h]}\norm{x-x_{t}}_{g}\norm p_{g^{-1}}\\
 & \lesssim\Par{h\sqrt{n}+h^{2}\Par{n+\sqrt{M_{1}}}}\norm{u-u'}_{g^{-1}}.\quad(\text{First item})
\end{align*}

For $S$, we can bound it as follows:
\begin{align*}
\norm{(I_{u}-I_{u'})u'}_{g^{-1}} & \leq\frac{1}{h}\int_{0}^{h}\norm{g(g_{t}'^{-1}-g_{t}^{-1})u'}_{g^{-1}}dt\leq\sup_{t\in[0,h]}\norm{(g_{t}'^{-1}-g_{t}^{-1})u'}_{g}\\
 & \lesssim\sup_{t\in[0,h]}\norm{(g_{t}'^{-1}-g_{t}^{-1})u'}_{g_{t}'}\lesssim\sup_{t\in[0,h]}\norm{x_{t}-x_{t}'}_{g_{t}}\norm{u'}_{g_{t}'}\\
 & \lesssim\phi\sqrt{n+\sqrt{M_{1}}}\lesssim h\sqrt{n+\sqrt{M_{1}}}\norm{u-u'}_{g^{-1}}.\quad(\text{Third item})
\end{align*}
Substituting the bounds on $F$ and $S$ into (\ref{eq:lastbound})
with a small constant in $h$ taken, we can conclude that 
\[
\norm{\Upsilon(u)-\Upsilon(u')}_{g^{-1}}<\frac{1}{3}\norm{u-u'}_{g^{-1}}.
\]

Next, we show that the image of $\Upsilon$ is included in $V_{\text{dom}}$.
For $u\in V_{\text{dom}}$,
\begin{align*}
\norm{\Upsilon(u)-v}_{g^{-1}} & =\ginvnorm{u-v-\frac{1}{h}g(\Tx(u)-\bTx(v))}\\
 & =\ginvnorm{u-v-\frac{1}{h}g(\Tx(u)-\Tx(v))+\frac{1}{h}g(\Tx(v)-\bTx(v))}\\
 & \leq\norm{u-v-\frac{1}{h}g(\Tx(u)-\Tx(v))}_{g^{-1}}+\frac{1}{h}\norm{\Tx(v)-\bTx(v)}_{g}.
\end{align*}
Repeating the proof for the first item\footnote{For $u\in V_{\text{dom}}$, we might have $\norm u_{g^{-1}}\geq128\sqrt{n}$
though, it is still bounded above by $132\sqrt{n}$. The proofs of
the second to fifth items can be exactly reproduced for $V_{\text{relaxed}}\defeq\{v'\in\Rn:\norm{v'}_{g^{-1}}\leq132\sqrt{n}\}$,
leading to a similar conclusion like $\norm{\Upsilon(u)-\Upsilon(u')}_{g^{-1}}<(\frac{1}{3}+\epsilon)\norm{u-u'}_{g^{-1}}$
for a small constant $\epsilon>0$.} (see (\ref{eq:repeat})), we can bound the first term by $\half\norm{u-v}_{g^{-1}}$
and thus by $2\sqrt{n}$, due to $u\in V_{\text{dom}}$. By Lemma
\ref{lem:dist_const} and $\fourcirc$, we can bound the second term
by 
\begin{align*}
\frac{1}{h}\norm{\Tx(v)-\bTx(v)}_{g} & \leq\frac{2}{h}d_{g}(\Tx(v),\bTx(v))\leq\frac{2}{h}C_{x}(x,v)h^{2}\leq2\sqrt{n}.
\end{align*}
Putting them together, we obtain $\norm{\Upsilon(u)-v}_{g^{-1}}\leq4\sqrt{n}$.

By the Banach fixed-point theorem, there is a unique fixed point $v^{*}$
of $\Upsilon$ such that $\Tx(v^{*})=\bTx(v)$ and 
\[
\norm{\Upsilon(v)-v^{*}}_{g^{-1}}<\frac{1}{3}\norm{v-v^{*}}_{g^{-1}}.
\]
Moreover, $\norm{\Upsilon(v)-v^{*}}_{g^{-1}}=\norm{v-v^{*}-\frac{1}{h}g(\Tx(v)-z)}_{g^{-1}}\geq\norm{v-v^{*}}_{g^{-1}}-\frac{1}{h}\norm{\Tx(v)-\bTx(v)}_{g}$.
Relating these two inequalities, we obtain 
\[
\norm{v-v^{*}}_{g^{-1}}\lesssim\frac{1}{h}\norm{\Tx(v)-\bTx(v)}_{g^{-1}}.
\]

 We now show a one-to-one correspondence between $v\in\vgood$ and
$v^{*}$. Let $F_{h}(x,v)=(x_{2},v_{2})$ and $F_{h}(x,v^{*})=(z,v')$.
By the reversibility of the Hamiltonian trajectories, we have a one-to-one
correspondence between $v^{*}$ and $v'$ in a sense that $F_{h}(x,v^{*})=(z,v')$
and $F_{h}(z,-v')=(x,-v^{*})$. Similarly, we also have a one-to-one
correspondence between $v$ and $v_{2}$. Thus, it suffices to show
a one-to-one correspondence between $v_{2}\in T_{x_{2}}\mcal$ and
$v'\in T_{z}\mcal$. 

Consider the straight line between $z$ and $x_{2}$. We have that
$\norm{z-x_{2}}_{g}=\norm{\bTx(v)-\Tx(v)}_{g}\leq2C_{x}(x,v)h^{2}\leq10^{-9}\min\Par{1,\frac{\bar{\ell}_{0}}{\bar{\ell}_{1}}}$
by $\threecirc$ and $\bar{\ell}(\ham_{x_{2},h}(-g(x_{2})^{-1}v_{2}))=\bar{\ell}(\ham_{x,h}(g(x)^{-1}v))\leq\bar{\ell}_{0}/2$
by the symmetry of $\bar{\ell}$. Due to $\twocirc$, we can apply
Lemma \ref{lem:extendGeo} to $x_{2}$ with an initial velocity $-g(x_{2})^{-1}v_{2}$.
Thus, a one-to-one correspondence between $v'$ and $v_{2}$ follows,
and we also have that $\bar{\ell}_{0}\geq\bar{\ell}(\ham_{z,h}(-g(z)^{-1}v'))=\bar{\ell}(\ham_{x,h}(g(x)^{-1}v^{*})$.

\subsubsection{One-step coupling}

As elaborated in Section \ref{subsec:couple-ideal-dis}, it suffices
to prove that for $v\in\vgood$ the term of $1-\frac{p_{x}^{\ast}(v^{*})}{p_{x}^{\ast}(v)}\frac{\Abs{D\bTx(v)}}{\Abs{DT_{x}(v^{*})}}$
is bounded by a constant smaller than $1$.
\begin{lem}
\label{lem:tvidealdis} For $x\in\mcal_{\rho}$ and $v\in\vgood$,
let step size $h$ guarantee the sensitivity of a numerical
integrator at $(x,v)$, and satisfy
\[
\underbrace{h^{2}\leq\frac{10^{-10}}{n+\sqrt{M_{1}}+M_{2}}}_{\onecirc},\underbrace{h^{2}\leq\frac{10^{-10}}{\bar{R}_{1}}}_{\twocirc},\underbrace{C_{x}(x,v)h^{2}\leq\frac{1}{10^{10}}\min\Par{1,\frac{\bar{\ell}_{0}}{\bar{\ell}_{1}}}}_{\threecirc},\underbrace{C_{x}(x,v)h\leq\frac{1}{10^{10}\sqrt{n}}}_{\fourcirc}.
\]
Then $\dtv(\oprop_{x},\pcal_{x})\leq\frac{1}{10}$.
\end{lem}

\begin{proof}
For given $v\in\vgood$, Proposition \ref{prop:dynamics}-5 ($\onecirc\sim\fourcirc$
required) and the order of the numerical integrator ensure that there
exists $v^{*}\in T_{x}\mcal$ such that $T_{x}(v^{*})=\overline{T}_{x}(v)$
and $\norm{v-v^{*}}_{g^{-1}}\lesssim C_{x}(x,v)h\leq\frac{1}{10^{10}\sqrt{n}}$
by $\fourcirc$. As $p_{x}^{*}$ is the probability density function
of $\ncal(0,g(x))$, 
\[
\Abs{\log\Par{\frac{p_{x}^{\ast}(v^{*})}{p_{x}^{\ast}(v)}}}=\Abs{\norm{v^{*}}_{g^{-1}}^{2}-\norm v_{g^{-1}}^{2}}\leq\norm{v^{*}-v}_{g^{-1}}\Par{\norm v_{g^{-1}}+\norm{v^{*}}_{g^{-1}}}\leq\frac{1}{10^{5}},
\]
and thus the ratio of $\frac{p_{x}^{\ast}(v^{*})}{p_{x}^{\ast}(v)}$
is bounded below by $0.999$. Also, the sensitivity of the numerical
integrator yields $\frac{\Abs{D\bTx(v)}}{\Abs{DT_{x}(v^{*})}}\geq0.998$.
Hence, for any $v\in\vgood$ 
\[
1-\frac{p_{x}^{\ast}(v^{*})}{p_{x}^{\ast}(v)}\frac{\Abs{D\bTx(v)}}{\Abs{DT_{x}(v^{*})}}\leq0.003,
\]
and the claim follows.
\end{proof}
We finish this section by providing a sufficient condition on the
step size for the sensitivity of a numerical integrator, which we find
useful later when checking the sensitivity of IMM and LM.
\begin{prop}
\label{prop:stability} Let $x\in\mcal_{\rho}$ and $v\in\vgood$.
Let step size $h$ satisfy $h^{2}\leq\frac{1}{10^{5}\sqrt{n}\bar{R}_{1}}$
in addition to the step-size conditions in Proposition \ref{prop:dynamics}.
A numerical integrator $\overline{T}_{x,h}$ is sensitive at $(x,v)$
if $\Abs{D\bTx(v)}\geq\frac{(1-10^{-6})h^{n}}{\sqrt{\Abs{g(x')}\Abs{g(x)}}}$
for $x'=\bTx(v)$.
\end{prop}

\begin{proof}
By Proposition \ref{prop:dynamics}-5, there exists $v^{*}$ such
that $\Tx(v^{*})=\bTx(v)$ and $\bar{\ell}(\ham_{x,t}(g(x)^{-1}v^{*}))\leq\bar{\ell}_{0}$.
Let us estimate $\Abs{D\Tx(v^{*})}$. Recall that $\ham_{x,h}$ is
the Hamiltonian map from $T_{x}\mcal$ to $\mcal$, where both spaces
are endowed with the local metric $g$. Even though $\Tx$ has the
same domain and range, these spaces are endowed with the Euclidean
metric. Therefore, we can relate $\Tx$ to $\ham_{x,h}$ by 
\[
\Tx(v^{*})=\Par{\idftn_{\mcal\to\Rn}\circ\ham_{x,h}\circ\idftn_{\Rn\to T_{x}\mcal}}(g(x)^{-1}v^{*}),
\]
where $\idftn_{\mcal\to\Rn}$ is the embedding with transition of
metric from $g(x)$ to the Euclidean, and $\idftn_{\R^{n}\to T_{x}\mcal}$
is the embedding with transition of metric from the Euclidean to $g(x)$.
Note that we have to normalize $v^{*}$ by $g(x)^{-1}$ before $\ham_{x,h}$
takes it as input. Using this formula and the chain rule, 
\begin{align*}
\Abs{D\Tx(v^{*})} & =\Abs{D\idftn_{\mcal\to\R^{n}}(x')}\Abs{D\ham_{x,h}(g(x)^{-1}v^{*})}\Abs{D\idftn_{\R^{n}\to T_{x}\mcal}(g(x)^{-1}v^{*})}\\
 & \leq\Abs{g(x')}^{-\half}\cdot h^{n}\Par{1+\frac{2}{1000}}\cdot\Abs{g(x)}^{-1}\cdot\Abs{g(x)}^{\half}\\
 & =\frac{h^{n}}{\sqrt{\Abs{g(x')}\Abs{g(x)}}}\Par{1+\frac{2}{1000}},
\end{align*}
where we used Corollary \ref{lem:jacoIdeal} for $\bar{\ell}$ in
the second line. Hence, if $\Abs{D\bTx(v)}\geq\frac{(1-10^{-6})h^{n}}{\sqrt{\Abs{g(x')}\Abs{g(x)}}}$,
then 
\[
\frac{\Abs{D\bTx(v)}}{\Abs{D\Tx(v^{*})}}\geq\frac{1-10^{-6}}{1+0.002}\geq0.998.
\]
\end{proof}

\subsection{Bound on rejection probability\label{subsec:rej-prob}}

Lastly, we bound the rejection probability $\dtv(\oprop_{x}',\oprop_{x})$.
\begin{lem}
\label{lem:rejprob} For $x,x'\in\mcal$, let $g=g(x)$ and $g'=g(x')$.
If for some $0<\delta_{x}<1$ and $0<\delta_{v}$ we have $\norm{x-x'}_{g}\leq\delta_{x}$
and $\norm{v-v'}_{g^{-1}}\leq\delta_{v}$, then 
\[
\Abs{-H(x',v')+H(x,v)}\leq\Theta\Par{\Abs{f(x)-f(x')}+\Par{\dx\norm v_{\ginverse}^{2}+\dv^{2}+\dv\norm v_{\ginverse}}+n\delta_{x}},
\]
where the Hamiltonian is $H(x,v)=f(x)+\frac{1}{2}v^{\top}\ginv xv+\frac{1}{2}\log\det g(x)$.
\end{lem}

\begin{proof}
We consider each term separately. For the second term, 
\begin{align*}
\Abs{\frac{1}{2}v^{\top}\ginverse v-\frac{1}{2}v'^{\top}g'^{-1}v} & \leq\half\underbrace{\Abs{v^{\top}g^{-1}v-v^{\top}g'^{-1}v}}_{F}+\half\underbrace{\Abs{v^{\top}g'^{-1}v-v'^{\top}g'^{-1}v'}}_{S}.
\end{align*}
For $F$, we have $F\leq O(\dx)\norm v_{\ginverse}^{2}$ by Lemma
\ref{lem:computation_sc}-3. For $S$, it follows that 
\begin{align*}
S & =\Abs{\norm v_{g'^{-1}}^{2}-\norm{v'}_{g'^{-1}}^{2}}\leq\norm{v-v'}_{g'^{-1}}\norm{v+v'}_{g'^{-1}}\\
 & \leq O(1)\norm{v-v'}_{\ginverse}(\norm v_{\ginverse}+\norm{v'}_{\ginverse})\\
 & \leq O(1)\dv(\dv+\norm v_{\ginverse}).
\end{align*}
Therefore, the second term is bounded by $O(1)\Par{\dx\norm v_{\ginverse}^{2}+\dv^{2}+\dv\norm v_{\ginverse}}$.

For the third term, 
\[
\Abs{\frac{1}{2}\left(\log\det g(x)-\log\det g(x')\right)}=\frac{1}{2}\Abs{\log\det g'^{-\half}gg'^{-\half}}\leq O(n\delta_{x}),
\]
where the inequality follows from Lemma \ref{lem:computation_sc}
and the fact that the determinant is the product of eigenvalues.
\end{proof}
Since the ideal RHMC preserves the Hamiltonian along its Hamiltonian
curve, $H(x,v)=H(x_{h},v_{h})$. Hence, we can obtain a lower bound
on the acceptance probability by computing either 
\[
\min\Par{1,\frac{e^{-H(\bx_{h},\bv_{h})}}{e^{-H(x,v)}}}\ \text{or}\ \min\Par{1,\frac{e^{-H(\bx_{h},\bv_{h})}}{e^{-H(x_{h},v_{h})}}}.
\]

\begin{lem}
\label{lem:rejFinal} Let $(x_{h},v_{h})$ and $(\bx,\bv)$ be the
points obtained by the ideal RHMC and discretized RHMC with a sensitive
numerical integrator starting at $x\in\mcal_{\rho}$
with $v\in\vgood$. If the step size $h$ satisfies 
\[
h^{2}\leq\frac{10^{-10}}{n+\sqrt{M_{1}}+M_{2}},\:h^{2}C_{x}(x,v)\leq\frac{10^{-10}}{n+\sqrt{M_{1}}+\sqrt{M_{1}^{*}}},\,h^{2}C_{v}(x,v)\leq\frac{10^{-10}}{\sqrt{n+\sqrt{M_{1}}}},
\]
then the rejection probability of the Metropolis filter is bounded
by $10^{-3}$.
\end{lem}

\begin{proof}
We use the first condition on the step size to obtain $\norm{v_{h}}_{g^{-1}}=O\Par{h\Par{n+\sqrt{M_{1}}}}=O\Par{\sqrt{n+\sqrt{M_{1}}}}$
by Proposition \ref{prop:dynamics}-1. Then the claims follows from
\begin{align*}
\dtv(\oprop_{x}',\oprop_{x}) & \leq10^{5}\Par{\dx\sqrt{M_{1}^{*}}+\dx\Par{n+\norm{v_{h}}_{g^{-1}}^{2}}+\dv\Par{\dv+\norm{v_{h}}_{g^{-1}}}}\\
 & \leq10^{6}h^{2}\Par{C_{x}(x,v)\Par{n+\sqrt{M_{1}}+\sqrt{M_{1}^{*}}}+C_{v}(x,v)\Par{C_{v}(x,v)h^{2}+\sqrt{n+\sqrt{M_{1}}}}}\\
 & \leq10^{-4}+10^{-20}+10^{-4}\leq10^{-3},
\end{align*}
where we used the second and third step-size conditions in the last
inequality.
\end{proof}
Putting three main parts together, we obtain the result on the mixing
rate of RHMC discretized by a sensitive numerical integrator.

\begin{proof}[Proof of Theorem~\ref{thm:discGen}]
By Lemma \ref{lem:onestepIdeal}, \ref{lem:tvidealdis} and \ref{lem:rejFinal},
we have $\dtv(\oprop_{x},\oprop_{y})\leq\frac{9}{10}$ if $d_{\phi}(x,y)\leq h$
for $x,y\in\mcal_{\rho}$. Then the claim follows by reproducing the
proof of Proposition \ref{prop:idealConv}.
\end{proof}

\section{Numerical integrators \label{sec:numerical}}

We examine two numerical integrators commonly used in practice, the
implicit midpoint method (IMM) in Section \ref{subsec:imm} and the
generalized Leapfrog method (LM) in Section \ref{subsec:leapfrog}.
To this end, we bound $C_x$ and $C_v$, the second-orderness parameters, and then find a condition on step size for the sensitivity. We note that these integrators
are symplectic (so measure-preserving) and time-reversible (see \cite{hairer2006geometric}). 

\subsection{Implicit midpoint method \label{subsec:imm}}

For an initial condition $(x,v)$ and step size $h$, the implicit
midpoint method attempts to find the solution $(x',v')$ for the following
implicit equation: 
\[
x'=x+h\frac{\partial H}{\partial v}\Par{\frac{x+x'}{2},\frac{v+v'}{2}},\ v'=v-h\frac{\partial H}{\partial x}\Par{\frac{x+x'}{2},\frac{v+v'}{2}}.
\]
In general, these implicit equations require several iterations so
that an initial guess for this equation converges to the fixed point
$(x',v')$. 

In this section, we consider a variant of IMM in Algorithm \ref{alg:Leap-IMM}
instead. It has computational benefits over the original IMM, since
iterations for finding the fixed point of the integrator run with
a simpler Hamiltonian $H_{2}(x,v)=\half v^{\top}g(x)^{-1}v$ instead
of $H=H_{1}+H_{2}$. We then prove that if for $x\in\mcal_{\rho}$
and $v\in\vgood$ step size $h$ satisfies $h^{2}\Par{n+\sqrt{M_{1}}}\leq10^{-10}$,
then IMM is second-order with $C_{x}(x,v)=O\Par{n+\sqrt{M_{1}}}$
and $C_{v}(x,v)=O\Par{\sqrt{n+\sqrt{M_{1}}}\Par{n+\sqrt{M_{1}}+M_{2}^{*}}}$.
Moreover, if the step size $h$ satisfies $h^{2}\leq\min\Par{\frac{10^{-10}}{\Par{n+\sqrt{M_{1}}}^{2}},\frac{10^{-5}}{\sqrt{n}\bar{R}_{1}}}$
in addition to the step-size conditions in Proposition \ref{prop:dynamics},
then IMM is sensitive at $x\in\mcal_{\rho}$ and $v\in\vgood$.

\begin{algorithm2e}[h!]

\caption{$\texttt{Implicit Midpoint Method}$}\label{alg:Leap-IMM}

\SetAlgoLined

\textbf{Input:} Initial point $x$, velocity $v$, step size $h$

\tcp{Step 1: Solve $\frac{dx}{dt}=\frac{\partial H_{1}(x,v)}{\partial v},\frac{dv}{dt}=-\frac{\partial H_{1}(x,v)}{\partial x}$
}

Set $x_{\frac{1}{3}}\gets x$ and $v_{\frac{1}{3}}\gets v-\frac{h}{2}\frac{\partial H_{1}(x,v)}{\partial x}$. 

\ 

\tcp{Step 2: Solve $\frac{dx}{dt}=\frac{\partial H_{2}(x,v)}{\partial v},\frac{dv}{dt}=-\frac{\partial H_{2}(x,v)}{\partial x}$
(Implicit)}

Find $(x_{{\tt }},v_{{\tt }})$ such that 

\begin{align*}
x_{{\tt }} & =x_{\ot}+h\frac{\partial H_{2}}{\partial v}\Par{\frac{x_{\ot}+x_{{\tt }}}{2},\frac{v_{\ot}+v_{{\tt }}}{2}},\\
v_{{\tt }} & =v_{\ot}-h\frac{\partial H_{2}}{\partial x}\Par{\frac{x_{\ot}+x_{{\tt }}}{2},\frac{v_{\ot}+v_{{\tt }}}{2}}.
\end{align*}

\ 

\tcp{Step 3: Solve $\frac{dx}{dt}=\frac{\partial H_{1}(x,v)}{\partial v},\frac{dv}{dt}=-\frac{\partial H_{1}(x,v)}{\partial x}$
}

Set $x_{1}\gets x_{\frac{2}{3}}$ and $v_{1}\gets v_{\frac{2}{3}}-\frac{h}{2}\frac{\partial H_{1}}{\partial x}\Par{x_{\frac{2}{3}},v_{\frac{2}{3}}}$.

\textbf{Output:} $x_{1,}v_{1}$

\end{algorithm2e}

\subsubsection{Second-order}
\begin{lem}
\label{lem:LeapIMM-second} For $x\in\mcal_{\rho}$ and $v\in\vgood$,
let $g=g(x)$ and $h$ step size of IMM with $h^{2}\Par{n+\sqrt{M_{1}}}\leq10^{-10}$.
Let $(\overline{x},\overline{v})$ be the point obtained from RHMC
discretized by IMM with the step size $h$ and initial condition $(x,v)$.
\begin{enumerate}
\item $\norm{x-\overline{x}}_{g}=O\Par{h\sqrt{n}+h^{2}\Par{n+\sqrt{M_{1}}}}$
and $\norm{v-\overline{v}}_{g^{-1}}=O\Par{h\Par{\norm{\grad f(\bx)}_{g^{-1}}+n+\sqrt{M_{1}}}}$.
\item $C_{x}(x,v)=O\Par{n+\sqrt{M_{1}}}$.
\item $C_{v}(x,v)=O\Par{\sqrt{n+\sqrt{M_{1}}}\Par{n+\sqrt{M_{1}}+M_{2}^{*}}}$
.
\end{enumerate}
\end{lem}

\paragraph{Proof of 1.}

Let $\bx=x_{{\tt }}=\bTx(v)$ and $v_{\frac{1}{3}},v_{\frac{2}{3}},\bv$
be the velocity points obtained when starting with $(x,v)$. Then
$x_{\mid}$ and $v_{\mid}$ satisfy 
\begin{align}
x_{{\tt }} & =x_{\ot}+hg_{\mid}^{-1}v_{\mid},\label{eq:mixLI-onestep}\\
v_{{\tt }} & =v_{\ot}+\frac{h}{2}Dg_{\mid}\Brack{g_{\mid}^{-1}v_{\mid},g_{\mid}^{-1}v_{\mid}},\nonumber 
\end{align}
where $x_{\mid}=\frac{x_{\ot}+x_{{\tt }}}{2},v_{\mid}=\frac{v_{\ot}+v_{{\tt }}}{2}$
and $g_{\mid}=g\Par{x_{\mid}}$. Since $\norm{\bx-x}_{g_{\mid}}=\norm{x_{{\tt }}-x}_{g_{\mid}}\to0$
as $h\to0$, we can take $h_{0}>0$ such that $\norm{x_{{\tt }}-x}_{g_{\mid}}\leq\frac{1}{1000}$
for $h\leq h_{0}$ with the equality held at $h=h_{0}$, or $\norm{x_{\tt}-x}_{g_{\mid}}\leq\frac{1}{1000}$
for any $h>0$.

We start with the former. By adding $v_{\ot}$ to the second line
of (\ref{eq:mixLI-onestep}) and dividing by $2$, we have 
\begin{equation}
h_{0}v_{\mid}=h_{0}v_{\ot}+\frac{h_{0}^{2}}{4}Dg_{\mid}\Brack{g_{\mid}^{-1}v_{\mid},g_{\mid}^{-1}v_{\mid}}.\label{eq:IMMeq1}
\end{equation}
As $\norm{x_{{\tt }}-x}_{g_{\mid}}=h_{0}\norm{v_{\mid}}_{g_{\mid}^{-1}}$
from the first line of (\ref{eq:mixLI-onestep}), taking the $g_{\mid}^{-1}$-norm
on both sides of (\ref{eq:IMMeq1}) and using Lemma \ref{lem:partH_dist}
yield
\begin{align*}
\frac{1}{1000} & =h_{0}\norm{v_{\mid}}_{g_{\mid}^{-1}}\\
 & \leq h_{0}\norm{v_{\ot}}_{g_{\mid}^{-1}}+\frac{h_{0}^{2}}{2}\norm{v_{\mid}}_{g_{\mid}^{-1}}^{2}\leq h_{0}\norm{v_{\ot}}_{g_{\mid}^{-1}}+\frac{1}{2000},
\end{align*}
and we obtain $\frac{1}{2000}\leq h_{0}\norm{v_{\ot}}_{g_{\mid}^{-1}}$.
Recall that $\norm{\bx-x}_{g_{\mid}}\leq1/1000$ for $h\leq h_{0}$,
so we can swap the local norms between $x_{\mid}$ and $x$ due to
Lemma \ref{lem:computation_sc}, losing a multiplicative constant
like $1.001$. Using Lemma \ref{lem:partH_dist} on the first step
of the numerical integrator, 
\begin{align}
\norm{v_{\ot}}_{g_{\mid}^{-1}} & \leq1.001\norm{v_{\ot}}_{g^{-1}}\leq1.001\Par{\norm v_{g^{-1}}+h_{0}\Par{n+\sqrt{M_{1}}}}\label{eq:LMIMM-ineq1}\\
 & \leq200\sqrt{n}+2h_{0}\Par{n+\sqrt{M_{1}}}.\nonumber 
\end{align}
Due to $1/2000\leq h_{0}\norm{v_{\ot}}_{g_{\mid}^{-1}}$, it follows
that 
\begin{align*}
\frac{1}{2000} & \leq h_{0}\norm{v_{\ot}}_{g_{\mid}^{-1}}\leq200\sqrt{n}h_{0}+2h_{0}^{2}\Par{n+\sqrt{M_{1}}},
\end{align*}
and solving this for $h_{0}$ we have $h_{0}\geq\frac{1}{10^{4}\sqrt{n+\sqrt{M_{1}}}}$.
For the case of $\norm{x_{\tt}-x}_{g_{\mid}}\leq\frac{1}{1000}$ for
any $h>0$, we can simply think of $h_{0}$ as $\infty$.

Now for $h\leq h_{0}$, we can obtain from (\ref{eq:mixLI-onestep})
\begin{align}
\norm{x_{{\tt }}-x}_{g_{\mid}} & \leq h\norm{v_{\mid}}_{g_{\mid}^{-1}},\label{eq:LMIMM-ineq2}\\
\norm{v_{{\tt }}-v_{\ot}}_{g_{\mid}^{-1}} & \leq h\norm{v_{\mid}}_{g_{\mid}^{-1}}^{2}.\nonumber 
\end{align}
Using this and $h\norm{v_{\mid}}_{g_{\mid}^{-1}}=\norm{x_{{\tt }}-x}_{g_{\mid}}\leq\frac{1}{1000}$,
we can bound $\norm{v_{\mid}}_{g_{\mid}^{-1}}$ by 
\begin{align*}
\norm{v_{\mid}}_{g_{\mid}^{-1}} & =\norm{v_{\ot}+\frac{v_{{\tt }}-v_{\ot}}{2}}_{g_{\mid}^{-1}}\leq\norm{v_{\ot}}_{g_{\mid}^{-1}}+\half h\norm{v_{\mid}}_{g_{\mid}^{-1}}^{2}\\
 & \leq\norm{v_{\ot}}_{g_{\mid}^{-1}}+\frac{1}{2000}\norm{v_{\mid}}_{g_{\mid}^{-1}},
\end{align*}
 and thus 
\begin{align}
\norm{v_{\mid}}_{g_{\mid}^{-1}}\leq\frac{2000}{1999}\norm{v_{\ot}}_{g_{\mid}^{-1}} & \underset{\eqref{eq:LMIMM-ineq1}}{\leq}200\sqrt{n}+2h\Par{n+\sqrt{M_{1}}}.\label{eq:LMIMM-ineq4}
\end{align}
Putting this back to (\ref{eq:LMIMM-ineq2}) for step size $h\leq\frac{1}{10^{5}\sqrt{n+\sqrt{M_{1}}}}$,
we have 
\begin{align*}
\norm{x_{{\tt }}-x}_{g_{\mid}} & \leq200h\sqrt{n}+2h^{2}\Par{n+\sqrt{M_{1}}},\\
\norm{v_{{\tt }}-v_{\ot}}_{g_{\mid}^{-1}} & \leq125h\Par{n+\sqrt{M_{1}}}.
\end{align*}
Hence by substituting the step size into above and switching local
norms properly, we have $\norm{x_{{\tt }}-x}_{g}\leq10^{-8}$.

By applying Lemma \ref{lem:partH_dist} to $\overline{v}=v_{{\tt }}-\frac{h}{2}\frac{\partial H_{1}}{\partial x}\Par{x_{{\tt }},v_{{\tt }}}$
in the third step, we also have 
\begin{align*}
\norm{\bv-v}_{g^{-1}} & \leq\norm{\bv-v_{\frac{2}{3}}}_{g^{-1}}+\norm{v_{\frac{2}{3}}-v_{\frac{1}{3}}}_{g^{-1}}+\norm{v_{\ot}-v}_{g^{-1}}\\
 & \leq1.001h\Par{\norm{\grad f(x_{{\tt }})}_{g^{-1}}+n+125\Par{n+\sqrt{M_{1}}}+\Par{n+\sqrt{M_{1}}}}\\
 & \leq200h\Par{\norm{\grad f(x_{{\tt }})}_{g^{-1}}+n+\sqrt{M_{1}}}.
\end{align*}
In conclusion,
\begin{align*}
\norm{\bx-x}_{g} & \leq200h\sqrt{n}+3h^{2}\Par{n+\sqrt{M_{1}}},\\
\norm{\bv-v}_{g^{-1}} & \leq200h\Par{\norm{\grad f(\bx)}_{g^{-1}}+n+\sqrt{M_{1}}}.
\end{align*}

\paragraph{Proof of 2.}

For $t\in[0,h]$, let $(x_{t},v_{t})$ be the Hamiltonian curve of
the ideal RHMC at time $t$ starting from $(x,v)$. Recall that for
$g_{t}=g(x_{t})$
\begin{align*}
\Tx(v) & =x+\int_{0}^{h}\frac{\partial H}{\partial v}(x_{t},v_{t})dt=x+\int_{0}^{h}g_{t}^{-1}v_{t}dt,\\
\bTx(v) & =\overline{x}=x+hg_{\mid}^{-1}v_{\mid}.
\end{align*}
Thus, 
\begin{align*}
\norm{\Tx(v)-\bTx(v)}_{g_{\mid}} & =\norm{\Par{x+\int_{0}^{h}g_{t}^{-1}v_{t}dt}-\Par{x+hg_{\mid}^{-1}v_{\mid}}}_{g_{\mid}}\\
 & =\norm{\int_{0}^{h}\Par{g_{t}^{-1}v_{t}-g_{\mid}^{-1}v_{\mid}}dt}_{g_{\mid}}\\
 & \leq h\max_{t\in[0,h]}\norm{g_{t}^{-1}v_{t}-g_{\mid}^{-1}v_{\mid}}_{g_{\mid}}
\end{align*}
By Lemma \ref{lem:computation_sc}-11, 
\begin{align*}
\norm{g_{t}^{-1}v_{t}-g_{\mid}^{-1}v_{\mid}}_{g_{\mid}} & \lesssim\norm{v_{t}-v_{\mid}}_{g_{\mid}^{-1}}+\norm{x_{t}-x_{\mid}}_{g_{\mid}}\norm{v_{t}}_{g_{\mid}^{-1}}\\
 & \leq\half\Par{\norm{v_{t}-v_{\ot}}_{g_{\mid}^{-1}}+\norm{v_{t}-v_{{\tt }}}_{g_{\mid}^{-1}}}\\
 & \quad+\half\Par{\norm{x_{t}-x}_{g_{\mid}}+\norm{x_{t}-\overline{x}}_{g_{\mid}}}\norm{v_{t}}_{g_{\mid}^{-1}}\\
 & \lesssim\Par{\norm{v-v_{t}}_{g^{-1}}+\norm{v-v_{\ot}}_{g^{-1}}+\norm{v_{t}-v}_{g^{-1}}+\norm{v-v_{{\tt }}}_{g^{-1}}}\\
 & \quad+\Par{\norm{x_{t}-x}_{g}+\norm{x_{t}-x}_{g}+\norm{\overline{x}-x}_{g}}\Par{\norm{v_{t}-v}_{g^{-1}}+\norm v_{g^{-1}}}\\
 & \lesssim\Par{\norm{v-v_{t}}_{g^{-1}}+\norm{v-v_{\ot}}_{g^{-1}}+\norm{v_{\ot}-v_{{\tt }}}_{g^{-1}}}\\
 & \quad+\Par{\norm{x_{t}-x}_{g}+\norm{\overline{x}-x}_{g}}\Par{\norm{v_{t}-v}_{g^{-1}}+\norm v_{g^{-1}}}.
\end{align*}
Using our bounds on $\norm{\bx-x}_{g},\norm{x_{t}-x}_{g}$ and $\norm v_{g^{-1}},\norm{v_{t}-v}_{g^{-1}},\norm{v-v_{\ot}}_{g^{-1}},\norm{v_{\ot}-v_{{\tt }}}_{g^{-1}}$,
we conclude that $\max_{t\in[0,h]}\norm{g_{t}^{-1}v_{t}-g_{\mid}^{-1}v_{\mid}}_{g_{\mid}}\leq10^{4}h\Par{n+\sqrt{M_{1}}}$,
and thus 
\[
\norm{\Tx(v)-\bTx(v)}_{g}\leq10^{4}h^{2}\Par{n+\sqrt{M_{1}}}.
\]

\paragraph{Proof of 3.}

From the algorithm,
\begin{align*}
v_{h} & =v-\int_{0}^{h}\frac{\partial H}{\partial x}(x_{t},v_{t})dt\\
 & =v-\int_{0}^{h}\frac{\del H_{1}}{\del x}(x_{t},v_{t})dt-\int_{0}^{h}\frac{\del H_{2}}{\del x}(x_{t},v_{t})dt,\\
\bv & =v_{{\tt }}-\frac{h}{2}\frac{\del H_{1}}{\del x}(x_{{\tt }},v_{{\tt }})=v_{\ot}-h\frac{\del H_{2}}{\del x}(x_{\mid},v_{\mid})-\frac{h}{2}\frac{\del H_{1}}{\del x}(x_{{\tt }},v_{{\tt }})\\
 & =v-\frac{h}{2}\Par{\frac{\del H_{1}}{\del x}(x,v)+\frac{\del H_{1}}{\del x}(x_{{\tt }},v_{{\tt }})}-h\frac{\del H_{2}}{\del x}(x_{\mid},v_{\mid}).
\end{align*}
Thus,
\begin{align}
 & \norm{v_{h}-\bv}_{g_{\mid}^{-1}}\label{eq:LMIMM-ineq3}\\
 & =\norm{\int_{0}^{h}\Par{\frac{\del H_{1}}{\del x}(x_{t},v_{t})-\half\Par{\frac{\del H_{1}}{\del x}(x,v)+\frac{\del H_{1}}{\del x}(x_{{\tt }},v_{{\tt }})}}dt+\int_{0}^{h}\Par{\frac{\del H_{2}}{\del x}(x_{t},v_{t})-\frac{\del H_{2}}{\del x}(x_{\mid},v_{\mid})}dt}_{g_{\mid}^{-1}}\nonumber \\
 & \leq h\underbrace{\max_{t\in[0,h]}\norm{\frac{\del H_{1}}{\del x}(x_{t},v_{t})-\half\Par{\frac{\del H_{1}}{\del x}(x,v)+\frac{\del H_{1}}{\del x}(x_{{\tt }},v_{{\tt }})}}_{g_{\mid}^{-1}}}_{F}\nonumber \\
 & \quad+h\underbrace{\max_{t\in[0,h]}\norm{\frac{\del H_{2}}{\del x}(x_{t},v_{t})-\frac{\del H_{2}}{\del x}(x_{\mid},v_{\mid})}_{g_{\mid}^{-1}}}_{S}.\nonumber 
\end{align}
We separately bound $F$ and $S$. For $F$, by following the proof
of Proposition \ref{lem:computation_sc}-12, we have
\begin{align*}
F & \lesssim n\Par{\norm{v_{t}-v}_{g^{-1}}+\norm{v_{t}-v_{{\tt }}}_{g^{-1}}}+M_{2}^{*}\Par{\norm{x_{h}-x}_{g}+\norm{x_{h}-\bx}_{g}}.
\end{align*}
Using our bounds on $\norm{x_{h}-\bx}_{g},\norm{x_{h}-x}_{g}$ and
$\norm{v_{t}-v}_{g^{-1}},\norm{v-v_{{\tt }}}_{g^{-1}}$, we obtain
\[
F\lesssim h\Par{n+\sqrt{M_{1}}}^{3/2}+h\sqrt{n+\sqrt{M_{1}}}M_{2}^{*}.
\]
We remark that the smoothness of $f$ guarantees that $M_{2}^{*}$
is bounded by some constant for all sufficiently small $h$.

Similarly for $S$, we have that for $\dv=\norm{v_{t}-v_{\mid}}_{g_{\mid}^{-1}}$
and $\dx=\norm{x_{t}-x_{\mid}}_{g_{\mid}}$
\[
S\lesssim\max_{t\in[0,h]}\Par{\dv+\dx\norm{v_{\mid}}_{g_{\mid}^{-1}}}\Par{\norm{v_{\mid}}_{g_{\mid}^{-1}}+\norm{v_{t}}_{g_{\mid}^{-1}}}.
\]
Using our bounds on $\norm{\bx-x}_{g},\norm{x_{t}-x}_{g}$ and $\norm v_{g^{-1}},\norm{v_{t}-v}_{g^{-1}},\norm{v-v_{\ot}}_{g^{-1}},\norm{v_{\ot}-v_{{\tt }}}_{g^{-1}}$,
it follows that 
\[
S\lesssim h\Par{n+\sqrt{M_{1}}}^{3/2}.
\]
Substituting the bounds on $F$ and $S$ into (\ref{eq:LMIMM-ineq3})
we can conclude that
\[
\norm{v_{h}-\bv}_{g^{-1}}\leq10^{10}h^{2}\sqrt{n+\sqrt{M_{1}}}\Par{n+\sqrt{M_{1}}+M_{2}^{*}}.
\]

\subsubsection{Sensitivity}

We use Proposition \ref{prop:stability} to show that IMM is sensitive.
Here we assume that $\log\det g(x)$ is convex in $\mcal$, which
is the case for the logarithmic barriers of polytopes.
\begin{lem}
\label{lem:IMMstable} For $x\in\mcal_{\rho}$ and $v\in\vgood$,
if $\log\det g(x)$ is convex in $x$, then IMM is sensitive at $(x,v)$
for step size $h$ satisfying $h^{2}\leq\min\Par{\frac{10^{-10}}{\Par{n+\sqrt{M_{1}}}^{2}},\frac{10^{-5}}{\sqrt{n}\bar{R}_{1}}}$
and the step-size conditions in Proposition \ref{prop:dynamics}.
\end{lem}

\begin{proof}
Let $\overline{F}(x,v)=(\bx,\bv)$ and $\bTx(v)=\bx$. We lower bound
$\Abs{D\bTx(v)}.$ Recall that one iteration of IMM consists of three
steps with input $(x,v)$ and output $(x_{1},v_{1})$, as described
in the following diagram: 
\[
\Par{x,v}\overset{X}{\longrightarrow}\Par{x_{\ot},v_{\ot}}\overset{Y}{\longrightarrow}\Par{x_{{\tt }},v_{{\tt }}}\overset{Z}{\longrightarrow}\Par{x_{1},v_{1}},
\]
where each of the maps $X,Y$ and $Z$ is defined by 
\begin{align*}
X\Par{x,v} & =\Par{x,v-\frac{h}{2}\frac{\partial H_{1}(x,v)}{\partial x}},\\
Y\Par{x_{\ot},v_{\ot}} & =\Par{x_{\ot}+h\frac{\partial H_{2}(x_{\mid},v_{\mid})}{\partial v},v_{\frac{1}{3}}-h\frac{\partial H_{2}(x_{\mid},v_{\mid})}{\partial x}},\\
Z\Par{x_{{\tt }},v_{{\tt }}} & =\Par{x_{{\tt }},v_{{\tt }}-\frac{h}{2}\frac{\partial H_{1}(x_{{\tt }},v_{{\tt }})}{\partial x}},
\end{align*}
for $x_{\mid}=\frac{x_{\ot}+x_{{\tt }}}{2}$ and $v_{\mid}=\frac{v_{\ot}+v_{{\tt }}}{2}$.
Due to $\bTx(v)=\pi_{x}\circ(Z\circ Y\circ X)(x,v)$, it follows that
$D\bTx(v)$ is the upper-right $n\times n$ submatrix of $D(Z\circ Y\circ X)(x,v)$.
From direct computation, we have 
\begin{align*}
DX(x,v) & =\left[\begin{array}{cc}
I & 0\\*
* & I
\end{array}\right],\\
DY(x_{\ot},v_{\ot}) & =\left[\begin{array}{cc}
P & Q\\
R & S
\end{array}\right],\\
DZ(x_{{\tt }},v_{{\tt }}) & =\left[\begin{array}{cc}
I & 0\\*
* & I
\end{array}\right],
\end{align*}
and due to $D(Z\circ Y\circ X)=DZ\cdot DY\cdot DX$ we have $D\bTx(v)=Q$.
Thus, it suffices to focus on the second step only (i.e., the map
$Y$).

Now let us represent the map $Y$ in a compact way. With two symbols
\[
r=\left[\begin{array}{c}
x\\
v
\end{array}\right]\in\R^{2n}\ \text{and}\ J=\left[\begin{array}{cc}
0 & I_{n}\\
-I_{n} & 0
\end{array}\right]\in\R^{2n\times2n},
\]
the second step can be rewritten as 
\[
r_{{\tt }}=r_{\ot}+hJ\nabla_{(x,v)}H_{2}(r_{\mid}),
\]
where $r_{*}=\left[\begin{array}{c}
x_{*}\\
v_{*}
\end{array}\right]$ for $*\in\Brace{\ot,{\tt ,\mid}}$ and $H_{2}(x,v)=\half v^{\top}g(x)^{-1}v$.
Differentiating both sides by $r_{\ot}$,
\[
\frac{\partial r_{{\tt }}}{\del r_{\ot}}=I_{2n}+hJ\nabla^{2}H_{2}(r_{\mid})\Par{\half I_{2n}+\half\frac{\partial r_{{\tt }}}{\del r_{\ot}}}.
\]
As $DY\Par{r_{\ot}}=\frac{\partial r_{{\tt }}}{\del r_{\ot}}$, we
have that
\[
\Par{I_{2n}-\frac{h}{2}J\nabla^{2}H_{2}(r_{\mid})}DY\Par{r_{\ot}}=I_{2n}+\frac{h}{2}J\nabla^{2}H_{2}(r_{\mid}).
\]
 For $G(x)\defeq\left[\begin{array}{cc}
g(x)^{\half} & 0\\
0 & g(x)^{-\half}
\end{array}\right]$, we have 
\begin{align*}
G(x_{\mid}) & \Par{I_{2n}-\frac{h}{2}J\nabla^{2}H_{2}(r_{\mid})}G(x_{\mid})^{-1}G(x_{\mid})DY\Par{r_{\ot}}G(x_{\mid})^{-1}\\
 & =G(x_{\mid})\Par{I_{2n}+\frac{h}{2}J\nabla^{2}H_{2}(r_{\mid})}G(x_{\mid})^{-1}
\end{align*}
and 
\begin{align}
 & \Par{I_{2n}-\frac{h}{2}G(x_{\mid})J\nabla^{2}H_{2}(r_{\mid})G(x_{\mid})^{-1}}G(x_{\mid})DY\Par{r_{\ot}}G(x_{\mid})^{-1}\nonumber \\
 & =I_{2n}+\frac{h}{2}G(x_{\mid})J\nabla^{2}H_{2}(r_{\mid})G(x_{\mid})^{-1}.\label{eq:impdiff}
\end{align}
Let us look into the term $B\defeq G(x_{\mid})J\nabla^{2}H_{2}(r_{\mid})G(x_{\mid})^{-1}$.
By direct computation, for block matrices $B_{1},B_{2},B_{3},B_{4}$
of size $n\times n$ we have
\begin{align*}
B & \defeq\left[\begin{array}{cc}
B_{1} & B_{2}\\
B_{3} & B_{4}
\end{array}\right]\\
 & =\left[\begin{array}{cc}
g(x_{\mid})^{\half}\\
 & g(x_{\mid})^{-\half}
\end{array}\right]\left[\begin{array}{cc}
\frac{\del^{2}H_{2}}{\del v\del x}(r_{\mid}) & \frac{\del^{2}H_{2}}{\del v^{2}}(r_{\mid})\\
-\frac{\del^{2}H_{2}}{\del x^{2}}(r_{\mid}) & -\Par{\frac{\del^{2}H_{2}}{\del x\del v}(r_{\mid})}^{\top}
\end{array}\right]\left[\begin{array}{cc}
g(x_{\mid})^{-\half}\\
 & g(x_{\mid})^{\half}
\end{array}\right]
\end{align*}
and thus 
\begin{align}
B_{1} & =g(x_{\mid})^{-\half}Dg(x_{\mid})\Brack{g(x_{\mid})^{-1}v_{\mid}}g(x_{\mid})^{-\half},\label{eq:LMIMM-eq4}\\
B_{2} & =I_{n},\nonumber \\
B_{3} & =g(x_{\mid})^{-\half}\bigg(-v_{\mid}^{\top}g(x_{\mid})^{-1}Dg(x_{\mid})g(x_{\mid})^{-1}Dg(x_{\mid})g(x_{\mid})^{-1}v_{\mid}\nonumber \\
 & ~~+\half v_{\mid}^{\top}g(x_{\mid})^{-1}D^{2}g(x_{\mid})g(x_{\mid})^{-1}v_{\mid}\bigg)g(x_{\mid})^{-\half}\nonumber \\
 & =-B_{1}^{2}+\half g(x_{\mid})^{-\half}D^{2}g(x_{\mid})\Brack{g(x_{\mid})^{-1}v_{\mid},g(x_{\mid})^{-1}v_{\mid}}g(x_{\mid})^{-\half},\nonumber \\
B_{4} & =-B_{1}^{\top}.\nonumber 
\end{align}
Now we bound the operator norm of $B_{i}$ for each $i\in[4]$ as
follows. 
\begin{align*}
\norm{B_{1}} & =\norm{B_{4}}\\
 & =\max_{p,q:\norm p_{2},\norm q_{2}\leq1}p^{\top}g(x_{\mid})^{-\half}Dg(x_{\mid})\Brack{g(x_{\mid})^{-1}v_{\mid}}g(x_{\mid})^{-\half}q\\
 & =\max_{p,q}Dg(y_{\mid})\Brack{g(x_{\mid})^{-1}v_{\mid},g(x_{\mid})^{-\half}p,g(x_{\mid})^{-\half}q}\\
 & \leq2\max_{p,q}\norm{g(x_{\mid})^{-1}v_{\mid}}_{g(x_{\mid})}\norm{g(x_{\mid})^{-\half}p}_{g(x_{\mid})}\norm{g(x_{\mid})^{-\half}q}_{g(x_{\mid})}\\
 & =2\max_{p,q}\norm{g(x_{\mid})^{-1}v_{\mid}}_{g(x_{\mid})}\norm p_{2}\norm q_{2}\leq2\norm{v_{\mid}}_{g_{\mid}^{-1}}\\
 & \leq O\Par{\sqrt{n+\sqrt{M_{1}}}},
\end{align*}
where we used (\ref{eq:LMIMM-ineq4}) guaranteed by the condition
of $h^{2}\Par{n+\sqrt{M_{1}}}\leq10^{-10}$ ($\onecirc$ in Proposition
\ref{prop:dynamics}). For $B_{2}$ and $B_{3}$, we have
\begin{align*}
\norm{B_{2}} & =1,\\
\norm{B_{3}} & \leq\norm{B_{1}}^{2}+\half\max_{p,q:\norm p_{2},\norm q_{2}\leq1}D^{2}g(x_{\mid})\Brack{g(x_{\mid})^{-1}v_{\mid},g(x_{\mid})^{-1}v_{\mid},g(x_{\mid})^{-\half}p,g(x_{\mid})^{-\half}q}\\
 & \leq O\Par{n+\sqrt{M_{1}}}+3\max_{p,q}\norm{g(x_{\mid})^{-1}v_{\mid}}_{g(x_{\mid})}^{2}\norm{g(x_{\mid})^{-\half}p}_{g(x_{\mid})}\norm{g(x_{\mid})^{-\half}q}_{g(x_{\mid})}\\
 & =O\Par{n+\sqrt{M_{1}}}+O\Par{n+\sqrt{M_{1}}}\max_{p,q}\norm p_{2}\norm q_{2}\\
 & =O\Par{n+\sqrt{M_{1}}},
\end{align*}
where the second inequality for $\norm{B_{3}}$ follows from the highly
self-concordance of $\phi$. Due to $\norm B\leq\sum_{i=1}^{4}\norm{B_{i}}$,
we have $\norm{\frac{h}{2}B}=O\Par{h\Par{n+\sqrt{M_{1}}}}$. Hence,
the condition of $h^{2}\Par{n+\sqrt{M_{1}}}^{2}\leq10^{-10}$ ensures
that the inverse of $I_{2n}-\frac{h}{2}B$ exists, and it can be written
as a series of matrices, 
\[
\Par{I_{2n}-\frac{h}{2}B}^{-1}=\sum_{i=0}^{\infty}(hB/2)^{i}.
\]
By substituting this series into (\ref{eq:impdiff}), 
\begin{align*}
G(x_{\mid})DY\Par{r_{\ot}}G(x_{\mid})^{-1} & =\sum_{i=0}^{\infty}(hB/2)^{i}\Par{I_{2n}+\frac{h}{2}B}=\sum_{i=0}^{\infty}(hB/2)^{i}+\sum_{i=1}^{\infty}(hB/2)^{i}\\
 & =I_{2n}+2\sum_{i=1}^{\infty}(hB/2)^{i}.
\end{align*}
By multiplying $\left[\begin{array}{cc}
I_{n} & 0\end{array}\right]^{\top}$ to the left and $\left[\begin{array}{c}
0\\
I_{n}
\end{array}\right]$ to the right on both sides, 
\[
g(x_{\mid})^{\half}D\bTx(v)g(x_{\mid})^{\half}=2\sum_{i=1}^{\infty}\left[\begin{array}{cc}
I_{n} & 0\end{array}\right]^{\top}(hB/2)^{i}\left[\begin{array}{c}
0\\
I_{n}
\end{array}\right].
\]
From (\ref{eq:LMIMM-eq4}), $B$ is of the form
\[
B=\left[\begin{array}{cc}
C & I_{n}\\
-C^{2}+R & -C
\end{array}\right],
\]
where $C\in\R^{n\times n}$ is symmetric and $R\in\R^{n\times n}$,
and thus by Lemma \ref{lem:matrixSeries}
\[
g(x_{\mid})^{\half}D\bTx(v)g(x_{\mid})^{\half}=h\sum_{i=0}^{\infty}(h^{2}R/2)^{i},
\]
where $R=\half g(x_{\mid})^{-\half}D^{2}g(x_{\mid})\Brack{g(x_{\mid})^{-1}v_{\mid},g(x_{\mid})^{-1}v_{\mid}}g(x_{\mid})^{-\half}$.
Thus for $E\defeq\sum_{i=1}^{\infty}(h^{2}R/2)^{i}$ 
\begin{equation}
\frac{1}{h}g(x_{\mid})^{\half}D\bTx(v)g(x_{\mid})^{\half}=I+E.\label{eq:LMIMM-last}
\end{equation}

We now bound its operator norm, trace and Frobenius norm. It is easy
to see that 
\begin{align*}
\norm E_{2} & \lesssim\sum_{i\geq1}\Par{\frac{h^{2}}{2}\Par{n+\sqrt{M_{1}}}}^{i},\\
\tr(E) & \lesssim\sum_{i\geq1}\Par{\frac{h^{2}}{2}\tr(R)}^{i}\leq\sum_{i\geq1}\Par{\frac{h^{2}}{2}n\Par{n+\sqrt{M_{1}}}},\\
\norm E_{F} & \lesssim\sum_{i\geq1}\Par{\frac{h^{2}}{2}\sqrt{n}\Par{n+\sqrt{M}}}^{i},
\end{align*}
where we used the following estimations
\begin{align*}
\norm R_{2} & \leq O\Par{n+\sqrt{M_{1}}}\\
\tr(R) & =\half\E_{p\sim\ncal(0,I)}p^{\top}g(x_{\mid})^{-\half}D^{2}g(x_{\mid})[g(x_{\mid})^{-1}v_{\mid},g^{-1}v_{\mid}]g(x_{\mid})^{-\half}p\\
 & \leq\E D^{2}g(x_{\mid})\Brack{g(x_{\mid})^{-1}v_{\mid},g(x_{\mid})^{-1}v_{\mid},g(x_{\mid})^{-\half}p,g(x_{\mid})^{-\half}p}\\
 & \leq\E\norm{v_{\mid}}_{g_{\mid}^{-1}}^{2}\norm p_{2}^{2}=O\Par{n\Par{n+\sqrt{M_{1}}}},\\
\norm R_{F} & \leq\sqrt{n}\norm R_{2}=O\Par{\sqrt{n}\Par{n+\sqrt{M}}}.
\end{align*}
Therefore, the step-size condition of $h^{2}\Par{n+\sqrt{M_{1}}}^{2}\leq10^{-10}$
ensures that these three quantities can be made smaller than $10^{-8}$.
Applying Lemma \ref{lem:boundLogDet} to (\ref{eq:LMIMM-last}), we
have 
\begin{align*}
e^{\tr(E)}e^{-\norm E_{F}^{2}}\leq\Abs{\frac{1}{h}g(x_{\mid})^{\half}D\bTx(v)g(x_{\mid})^{\half}} & \leq e^{\tr(E)}e^{\norm E_{F}^{2}},
\end{align*}
and thus 
\[
\Abs{D\bTx(v)}\geq(1-10^{-6})\cdot\frac{h^{n}}{\Abs{g(x_{\mid})}}.
\]
Since $\log\det g(x)=\log\det\nabla^{2}\phi(x)$ is convex in $x$,
it follows that 
\begin{align*}
\log\Abs{g(x_{\mid})} & =\log\Abs{g\Par{\frac{x_{\ot}+x_{{\tt }}}{2}}}\leq\half\Par{\log\Abs{g(x_{\ot})}+\log\Abs{g(x_{{\tt }})}}=\half\log\Abs{g(x_{\ot})g(x_{{\tt }})}\\
 & =\log\sqrt{\Abs{g(x)}\Abs{g(\bx)}},
\end{align*}
and thus $\Abs{D\bTx(v)}\geq\frac{(1-10^{-6})h^{n}}{\sqrt{\Abs{g(x)}\Abs{g(\bx)}}}$.
Due to the step-size conditions of $h^{2}\sqrt{n}\bar{R}_{1}\leq10^{-5}$
and in Proposition \ref{prop:dynamics}, we can use Proposition \ref{prop:stability}
to conclude that $\bTx(v)$ is sensitive at $(x,v)$.
\end{proof}

\subsection{Generalized Leapfrog method (Störmer--Verlet) \label{subsec:leapfrog}}

We now analyze the generalized Leapfrog method (Algorithm \ref{alg:LM}),
which is symplectic and reversible in the Riemannian settings. In
a similar way we analyzed IMM, we show that if step size $h$ satisfies
$h^{2}\Par{n+\sqrt{M_{1}}}\leq10^{-10}$, then LM is second-order
for $x\in\mcal_{\rho}$ and $v\in\vgood$ with $C_{x}(x,v)=O\Par{n+\sqrt{M_{1}}}$
and $C_{v}(x,v)=O\Par{\sqrt{n+\sqrt{M_{1}}}\Par{n+\sqrt{M_{1}}+M_{2}^{*}}}$.
Next, if the step size $h$ satisfies $h^{2}\leq\min\Par{\frac{10^{-20}}{n^{2}(n+\sqrt{M_{1}})},\frac{10^{-5}}{\sqrt{n}\bar{R}_{1}}}$
in addition to the step-size conditions in Proposition \ref{prop:dynamics},
then LM is sensitive at $x\in\mcal_{\rho}$ and $v\in\vgood$.

\begin{algorithm2e}[h!]

\caption{$\texttt{Generalized Leapfrog Method}$}\label{alg:LM}

\SetAlgoLined

\textbf{Input:} Initial point $x$, velocity $v$, step size $h$

\tcp{Step 1: Update $v$ (Implicit)}

Find $v_{\half}$ such that $v_{\half}\gets v-\frac{h}{2}\frac{\partial H(x,v)}{\partial x}$. 

\ 

\tcp{Step 2: Update $x$ (Implicit)}

Find $x_{1}$ such that 
\begin{align*}
x_{1} & =x+\frac{h}{2}\Par{\frac{\del H}{\del v}\Par{x,v_{\half}}+\frac{\del H}{\del v}\Par{x_{1},v_{\half}}}.
\end{align*}

\ 

\tcp{Step 3: Update $v$ (Explicit) }

Set $v_{1}\gets v_{\half}-\frac{h}{2}\frac{\partial H}{\partial x}\Par{x_{1},v_{\half}}$.

\textbf{Output:} $x_{1,}v_{1}$

\end{algorithm2e}

\subsubsection{Second-order}
\begin{lem}
\label{lem:secondLeap} For $x\in\mcal_{\rho}$ and $v\in\vgood$,
let $g=g(x)$ and $h$ step size of LM with $h^{2}\Par{n+\sqrt{M_{1}}}\leq10^{-10}$.
Let $(\overline{x},\overline{v})$ be the point obtained from RHMC
discretized by LM with the step size $h$ and initial condition $(x,v)$.
\begin{enumerate}
\item $\norm{x-\overline{x}}_{g}=O\Par{h\sqrt{n}+h^{2}\Par{n+\sqrt{M_{1}}}}$
and $\norm{v-\overline{v}}_{g^{-1}}=O\Par{h\Par{\norm{\grad f(\bx)}_{g^{-1}}+n+\sqrt{M_{1}}}}$.
\item $C_{x}(x,v)=O\Par{n+\sqrt{M_{1}}}$.
\item $C_{v}(x,v)=O\Par{\sqrt{n+\sqrt{M_{1}}}\Par{n+\sqrt{M_{1}}+M_{2}^{*}}}$.
\end{enumerate}
\end{lem}

\paragraph{Proof of 1.}

Let $\bx=x_{1}=\bTx(v)$ and $v_{1}(=\bv),v_{\half}$ be the velocity
obtained from LM with the initial condition $(x,v)$ and the step
size $h$. Let $g_{1}=g(x_{1})$. As $v_{\half}\to v$ as $h\to0$,
we can take $h_{0}>0$ such that $h\Par{\norm{v_{\half}}_{g^{-1}}+\norm{g_{1}^{-1}v_{\half}}_{g}}\leq\frac{2}{1000}$
for $h\leq h_{0}$ with the equality held at $h=h_{0}$. Thus for
$h\leq h_{0}$ we have $h\norm{v_{\half}}_{g^{-1}}\leq\frac{1}{500}$.

From the first step of Algorithm \ref{alg:LM} and Lemma \ref{lem:partH_dist},
for step size $h\leq h_{0}$ it follows from $x\in\mcal_{\rho}$ that
\[
\norm{v_{\half}}_{g^{-1}}\leq\norm v_{g^{-1}}+\frac{h}{2}\Par{\sqrt{M_{1}}+n+\norm{v_{\half}}_{g^{-1}}^{2}}.
\]
Multiplying $h$ to both sides, and using $h\norm{v_{\half}}_{g^{-1}}\leq\frac{1}{500}$
and $v\in\vgood$,

\begin{align*}
h\norm{v_{\half}}_{g^{-1}} & \leq150h\sqrt{n}+\frac{h^{2}}{2}\Par{\sqrt{M_{1}}+n}+\frac{h^{2}}{2}\norm{v_{\half}}_{g^{-1}}^{2}\\
 & \leq150h\sqrt{n}+\frac{h^{2}}{2}\Par{\sqrt{M_{1}}+n}+\frac{1}{1000}h\norm{v_{\half}}_{g^{-1}},
\end{align*}
and thus 
\begin{equation}
\norm{v_{\half}}_{g^{-1}}\leq200\sqrt{n}+h^{2}\Par{n+\sqrt{M_{1}}}.\label{eq:LM-ineq1}
\end{equation}

From the second step of Algorithm \ref{alg:LM} and Lemma \ref{lem:partH_dist},
for step size $h\leq h_{0}$
\begin{align*}
\norm{x_{1}-x}_{g} & \leq\frac{h}{2}\Par{\norm{\frac{\del H}{\del v}\Par{x,v_{\half}}}_{g}+\norm{\frac{\del H}{\del v}\Par{x_{1},v_{\half}}}_{g}}\leq\frac{h}{2}\Par{\norm{v_{\half}}_{g^{-1}}+\norm{g_{1}^{-1}v_{\half}}_{g}},
\end{align*}
and thus its is obvious that $\norm{x-x_{1}}_{g}\leq\frac{1}{500}$.
We now lower bound $h_{0}$ as follows:
\begin{align*}
\frac{1}{500} & =h_{0}\Par{\norm{v_{\half}}_{g^{-1}}+\norm{g_{1}^{-1}v_{\half}}_{g}}\leq h_{0}\Par{\norm{v_{\half}}_{g^{-1}}+1.1\norm{v_{\half}}_{g^{-1}}}\\
 & \leq3h_{0}\norm{v_{\half}}_{g^{-1}}\\
 & \leq600h_{0}\sqrt{n}+3h_{0}^{2}\Par{n+\sqrt{M_{1}}},
\end{align*}
where in the first inequality we switched the local norm at from $x_{1}$
to $x$ due to $\norm{x_{1}-x}_{g}\leq\frac{1}{500}$ and used (\ref{eq:LM-ineq1})
in the last inequality. Therefore, $h_{0}\geq\frac{1}{10^{4}\sqrt{n+\sqrt{M_{1}}}}$
and for step size $h\leq\frac{1}{10^{5}\sqrt{n+\sqrt{M_{1}}}}$ we
have 
\begin{align*}
\norm{x-x_{1}}_{g} & \leq600h\sqrt{n}+31h^{2}\Par{n+\sqrt{M_{1}}},\\
\norm{v-v_{\half}}_{g^{-1}} & \leq h\Par{20000n+\sqrt{M_{1}}}.
\end{align*}
Similarly, we can bound $\norm{v_{1}-v_{\half}}_{g^{-1}}$ by Lemma
\ref{lem:partH_dist}:
\begin{align*}
\norm{v_{1}-v_{\half}}_{g^{-1}} & \leq\frac{h}{2}\norm{\frac{\del H}{\del x}\Par{x_{1},v_{\half}}}_{g^{-1}}\leq\frac{h}{2}\Par{\norm{\grad f(x_{1})}_{g^{-1}}+n+\norm{v_{\half}}_{g_{\mid}^{-1}}^{2}}\\
 & \leq40000h\Par{\norm{\grad f(x_{1})}_{g^{-1}}+n+\sqrt{M_{1}}},
\end{align*}
and thus by adding it to the inequality for $\norm{v-v_{\half}}_{g^{-1}}$
we have 
\[
\norm{v-v_{1}}_{g^{-1}}\leq40000\Par{\norm{\grad f(x_{1})}_{g^{-1}}+n+\sqrt{M_{1}}}.
\]

\paragraph{Proof of 2.}

For $t\in[0,h]$, let $(x_{t},v_{t})$ be the Hamiltonian curve of
the ideal RHMC at time $t$ starting from $(x,v)$. Recall that 
\begin{align*}
\Tx(v) & =x+\int_{0}^{h}\frac{\partial H}{\partial v}(x_{t},v_{t})dt=x+\int_{0}^{h}g_{t}^{-1}v_{t}dt,\\
\bTx(v) & =\overline{x}=x+\frac{h}{2}\Par{g^{-1}+g_{1}^{-1}}v_{\half}.
\end{align*}
Thus, 
\begin{align*}
\norm{T_{x}(v)-\bTx(v)}_{g} & =\norm{\Par{x+\int_{0}^{h}g_{t}^{-1}v_{t}dt}-\Par{x+\frac{h}{2}\Par{g^{-1}+g_{1}^{-1}}v_{\half}}}_{g}\\
 & =\norm{\int_{0}^{h}\Par{g_{t}^{-1}v_{t}-\half\Par{g^{-1}+g_{1}^{-1}}v_{\half}}dt}_{g}\\
 & \leq h\max_{t\in[0,h]}\norm{g_{t}^{-1}v_{t}-\half\Par{g^{-1}+g_{1}^{-1}}v_{\half}}_{g}.
\end{align*}
By Lemma \ref{lem:computation_sc}-11, 
\begin{align*}
\norm{g_{t}^{-1}v_{t}-\half\Par{g^{-1}+g_{1}^{-1}}v_{\half}}_{g} & \leq\half\norm{g_{t}^{-1}v_{t}-g^{-1}v_{\half}}_{g}+\half\norm{g_{t}^{-1}v_{t}-g_{1}^{-1}v_{\half}}_{g}\\
 & \lesssim\norm{v_{t}-v_{\half}}_{g^{-1}}+\norm{x_{t}-x}_{g}\norm{v_{t}}_{g^{-1}}\\
 & \quad+\norm{v_{t}-v_{\half}}_{g^{-1}}+\norm{x_{t}-x_{1}}_{g}\norm{v_{t}}_{g^{-1}}\\
 & \lesssim\Par{\norm{v_{t}-v}_{g^{-1}}+\norm{v-v_{\half}}_{g^{-1}}}\\
 & \quad+\Par{\norm{x_{t}-x}_{g}+\norm{x-x_{1}}_{g}}\Par{\norm{v_{t}-v}_{g^{-1}}+\norm v_{g^{-1}}}.
\end{align*}
Using our bounds on $\norm{x_{1}-x}_{g},\norm{x_{t}-x}_{g}$ and $\norm v_{g^{-1}},\norm{v_{t}-v}_{g^{-1}},\norm{v-v_{\half}}_{g^{-1}}$,
we conclude that \\
$\max_{t\in[0,h]}\norm{g_{t}^{-1}v_{t}-g_{\mid}^{-1}v_{\mid}}_{g_{\mid}}\leq10^{4}h\Par{n+\sqrt{M_{1}}}$
and thus 
\[
\norm{T_{x}(v)-\bTx(v)}_{g}\leq10^{4}h^{2}\Par{n+\sqrt{M_{1}}}.
\]

\paragraph{Proof of 3.}

From the algorithm,
\begin{align*}
v_{h} & =v-\int_{0}^{h}\frac{\partial H}{\partial x}(x_{t},v_{t})dt,\\
\bv & =v_{\half}-\frac{h}{2}\frac{\del H}{\del x}\Par{x_{1},v_{\half}}\\
 & =v-\frac{h}{2}\frac{\del H}{\del x}\Par{x_{1},v_{\half}}-\frac{h}{2}\frac{\partial H(x,v)}{\partial x}.
\end{align*}
Thus,
\begin{align*}
\norm{v_{h}-\bv}_{g^{-1}} & =\norm{\int_{0}^{h}\half\Par{\frac{\del H}{\del x}(x_{t},v_{t})-\frac{\del H}{\del x}\Par{x_{1},v_{\half}}}dt+\int_{0}^{h}\half\Par{\frac{\del H}{\del x}(x_{t},v_{t})-\frac{\del H}{\del x}\Par{x,v}}dt}_{g^{-1}}\\
 & \leq\frac{h}{2}\underbrace{\max_{t\in[0,h]}\norm{\frac{\del H}{\del x}(x_{t},v_{t})-\frac{\del H}{\del x}\Par{x_{1},v_{\half}}}_{g^{-1}}}_{F}++\frac{h}{2}\underbrace{\max_{t\in[0,h]}\norm{\frac{\del H}{\del x}(x_{t},v_{t})-\frac{\del H}{\del x}\Par{x,v}}_{g^{-1}}}_{S}.
\end{align*}
For $\dv=\norm{v_{t}-v_{\half}}_{g^{-1}}$ and $\dx=\norm{x_{t}-x_{1}}_{g}$,
we use Proposition \ref{lem:computation_sc}-12 to show that 
\begin{align*}
F & \lesssim\max_{t\in[0,h]}\Par{\dv+\dx\norm{v_{\half}}_{g^{-1}}}\Par{\norm{v_{\half}}_{g^{-1}}+\norm{v_{t}}_{g^{-1}}}+M_{2}^{*}\dx.
\end{align*}
In a similar way, $S$ can be bounded as follows:
\[
S\lesssim h\Par{n+\sqrt{M_{1}}}^{3/2}+M_{2}^{*}\dx.
\]
Using our bounds on $\norm{x_{1}-x}_{g},\norm{x_{t}-x}_{g}$ and $\norm v_{g^{-1}},\norm{v_{t}-v}_{g^{-1}},\norm{v-v_{\half}}_{g^{-1}}$,
\[
F+S\lesssim h\Par{n+\sqrt{M_{1}}}^{3/2}+h\sqrt{n+\sqrt{M_{1}}}M_{2}^{*}.
\]
Substituting the bounds on $F$ and $S$, we can conclude that
\[
\norm{v_{h}-\bv}_{g^{-1}}\leq10^{10}\Par{\Par{n+\sqrt{M_{1}}}^{3/2}+\sqrt{n+\sqrt{M_{1}}}M_{2}^{*}}h^{2}.
\]

\subsubsection{Sensitivity}

We show that for some step size $h$ the generalized Leapfrog integrator
is sensitive at $(x,v)$ for $x\in\mcal_{\rho}$ and $v\in\vgood$.
\begin{lem}
\label{lem:stableLeap} For $x\in\mcal_{\rho}$ and $v\in\vgood$,
LM is sensitive at $(x,v)$ if step size $h$ satisfies $h^{2}\leq\min\Par{\frac{10^{-10}}{n^{2}\Par{n+\sqrt{M_{1}}}},\frac{10^{-5}}{\sqrt{n}\bar{R}_{1}}}$
and the step-size conditions in Proposition \ref{prop:dynamics}.
\end{lem}

\begin{proof}
Let $\bTx(v)=\bx=x_{1}$. We lower bound $\Abs{D\bTx(v)}$. It suffices
to look into the determinant of composition of first two steps in
Algorithm \ref{alg:LM}, since the third step only changes $v$. The
first two steps are 

\begin{align*}
v_{\half} & =v-\frac{h}{2}\frac{\del H}{\del x}(x,v_{\half}),\\
x_{1} & =x+\frac{h}{2}\Par{\frac{\del H}{\del v}(x,v_{\half})+\frac{\del H}{\del v}(x_{1},v_{\half})}.
\end{align*}
Differentiating the first equation with respect to $v$, we have 
\[
\frac{\del v_{\half}}{\del v}=I_{n}-\frac{h}{2}\frac{\del^{2}H}{\del x\del v}\Par{x,v_{\half}}\frac{\del v_{\half}}{\del v},
\]
and so
\begin{equation}
\Par{I+\frac{h}{2}\frac{\del^{2}H}{\del x\del v}\Par{x,v_{\half}}}\frac{\del v_{\half}}{\del v}=I_{n}.\label{eq:step1diff}
\end{equation}
Differentiating the second equation with respect to $v$, we obtain
\[
\frac{\del x_{1}}{\del v}=\frac{h}{2}\Par{\frac{\del^{2}H}{\del v^{2}}\Par{x,v_{\half}}\frac{\del v_{\half}}{\del v}+\frac{\del^{2}H}{\del x\del v}\Par{x_{1},v_{\half}}\frac{\del x_{1}}{\del v}+\frac{\del^{2}H}{\del v^{2}}\Par{x_{1},v_{\half}}\frac{\del v_{\half}}{\del v}}.
\]
Collecting all $\del x_{1}/\del v$ terms from this equation, for
$g=g(x)$ and $g_{1}=g(x_{1})$
\begin{align*}
\Par{I_{n}-\frac{h}{2}\frac{\del^{2}H}{\del x\del v}\Par{x_{1},v_{\half}}}\frac{\del x_{1}}{\del v} & =\frac{h}{2}\Par{\frac{\del^{2}H}{\del v^{2}}\Par{x,v_{\half}}+\frac{\del^{2}H}{\del v^{2}}\Par{x_{1},v_{\half}}}\frac{\del v_{\half}}{\del v}\\
 & =\frac{h}{2}\Par{g^{-1}+g_{1}^{-1}}\Par{I+\frac{h}{2}\frac{\del^{2}H}{\del x\del v}\Par{x,v_{\half}}}^{-1},
\end{align*}
where we used \eqref{eq:step1diff}. Hence,
\begin{align*}
\frac{\del x_{1}}{\del v} & =h\Par{I_{n}-\frac{h}{2}\frac{\del^{2}H}{\del x\del v}\Par{x_{1},v_{\half}}}^{-1}\Par{\frac{g^{-1}+g_{1}^{-1}}{2}}\Par{I_{n}+\frac{h}{2}\frac{\del^{2}H}{\del x\del v}\Par{x,v_{\half}}}^{-1}\\
 & =h\Par{I_{n}-\frac{h}{2}g_{1}^{-1}Dg_{1}\Brack{g_{1}^{-1}v_{\half}}}^{-1}\Par{\frac{g^{-1}+g_{1}^{-1}}{2}}\Par{I_{n}+\frac{h}{2}g^{-1}Dg\Brack{g^{-1}v_{\half}}}^{-1}\\
 & =hg_{1}^{\half}\Par{I_{n}-\frac{h}{2}g_{1}^{-\half}Dg_{1}\Brack{g_{1}^{-1}v_{\half}}g_{1}^{-\half}}^{-1}g_{1}^{-\half}\Par{\frac{g^{-1}+g_{1}^{-1}}{2}}g^{\half}\Par{I_{n}+\frac{h}{2}g^{-\half}Dg\Brack{g^{-1}v_{\half}}g^{-\half}}^{-1}g^{-\half}.
\end{align*}
Due to the concavity of log-determinant in the set of positive definite
matrices, we have 
\begin{align*}
\log\Abs{\frac{g^{-1}+g_{1}^{-1}}{2}} & \geq\half\Par{\log\Abs{g^{-1}}+\log\Abs{g_{1}^{-1}}}=\log\frac{1}{\sqrt{\Abs g\Abs{g_{1}}}},
\end{align*}
and thus 
\begin{align*}
\Abs{D\bTx(v)} & =\Abs{\frac{\del x_{1}}{\del v}}\\
 & \geq\frac{h^{n}}{\sqrt{\Abs g\Abs{g_{1}}}}\Abs{I_{n}-\frac{h}{2}g_{1}^{-\half}Dg_{1}\Brack{g_{1}^{-1}v_{\half}}g_{1}^{-\half}}^{-1}\Abs{I_{n}+\frac{h}{2}g^{-\half}Dg\Brack{g^{-1}v_{\half}}g^{-\half}}^{-1}.
\end{align*}

For $E\defeq\frac{h}{2}g^{-\half}Dg\Brack{g^{-1}v_{\half}}g^{-\half}$,
as bounded in  (\ref{eq:LMIMM-eq4}), we have that
\begin{align*}
\norm E_{2} & \leq\frac{h}{2}\norm{v_{\half}}_{g^{-1}},\\
\tr(E) & \lesssim hn\norm{v_{\half}}_{g^{-1}},\\
\norm E_{F} & \leq\frac{h\sqrt{n}}{2}\norm{v_{\half}}_{g^{-1}}.
\end{align*}
Due to $h^{2}\leq\frac{10^{-10}}{n^{2}\Par{n+\sqrt{M_{1}}}}$, it
follows from (\ref{eq:LM-ineq1}) that $\norm{v_{\half}}_{g^{-1}}\leq O\Par{\sqrt{n+\sqrt{M_{1}}}}$.
This condition also allows us to make all these three quantities smaller
than $10^{-5}$. By Lemma \ref{lem:boundLogDet}, 
\[
\Abs{I_{n}-\frac{h}{2}g_{1}^{-\half}Dg_{1}\Brack{g_{1}^{-1}v_{\half}}g_{1}^{-\half}}^{-1}\geq1-10^{-8}.
\]
Similarly, we obtain 
\[
\Abs{I_{n}+\frac{h}{2}g^{-\half}Dg\Brack{g^{-1}v_{\half}}g^{-\half}}^{-1}\geq1-10^{-8},
\]
and thus $\Abs{D\bTx(v)}\geq\frac{(1-10^{-6})h^{n}}{\sqrt{\Abs{g(x)}\Abs{g(\bx)}}}$.
Using the step-size conditions in Proposition \ref{prop:dynamics},
we use Proposition \ref{prop:stability} to show that $\bTx$ is sensitive
at $(x,v)$.
\end{proof}

\section{Convergence rate of RHMC in polytopes \label{sec:polyotopes}}

In this section, we present the mixing times of the ideal and discretized
RHMC for an exponential density in a polytope. We set $f(x)=\alpha^{\top}x$
for $\alpha\in\Rn$. For a full-rank matrix $A\in\R^{m\times n}$
and $b\in\R^{m}$, the polytope is represented by $\Brace{x\in\Rn:Ax\geq b}$,
equipped with the logarithmic barrier $\phi(x)=-\sum_{i=1}^{m}\log(a_{i}^{\top}x-b_{i}),$
where $a_{i}$ is the $i^{th}$ row of $A$ and $b_{i}$ is the $i^{th}$
entry of $b$. We can check by direct computation that the logarithmic
barriers are highly self-concordant. We view this polytope as the
Hessian manifold $\mcal$ induced by the local norm $g(x)=\hess\phi(x)$.
We denote a slack vector by $s_{x}=Ax-b\in\R^{m}$ and its diagonalization
by $S_{x}=\Diag(s_{x})\in\R^{m\times m}$. We also define $A_{x}=S_{x}^{-1}A$
and $s_{v}=A_{x}v$ for $v\in T_{x}\mcal$, where $T_{x}\mcal$ is
endowed with the local metric $g$. One can check by direct computation
that $\hess\phi(x)=A_{x}^{\top}A_{x}$.

In this setting, we can compute all the parameters we have defined,
obtaining the mixing time of RHMC discretized by a sensitive numerical integrator.

\subsection{Isoperimetry of convex set \label{subsec:isoCvx}}

An isoperimetry inequality is one of the two main ingredients for
bounding the mixing rate. We use the Riemannian version of some isoperimetry
inequality. To state it, we need another distance called \emph{Hilbert
distance} in addition to Riemannian distance $d_{\phi}$. 
\begin{defn}
For a convex body $\kcal$, the cross-ratio distance $d_{\kcal}(x,y)$
between $x$ and $y$ is 
\[
d_{\kcal}(x,y)=\frac{|x-y||p-q|}{|p-x||y-q|},
\]
where $p$ and $q$ are on the boundary of $\kcal$ such that $p,x,y,q$
are on the straight line $\overline{xy}$ and are in order. The Hilbert
distance $d_{H}$ between $x,y\in\kcal$ is 
\[
d_{H}(x,y)=\log(1+d_{\kcal}(x,y))=\log\Par{1+\frac{|x-y||p-q|}{|p-x||y-q|}}.
\]
For sets $X$ and $Y$, we define $d_{*}(X,Y)=\inf_{x\in X,y\in Y}d_{*}(x,y)$
for $*\in\{\kcal,H,\phi\}$.
\end{defn}

\begin{lem}[\cite{vempala2005geometric}, Theorem 4.4] \label{lem: iso_general}
Let $\pi$ be a log-concave distribution supported on a convex body
$\kcal$. Let $S_{1},S_{2},S_{3}$ be a partition of $\kcal$. Then,
\[
\pi(S_{3})\geq d_{\kcal}(S_{1},S_{2})\pi(S_{1})\pi(S_{2}).
\]
\end{lem}

The following lemma is a generalization of Theorem 26 in \cite{lee2017geodesic}
to a subset $\kcal'$.
\begin{lem}
\label{lem:iso_kkp} Let $\pi$ be a log-concave distribution supported
on a convex body $\kcal$, and $\phi$ a self-concordant barrier of
$\kcal$. Let $\kcal'$ be a convex subset of $\kcal$, and $S_{1},S_{2},S_{3}$
a partition of $\kcal'$. Then 
\[
\pi(S_{3})\pi(\kcal')\geq\frac{d_{\phi}(S_{1},S_{2})}{G}\pi(S_{1})\pi(S_{2}),
\]
where $G=\sup_{x,y\in\kcal}\frac{d_{\phi}(x,y)}{d_{H}(x,y)}$.
\end{lem}

\begin{proof}
Applying Lemma \ref{lem: iso_general} to the distribution $\pi_{\kcal'}$
defined by $\pi$ restricted to $\kcal'$, we have 
\[
\pi(S_{3})\pi(\kcal')\geq d_{\kcal'}(S_{1},S_{2})\pi(S_{1})\pi(S_{2}).
\]
Due to $\kcal'\subseteq\kcal$, one can check $d_{\kcal'}(S_{1},S_{2})\geq d_{\kcal}(S_{1},S_{2})$
by simple algebra. As $d_{\kcal}(x,y)\geq d_{H}(x,y)$, it follows
that 
\[
\pi(S_{3})\pi(\kcal')\geq\frac{d_{\phi}(S_{1},S_{2})}{\frac{d_{\phi}(S_{1},S_{2})}{d_{H}(S_{1},S_{2})}}\pi(S_{1})\pi(S_{2})\geq\frac{d_{\phi}(S_{1},S_{2})}{G}\pi(S_{1})\pi(S_{2}).
\]
\end{proof}
We now define the symmetric self-concordance parameter of the barrier
$\phi$.
\begin{defn}[\cite{laddha2020strong}] For a convex body $\kcal\subseteq\Rn$,
the symmetric self-concordance parameter $\bar{\nu}_{\phi}$ of $\kcal$
is the smallest number such that for any $x\in\kcal$ 
\[
D(x)\subseteq\kcal\cap(2x-\kcal)\subseteq\sqrt{\bar{\nu}_{\phi}}D(x),
\]
where $D(x)=\Brace{y\in\Rn:\norm{y-x}_{\hess\phi(x)}\leq1}$ is the
Dikin ellipsoid at $x$.
\end{defn}

In general, it is known that $\bar{\nu}_{\phi}=O(\nu_{\phi}^{2})$
for the self-concordance parameter $\nu_{\phi}$ (see Definition \ref{def:sc-barrier}),
but a tighter bound of $\bar{\nu}_{\phi}=O(\nu_{\phi})$ holds for
important barriers such as the logarithmic barrier and Lee-Sidford
barrier \cite{lee2014path}.
\begin{lem}[\cite{laddha2020strong}, Lemma 2.3] \label{lem:ddh_ratio} $d_{\phi}(x,y)\lesssim\sqrt{\bar{\nu}_{\phi}}d_{H}(x,y)$
for any $x,y\in\kcal$.
\end{lem}

Using Lemma \ref{lem:iso_kkp} and \ref{lem:ddh_ratio} together,
we have 
\[
\pi(S_{3})\pi(\kcal')\geq\frac{d_{\phi}(S_{1},S_{2})}{\sqrt{\bar{\nu}_{\phi}}}\pi(S_{1})\pi(S_{2}),
\]
 and it implies that the isoperimetry of $\kcal'$ is at least $1/\sqrt{\bar{\nu}_{\phi}}$.
As $\nu_{\phi}=O(m)$ for the logarithmic barrier, $\psi_{\kcal'}\geq1/\sqrt{m}$
for a convex subset $\kcal'$.

\subsection{Good region $\protect\mcal_{\rho}$}

Taking a proper good region $\mcal_{\rho}$ plays an important role
in establishing a condition-number independent mixing rate of RHMC
for an exponential density in a polytope. To this end, we set our
good region to 
\begin{align*}
\mcal_{\rho} & \defeq\Brace{x\in\mcal:\norm{\alpha}_{g(x)^{-1}}^{2}\leq10n^{2}\log^{2}\Par{\frac{1}{\rho}}}.
\end{align*}
To establish the isoperimetry of $\mcal_{\rho}$ following Section
\ref{subsec:isoCvx}, we check its convexity in the following lemma.
Note that the assumption in the lemma is satisfied by the logarithmic
barriers.
\begin{lem}
If the fourth directional derivative of $\phi$ is positive (i.e.,
$D^{4}\phi[a,a,b,b]\geq0$), then $\mcal_{\rho}$ is convex.
\end{lem}

\begin{proof}
Let $\Upsilon(x):=\alpha^{\top}g(x)^{-1}\alpha=\alpha^{\top}(\nabla^{2}\phi(x))^{-1}\alpha$.
It suffices to show that $\Upsilon(x)$ is convex. Note that 
\[
\frac{\partial\Upsilon(x)}{\partial x_{i}}=\alpha^{\top}g(x)^{-1}\frac{\partial g(x)}{\partial x_{i}}g(x)^{-1}\alpha,
\]
and thus its directional derivative in $h=(h_{1},...,h_{n})$ is 
\[
\nabla\Upsilon(x)\cdot h=\sum_{i}h_{i}\Par{s(x)^{\top}\frac{\partial g(x)}{\partial x_{i}}s(x)},
\]
where $s(x):=g(x)^{-1}\alpha$. Note that 
\begin{align*}
\frac{\partial}{\partial x_{j}}\Par{s(x)^{\top}\frac{\partial g(x)}{\partial x_{i}}s(x)} & =s(x)^{\top}\frac{\partial^{2}g(x)}{\partial x_{i}\partial x_{j}}s(x)+2s(x)^{\top}\frac{\partial g(x)}{\partial x_{i}}\Par{\frac{\partial s(x)}{\partial x_{j}}}\\
 & =s(x)^{\top}\frac{\partial^{2}g(x)}{\partial x_{i}\partial x_{j}}s(x)+2s(x)^{\top}\frac{\partial g(x)}{\partial x_{i}}g(x)^{-1}\frac{\partial g(x)}{\partial x_{j}}s(x).
\end{align*}
Therefore, 
\begin{align*}
D^{2}\Upsilon(x)[h,h] & =\sum_{i,j}h_{i}h_{j}s(x)^{\top}\frac{\partial^{2}g(x)}{\partial x_{i}\partial x_{j}}s(x)+2\sum_{i,j}\Par{\frac{\partial g(x)}{\partial x_{i}}s(x)h_{i}}^{\top}g(x)^{-1}\Par{\frac{\partial g(x)}{\partial x_{j}}s(x)h_{j}}\\
 & =D^{4}\phi[h,h,s(x),s(x)]+2\sum_{i,j}\Par{\frac{\partial g(x)}{\partial x_{i}}s(x)h_{i}}^{\top}g(x)^{-1}\Par{\frac{\partial g(x)}{\partial x_{j}}s(x)h_{j}}.
\end{align*}
The first term is non-negative due to the assumption, and the second
term is also non-negative since $g(x)^{-1}$ is also positive semi-definite.
\end{proof}
Next, we show that $\mcal_{\rho}$ takes up probability of at least
$1-\rho$ over the stationary distribution $\pi$, where $\frac{d\pi(x)}{dx}\propto\exp(-\alpha^{\top}x)$.
\begin{lem}
$\pi(\mcal_{\rho})\geq1-\rho$.
\end{lem}

\begin{proof}
Let $g=g(x)$. For $\norm{\alpha}_{g^{-1}}$, note that 
\begin{align*}
\norm{\alpha}_{g(x)^{-1}}= & \max_{\norm u_{g(x)}=1}\alpha^{\top}u\\
= & \alpha^{\top}x-\min_{\norm{y-x}_{g}=1}\alpha^{\top}y\leq\alpha^{\top}x-\min_{y\in\mcal}\alpha^{\top}y,
\end{align*}
where the first equality is due to duality of norms and the last inequality
follows from the well-known fact that the Dikin ellipsoid at $x$
is inside $\mcal$. 

By Lemma \ref{lem:kalaivempala} with $c=\alpha/\norm{\alpha}_{2}$
and $T=1/\norm{\alpha}_{2}$, we have 
\[
\E_{x\sim\pi^{*}}[\alpha^{\top}x]\leq n+\min_{y\in\mcal}\alpha^{\top}y.
\]
Then, by Lemma \ref{lem:sinhoentropic} we have $\E[(\alpha^{\top}x-\min_{y\in\mcal}\alpha^{\top}y)^{2}]-\E[\alpha^{\top}x-\min_{y\in\mcal}\alpha^{\top}y]^{2}\leq n$
so that $\E[(\alpha^{\top}x-\min_{y\in\mcal}\alpha^{\top}y)^{2}]\leq n+n^{2}.$
By Lemma \ref{lem:lovaszgeo}, we have 
\[
\Pr_{x\sim\pi}\Brack{\alpha^{\top}x-\min_{y\in\mcal}\alpha^{\top}y>2\Par{\log\frac{1}{\rho}+1}n}\leq\rho.
\]
\end{proof}

\subsection{Auxiliary function $\ell$ and smoothness parameters $R$}

In this region $\mcal_{\rho}$ and step size $h$, the parameters
$M_{1},M_{2}$ and $M_{1}^{*},M_{2}^{*}$ (see Definition \ref{def:Mparameters})
are computed by 
\begin{align*}
M_{1} & =\max\Par{n,\norm{\alpha}_{g(x)^{-1}}^{2}}\leq10n^{2}\log^{2}\Par{\frac{1}{\rho}},\\
M_{1}^{*} & \leq\norm{\alpha}_{g(x)^{-1}}^{2}\leq10n^{2}\log^{2}\Par{\frac{1}{\rho}},\\
M_{2} & ,M_{2}^{*}=0.
\end{align*}
We use the following auxiliary function $\ell$ proposed in \cite{lee2018convergence}
and symmetric auxiliary function $\bar{\ell}$:
\begin{align*}
\ell(\gamma)=\max_{t\in[0,h]}\bigg(\frac{\norm{s_{\gamma'(t)}}_{2}}{\sqrt{n}+2M_{1}^{1/4}}+\frac{\norm{s_{\gamma'(t)}}_{4}}{2M_{1}^{1/4}}+ & \frac{\norm{s_{\gamma'(t)}}_{\infty}}{\sqrt{\log n}+2h\sqrt{M_{1}}}\bigg)+\frac{\norm{s_{\gamma'(0)}}_{2}}{\sqrt{n}}+\frac{\norm{s_{\gamma'(0)}}_{4}}{n^{1/4}}+\frac{\norm{s_{\gamma'(0)}}_{\infty}}{\sqrt{\log n}},\\
\bar{\ell}(\gamma)=\max_{t\in[0,h]}\bigg(\frac{\norm{s_{\gamma'(t)}}_{2}}{\sqrt{n}+2M_{1}^{1/4}}+\frac{\norm{s_{\gamma'(t)}}_{4}}{2M_{1}^{1/4}}+ & \frac{\norm{s_{\gamma'(t)}}_{\infty}}{\sqrt{\log n}+2h\sqrt{M_{1}}}\bigg).
\end{align*}
This measures how fast a Hamiltonian trajectory approaches the facets
of a polytope in the local norm. 

We make simple observations based on the self-concordance of $g$.
\begin{prop}
\label{prop:simpleObs} Let $\overline{\mcal_{\rho}}\defeq\Brace{x\in\mcal:\norm{\alpha}_{g(x)^{-1}}^{2}\leq20n^{2}\log^{2}\Par{\frac{1}{\rho}}}$
and $\gamma$ be any Hamiltonian curve $\gamma$ starting at $x\in\mcal_{\rho}$
with $v\in\vgood$. If step size $h$ satisfies $h^{2}\leq10^{-11}\min\Par{\frac{1}{n\log\frac{1}{\rho}},\frac{1}{C_{x}(x,v)}}$,
then $x_{h}$ and $\bx_{h}$ are contained in $\overline{\mcal_{\rho}}$.
\end{prop}

\begin{proof}
Due to the assumption on the step size, we can use Proposition \ref{prop:dynamics}-1,
obtaining $\norm{x-\gamma(t)}_{x}\leq O\Par{t\sqrt{n+\sqrt{M_{1}}}}=O\Par{t\sqrt{n\log\frac{1}{\rho}}}<\frac{1}{4}$.
Also, $\norm{x-\bx_{h}}_{g}\leq\norm{x-\gamma(h)}_{g}+\norm{\gamma(h)-\bx_{h}}_{g}\leq\frac{1}{4}+h^{2}C_{x}(x,v)\leq\frac{1}{3}$.
The claim follows from the self-concordance of $g(x)$, due to $\norm{\alpha}_{g(\gamma(h))^{-1}}^{2}\leq\Par{1+\norm{x-\gamma(h)}_{x}}\norm{\alpha}_{g(x)^{-1}}^{2}\leq20n^{2}\log^{2}\frac{1}{\rho}$
and $\norm{\alpha}_{g(\bx_{h})^{-1}}^{2}\leq\frac{4}{3}\norm{\alpha}_{g(x)^{-1}}^{2}\leq20n^{2}\log^{2}\frac{1}{\rho}$.
\end{proof}
As in \cite{lee2018convergence}, we can represent the parameters
$\ell_{0},\ell_{1}$ and the smoothness parameters $R_{1},R_{2},R_{3}$
in terms of $M_{1}$. The original proof in \cite{lee2018convergence}
relies on the fact that $\norm{\grad f(\gamma(t))}_{g(\gamma(t))^{-1}}^{2}\leq M_{1}$
for any time $t\in[0,h]$ and any regular Hamiltonian curves. In our
setting, $\norm{\grad f(\gamma(t))}_{g(\gamma(t))^{-1}}^{2}\leq2M_{1}$
for any time $t\in[0,h]$ if $h^{2}\leq\frac{10^{-11}}{n\log\frac{1}{\rho}}$,
we can simply reproduce Lemma 54\textasciitilde 59 by replacing $M_{1}$
by $2M_{1}$.
\begin{lem}
\label{lem:l0} Consider a Hamiltonian trajectory $\gamma$ starting
at $x\in\mcal_{\rho}$ with an initial (normalized) velocity randomly
chosen from $\ncal(0,g(x)^{-1})$, with step size $h$ satisfying
$h^{2}n\log\frac{1}{\rho}\leq10^{-11}$. For $n$ large enough, if
$s$ satisfies $sh=O\Par n$, then
\[
\P_{\gamma}\Par{\ell(\gamma)\geq128}\leq\frac{1}{100}\min\Par{1,\frac{\ell_{0}}{sh}}.
\]
\end{lem}

As we shortly see in Lemma \ref{lem:l1}, we have $\ell_{1}h=O\Par{h^{2}M_{1}^{1/4}}=O\Par{\frac{1}{\sqrt{n\log\frac{1}{\rho}}}}$,
and thus $\ell_{1}$ can be used in place of $s$ in this lemma.
\begin{lem}
\label{lem:R1} Let $\gamma$ be a Hamiltonian curve starting at $x\in\mcal_{\rho}$
with $\ell(\gamma)\leq\ell_{0}\leq256$ and step size $h$ satisfying
$h^{2}n\log\frac{1}{\rho}\leq10^{-11}$. Then 
\[
\sup_{t\in[0,h]}\norm{\Phi(\gamma,t)}_{F,\gamma(t)}\leq R_{1}
\]
 with $R_{1}=O\Par{\sqrt{M_{1}}}$.
\end{lem}

\begin{lem}
\label{lem:R2} Let $\gamma$ be a Hamiltonian curve starting at $x\in\mcal_{\rho}$
with $\ell(\gamma)\leq\ell_{0}\leq256$ and step size $h$ satisfying
$h^{2}n\log\frac{1}{\rho}\leq10^{-11}$. For any $t\in[0,h]$, any
curve $c(s)$ starting from $\gamma(t)$ and any velocity field $v(c(s))$
on $c(s)$ with $v(c(0))=v(\gamma(t))=\gamma'(t)$, we have that 
\[
\Abs{\frac{d}{ds}\tr\Phi(v(c(s))\bigg|_{s=0}}\leq R_{2}\Par{\norm{\frac{dc}{ds}\bigg|_{s=0}}_{\gamma(t)}+h\norm{D_{s}v\big|_{s=0}}_{\gamma(t)}}
\]
 with $R_{2}=O\Par{\sqrt{nM_{1}}+\sqrt{n}M_{1}h^{2}+\frac{M_{1}^{1/4}}{h}+\frac{\sqrt{n\log n}}{h}}$.
\end{lem}

\begin{lem}
\label{lem:R3} Let $\gamma$ be a Hamiltonian curve starting at $x\in\mcal_{\rho}$
with $\ell(\gamma)\leq\ell_{0}\leq256$ and step size $h$ satisfying
$h^{2}n\log\frac{1}{\rho}\leq10^{-11}$. Let $\zeta(t)$ be the parallel
transport of the vector $\gamma'(0)$ to $\gamma(t)$. Then 
\[
\sup_{t\in[0,h]}\norm{\Phi(\gamma,t)\zeta(t)}_{\gamma(t)}\leq R_{3}
\]
 with $R_{3}=O\Par{\sqrt{M_{1}\log n}+M_{1}^{3/4}n^{1/4}h}$.
\end{lem}

\begin{lem}
\label{lem:l1} Let $\gamma_{s}$ be a Hamiltonian variation starting
at $x\in\mcal_{\rho}$ with $\ell(\gamma_{s})\leq\ell_{0}\leq256$
and step size $h$ satisfying $h^{2}n\log\frac{1}{\rho}\leq10^{-11}$.
Then 
\[
\Abs{\frac{d}{ds}\ell(\gamma_{s})}\leq O\Par{M_{1}^{1/4}h+\frac{1}{h\sqrt{\log n}}}\Par{\norm{\frac{d}{ds}\gamma_{s}(0)}_{\gamma_{s}(0)}+h\norm{D_{s}\gamma_{s}'(0)}_{\gamma_{s}(0)}},
\]
and thus $\ell_{1}=O\Par{M_{1}^{1/4}h+\frac{1}{h\sqrt{\log n}}}$.
\end{lem}

For $\bar{\ell}_{0},\bar{\ell}_{1},\bar{R}_{1}$, we can repeat the
arguments so far for regular Hamiltonian curves starting from $\overline{\mcal_{\rho}}$,
in which $\norm{\alpha}_{g(\gamma(t))^{-1}}^{2}$ is within a constant
factor of $M_{1}$. Therefore, these three parameters also have the
same bounds in Lemma \ref{lem:l0}, \ref{lem:R1} and \ref{lem:l1}
up to a multiplicative constant factor.

\subsection{Convergence rate of RHMC with numerical integrators}

Now that we estimated all the parameters, we can put them together
and state the mixing time of RHMC discretized by a sensitive  numerical integrator.

\thmdiscPoly*
\begin{proof}
We first note that $\vgood=\Brace{v\in\Rn:\bar{\ell}\Par{\ham_{x,t}\Par{g(x)^{-1}v}}\leq128}$,
the measure of which is at least $0.99$ by the definition of $\bar{\ell}_{0}$.
We check the conditions on the step size in Theorem \ref{thm:discGen}.
Let $\rho=\frac{\veps}{2\Lambda}$. We first bound $M_{1},M_{1}^{*}$
by $20n^{2}\log^{2}\frac{1}{\rho}$ and set $M_{2}$ to $0$. Substituting
these to Lemma \ref{lem:l1}, \ref{lem:R1}, \ref{lem:R2} and \ref{lem:R3},
we have 
\begin{align*}
\ell_{1} & \lesssim h\sqrt{n\log\frac{1}{\rho}}+\frac{1}{h},\\
R_{1} & \lesssim n\log\frac{1}{\rho},\\
R_{2} & \lesssim n^{3/2}\log\frac{1}{\rho}+h^{2}n^{5/2}\log^{2}\frac{1}{\rho}+\frac{\sqrt{n\log\frac{1}{\rho}}}{h}+\frac{\sqrt{n\log n}}{h},\\
R_{3} & \lesssim n\sqrt{\log n}\log\frac{1}{\rho}+hn^{7/4}\log^{3/2}\frac{1}{\rho}.
\end{align*}
Due to $h\leq\frac{10^{-20}}{n^{7/12}\log^{1/2}\frac{1}{\rho}}$,
direct computation leads to $h^{2}\max\Par{R_{1},\bar{R}_{1}},h^{5}R_{1}^{2}\ell_{1}/\ell_{0},h^{3}R_{2}+h^{2}R_{3}\lesssim1$
and $h\lesssim\min\Par{1,\frac{\ell_{0}}{\ell_{1}}}$. The rest of
conditions on the step size, $hC_{x}(x,v)\leq\frac{10^{-20}}{\sqrt{n}},\,h^{2}C_{x}(x,v)\leq\frac{10^{-10}}{n\log\frac{1}{\rho}}\text{ and }h^{2}C_{v}(x,v)\leq\frac{10^{-10}}{\sqrt{n\log\frac{1}{\rho}}}$,
guarantee that 
\begin{equation*}
hC_{x}(x,v)\leq\frac{10^{-20}}{\sqrt{n}},\,h^{2}C_{x}(x,v)\leq10^{-10}\min\Par{1,\frac{\bar{\ell}_{0}}{\bar{\ell}_{1}},\frac{1}{n+\sqrt{M_{1}}+\sqrt{M_{1}^{*}}}},\ h^{2}C_{v}(x,v)\leq\frac{10^{-10}}{\sqrt{n+\sqrt{M_{1}}}}.
\end{equation*}
As the isoperimetry is lower bounded by $\frac{1}{\sqrt{m}}$, Theorem~\ref{thm:discGen} results in the mixing time of $T=O\Par{mh^{-2}\log\frac{\Lambda}{\veps}}$
that ensures $\dtv(\pi_{T},\pi)\leq\veps$.
\end{proof}
By setting $C_{x},C_{v}$ to $0$, we can obtain the mixing time of
the ideal RHMC for exponential densities in polytopes.

\coridealPoly*

\subsubsection{Implicit midpoint method}

In the polytope setting, we can explicitly compute $C_{x}(x,v)$ and
$C_{v}(x,v)$ of IMM in terms of $n$ and $\rho$.
\begin{lem}
\label{lem:leapIMMsecondfinal} For $x\in\mcal_{\rho}$ and $v\in\vgood$,
let $h$ be step size of IMM with $h^{2}n\log\frac{1}{\rho}\leq10^{-11}$.
Then 
\[
C_{x}(x,v)=O\Par{n\log\frac{1}{\rho}},\ C_{v}(x,v)=O\Par{n^{3/2}\log^{3/2}\frac{1}{\rho}}.
\]
\end{lem}

\begin{proof}
By Lemma \ref{lem:LeapIMM-second}-2, it follows that 
\begin{align*}
C_{x}(x,v) & =O\Par{n+\sqrt{M_{1}}}\lesssim n+n\log\frac{1}{\rho}=O\Par{n\log\frac{1}{\rho}}.
\end{align*}
For $C_{v}(x,v)$, we first note that $M_{2}^{*}=0$ due to $\grad^{2}f(x)=0$.
Thus by Lemma \ref{lem:LeapIMM-second}-3, we have 
\begin{align*}
C_{v}(x,v) & \lesssim\Par{n+\sqrt{M_{1}}}^{3/2}=O\Par{n^{3/2}\log^{3/2}\frac{1}{\rho}}.
\end{align*}
\end{proof}
We can also specify a sufficient condition on the step size for the
sensitivity of IMM in the polytope setting.
\begin{lem}
\label{lem:leapIMMstabilityfinal} For $x\in\mcal_{\rho}$, $v\in\vgood$
and step size $h$ with $h^{2}n^{2}\log\frac{1}{\rho}\leq10^{-10}$,
IMM is sensitive at $(x,v)$.
\end{lem}

\begin{proof}
Note that $\log\det g(x)$ is convex in $\mcal$, since the volumetric
barrier defined by $\log\det\hess\phi(x)$ is convex in $x$ (Lemma
1\textasciitilde 3 in \cite{vaidya1996new}). Thus, the claim follows
from Lemma \ref{lem:IMMstable}.
\end{proof}
Substituting the estimates of $C_{x}(x,v)$ and $C_{v}(x,v)$ as well
as the sufficient condition for the sensitivity to Theorem \ref{thm:discPoly},
we prove that the mixing rate of RHMC discretized by IMM for an exponential
density in a polytope is independent of the condition number and $\norm{\alpha}_{2}$.

\corimmPoly*
\begin{proof}
We can check that the step size $h=O\Par{\frac{1}{n^{3/2}\log\frac{\Lambda}{\veps}}}$
satisfies all the conditions in Theorem \ref{thm:discPoly}. Hence,
it suffices to choose $T=O\Par{mn^{3}\log^{3}\frac{\Lambda}{\veps}}$
to obtain $\dtv(\pi_{T},\pi)\leq\veps$.
\end{proof}

\subsubsection{Generalized Leapfrog method}

We now compute the mixing rate of RHMC discretized by LM. For LM,
we have the same results on $C_{x}(x,v)$ and $C_{v}(x,v)$ as IMM.
\begin{lem}
\label{lem:leapsecondfinal}For $x\in\mcal_{\rho}$ and $v\in\vgood$,
let $h$ be step size of LM with $h^{2}n\log\frac{1}{\rho}\leq10^{-10}$.
Then 
\[
C_{x}(x,v)=O\Par{n\log\frac{1}{\rho}},\ C_{v}(x,v)=O\Par{n^{3/2}\log^{3/2}\frac{1}{\rho}}.
\]
\end{lem}

For the sensitivity, LM requires a slightly stronger condition on step
size compared to IMM, which follows from Lemma \ref{lem:stableLeap}.
\begin{lem}
\label{lem:leapstabilityfinal} For $x\in\mcal_{\rho}$, $v\in\vgood$
and step size $h$ with $h^{2}n^{3}\log\frac{1}{\rho}\leq10^{-20}$,
LM is sensitive at $(x,v)$.
\end{lem}

We prove that the mixing rate of RHMC discretized by LM for an exponential
density in a polytope with $m$ constraints is also independent of
the condition number.

\corleapPoly*
\begin{proof}
For step size $h=O\Par{\frac{1}{n^{3/2}\log\frac{\Lambda}{\veps}}}$,
LM is sensitive in $\mcal_{\rho}\times V_{1}^{c}$ by Lemma \ref{lem:stableLeap},
and we can use the estimates of $C_{x}$ and $C_{v}$ proven in Lemma
\ref{lem:leapsecondfinal}. Thus, this step size satisfies all the
conditions in Theorem \ref{thm:discPoly}. Hence, it suffices to choose
$T=O\Par{mn^{3}\log^{3}\frac{\Lambda}{\veps}}$ to obtain $\dtv(\pi_{T},\pi)\leq\veps$.
\end{proof}

\paragraph{Acknowledgement. }

This work was supported in part by NSF awards CCF-1909756, CCF-2007443
and CCF-2134105.

\bibliographystyle{alpha}
\bibliography{main}

\appendix
\section{Definitions}
\begin{defn}[Self-concordant barrier] \label{def:sc-barrier} 
A self-concordant barrier $\phi:K\subset\Rn\to\R$ is a function such that $\phi(x)\to\infty$
as $x\to\del K$ and that $\Abs{Df^{3}(x)[h,h,h]}\leq2\Par{D^{2}f(x)[h,h]}^{3/2}$
for all $x\in K$ and $h\in\Rn$. If $\Abs{Df^{4}(x)[h,h,h,h]}\leq6\Par{D^{2}f(x)[h,h]}^{2}$
is also satisfied for all $h$, then $\phi$ is called a highly self-concordant
barrier.
\end{defn}

\begin{defn}[Self-concordance parameter] 
For a self-concordant function $\phi$, the self-concordance parameter of $\phi$ is the smallest non-negative real number $\nu_{\phi}$ such that 
\[
|D\phi(x)[h]|^{2}\leq\nu_{\phi}D^{2}\phi(x)[h,h],
\]
where $Df(x)[h]$ is the directional derivative of $f$ along direction
$h$ and $D^{2}f(x)[h_{1},h_{2}]$ is the second-order directional
derivative of $f$ along directions $h_{1}$ and $h_{2}$.
\end{defn}

\begin{defn}[Riemannian length and distance] \label{def:RiemannDistance} 
Let $\phi:\R^{n}\rightarrow\R$ be a self-concordant function. For all
$x\in\R^{d}$, we define the local norm induced by $\nabla^{2}\phi(x)$
by 
\[
\norm h_{\nabla^{2}\phi(x)}=\sqrt{h^{\top}\nabla^{2}\phi(x)h}.
\]
For any smooth curve $c:[0,1]\rightarrow\R^{n}$, we define the length
of the curve as 
\[
L_{\phi}(c)=\int_{0}^{1}\norm{\frac{d}{dt}c(t)}_{\nabla^{2}\phi(c(t))}dt.
\]
For any $x,y\in\R^{d}$, we define the distance $d_{\phi}(x,y)$ to
be the infimum of the lengths of all piecewise smooth curves with
$c(0)=x$ and $c(1)=y$.
\end{defn}

\begin{defn}[Total variation distance] \label{def:tvDistance} For probability
distributions $P$ and $Q$ supported on $K$, the total variation
distance (TV distance) is defined by 
\[
\dtv(P,Q)=\sup_{A\subset K}\Par{P(A)-Q(A)}.
\]
 
\end{defn}

\section{Lemmas}
\begin{lem}
\label{lem:matrixSeries} For $n\in\mathbb{N}$ and matrix $X\in\R^{2n\times2n}$
of the form 
\[
X=\left[\begin{array}{cc}
C & I_{n}\\
-C^{2}+R & -C
\end{array}\right]
\]
with a symmetric matrix $C\in\R^{n\times n}$ and matrix $R\in\R^{n\times n}$,
we have 
\begin{align*}
X^{2n} & =\left[\begin{array}{cc}
R^{n} & 0\\
R^{n}C-CR^{n} & R^{n}
\end{array}\right],\\
X^{2n+1} & =\left[\begin{array}{cc}
R^{n}C & R^{n}\\
R^{n+1}-CR^{n}C & -CR^{n}
\end{array}\right].
\end{align*}
\end{lem}

The claim immediately follows from induction.

\begin{lem}[\cite{lee2018convergence}, Lemma 7] \label{lem:HamODE} In the Euclidean coordinate,
the Hamiltonian equations in (\ref{eq:hmc_intro}) can be represented
via the second-order ODE as follows:
\begin{align*}
D_{t}\frac{dx}{dt} & =\mu(x),\\
\frac{dx}{dt}(0) & \sim\ncal(0,g(x)^{-1}),
\end{align*}
where $D_{t}$ is the covariant derivative along the Hamiltonian trajectory
$x(t)$ and $\mu(x)\defeq-g(x)^{-1}\grad f(x)-\half g(x)^{-1}\tr\Brack{g(x)^{-1}Dg(x)}$
.
\end{lem}

\begin{lem}[\cite{lee2018convergence}, Lemma 64] \label{lem:boundLogDet}
For matrix $E\in\R^{n\times n}$ with $\norm E_{2}<\frac{1}{4}$,
we have 
\[
\Abs{\log\det\Par{I+E}-\tr E}\leq\norm E_{F}^{2}.
\]
\end{lem}

\begin{lem}[\cite{kalai2006simulated}, Lemma 4.1] \label{lem:kalaivempala}
For a unit vector $c\in\R^{n}$, constant $T$ and convex set $K\subset\Rn$, we have 
\[
\E_{x\sim\pi}\Brack{c^{\top}x}\leq nT+\min_{x\in K}c^{\top}x,
\]
where $\pi$ is a probability density proportional to $e^{-\frac{c^{\top}x}{T}}$.
\end{lem}

\begin{lem}[\cite{nguyen2013dimensional}, Corollary 6] \label{lem:sinhoentropic}
Let $\pi$ be a log-concave density proportional to $\exp(-V)$ on
$\R^{n}$. Then,
\[
\text{Var}_{x\sim\pi}\Par{V(x)}\leq n.
\]
\end{lem}

\begin{lem} [\cite{lovasz2007geometry}, Lemma 5.17] \label{lem:lovaszgeo}Let
$X\in\Rn$ be randomly chosen from a log-concave distribution. Then
for any $R>1$,
\[
\P\Par{\Abs X>R\sqrt{\E X^{2}}}<e^{-R+1}.
\]
\end{lem}

\begin{lem}[\cite{nesterov2002riemannian}, Lemma 3.1] \label{lem:dist_const}
Suppose $\phi:\R^{n}\rightarrow\R$ is self-concordant and $\kcal\subset\Rn$
is convex. For any $x,y\in\kcal$, i
\begin{itemize}
\item If $d_{\phi}(x,y)\leq\delta-\delta^{2}<1$ for some $0<\delta<1$,
then $\norm{y-x}_{\nabla^{2}\phi(x)}\leq\delta$. 
\item If $\delta=\norm{x-y}_{\nabla^{2}\phi(x)}<1$, then $\delta-\frac{1}{2}\delta^{2}\leq d_{\phi}(x,y)\leq-\log(1-\delta)$.
\end{itemize}
\end{lem}

\end{document}